\documentclass[11pt]{article}

\usepackage[letterpaper, margin=1.0in] {geometry}
\usepackage{indentfirst}
\usepackage{amsfonts}
\usepackage{amsmath}
\usepackage{amsthm}
\usepackage{color}
\usepackage{eufrak}
\usepackage{verbatim}
\usepackage[makeroom]{cancel}

\usepackage{amssymb}
\usepackage{amsthm}
\usepackage{mathrsfs}
\usepackage{epsfig}
\usepackage{makeidx}
\usepackage{graphicx}
\usepackage{indentfirst}
\usepackage[numbers]{natbib}
\usepackage{color}
\usepackage[backref=page]{hyperref}
\usepackage{graphicx}
\usepackage{graphicx}
\usepackage{hyperref}
\usepackage{graphicx}
\usepackage{xparse}
\usepackage{float}
\usepackage[section]{placeins}

\definecolor{corlinks}{RGB}{0,0,150}
\definecolor{cormenu}{RGB}{0,0,150}
\definecolor{corurl}{RGB}{0,0,150}

\hypersetup{
colorlinks=true,
urlcolor=corlinks,
linkcolor=corlinks,
menucolor=cormenu,
citecolor=corlinks,
pdfborder= 0 0 0
}

\newtheorem{theorem}{Theorem}
\newtheorem{lemma}{Lemma}
\newtheorem{corollary}{Corollary}
\newtheorem{definition}{Definition}
\newtheorem{proposition}{Proposition}
\newtheorem{remark}{Remark}

\newcommand{\eqdef}{\stackrel{\rm def}{=}}

\DeclareMathOperator{\E}{\mathbb{E}}

\def\colorful{1}

\ifnum\colorful=1

\fi
\ifnum\colorful=0

\fi

\NewDocumentCommand{\diamondinclusion}{m >{\SplitArgument{1}{\\}}m m}
 {%
  \dodiamondinclusion{#1}#2{#3}%
 }
\NewDocumentCommand{\dodiamondinclusion}{mmmm}
 {%
  \begingroup
  \setlength{\arraycolsep}{0pt}%
  \renewcommand{\arraystretch}{-2}%
  \begin{matrix}
  && #2 \\
  & \rsubseteq{45} && \rsubseteq{-45} \\
  #1 &&&& #4 \\
  & \rsubseteq{-45} && \rsubseteq{45} \\
  && #3
  \end{matrix}%
  \endgroup
 }
\NewDocumentCommand{\rsubseteq}{m}
 {%
  \rotatebox[origin=c]{#1}{$\subseteq$}%
 }

\begin{document}

\title{Conspiracies between Learning Algorithms,\\Circuit Lower Bounds and Pseudorandomness\vspace{0.4cm}}

\author{Igor C. Oliveira\\
  \small{Charles University in Prague}\\
  \and Rahul Santhanam\\
  \small{University of Oxford}\vspace{0.3cm}
}


\maketitle

\vspace{-0.5cm}

\begin{abstract}

We prove several results giving new and stronger connections between learning theory, circuit complexity and pseudorandomness. Let $\mathfrak{C}$ be any typical class of Boolean circuits, and $\mathfrak{C}[s(n)]$ denote $n$-variable $\mathfrak{C}$-circuits of size $\leq s(n)$. We show:\\

\vspace{0.3cm}

\noindent \textbf{Learning Speedups.}  If $\mathfrak{C}[\mathsf{poly}(n)]$ admits a randomized weak learning algorithm under the uniform distribution with membership queries that runs in time $2^n/n^{\omega(1)}$, then for every $k\geq 1$ and $\varepsilon>0$ the class $\mathfrak{C}[n^k]$ can be learned to high accuracy in time $O(2^{n^\varepsilon})$. There is $\varepsilon > 0$ such that $\mathfrak{C}[2^{n^{\varepsilon}}]$ can be learned in time $2^n/n^{\omega(1)}$ if and only if $\mathfrak{C}[\mathsf{poly}(n)]$ can be learned in time $2^{(\log n)^{O(1)}}$.\\

\noindent \textbf{Equivalences between Learning Models.} We use learning speedups to obtain equivalences between various randomized learning and compression models, including sub-exponential time learning with membership queries, sub-exponential time learning with membership and equivalence queries, probabilistic function compression and probabilistic average-case function compression.\\

\noindent \textbf{A Dichotomy between Learnability and Pseudorandomness.} In the non-uniform setting, there is non-trivial learning for $\mathfrak{C}[\mathsf{poly}(n)]$ if and only if there are no exponentially secure pseudorandom functions computable in $\mathfrak{C}[\mathsf{poly}(n)]$.\\

\noindent \textbf{Lower Bounds from Nontrivial Learning.} If for each $k \geq 1$, (depth-$d$)-$\mathfrak{C}[n^k]$ admits a randomized weak learning algorithm with membership queries under the uniform distribution that runs in time $2^n/n^{\omega(1)}$, then for each $k\geq 1$, $\mathsf{BPE} \nsubseteq$ (depth-$d$)-$\mathfrak{C}[n^k]$. If for some $\varepsilon > 0$ there are $\mathsf{P}$-natural proofs useful against $\mathfrak{C}[2^{n^{\varepsilon}}]$, then $\mathsf{ZPEXP} \nsubseteq \mathfrak{C}[\mathsf{poly}(n)]$.\\

\noindent \textbf{Karp-Lipton Theorems for Probabilistic Classes.} If there is a $k > 0$ such that $\mathsf{BPE} \subseteq \mathtt{i.o.}\mathsf{Circuit}[n^k]$, then $\mathsf{BPEXP} \subseteq \mathtt{i.o.}\mathsf{EXP}/O(\mathsf{log}\,n)$. If $\mathsf{ZPEXP} \subseteq \mathtt{i.o.}\mathsf{Circuit}[2^{n/3}]$, then $\mathsf{ZPEXP} \subseteq \mathtt{i.o.}\mathsf{ESUBEXP}$.\\

\noindent \textbf{Hardness Results for $\mathsf{MCSP}$.} All functions in non-uniform $\mathsf{NC}^1$ reduce to the Minimum Circuit Size Problem via truth-table reductions computable by $\mathsf{TC}^0$ circuits. In particular, if $\mathsf{MCSP} \in \mathsf{TC}^0$ then $\mathsf{NC}^1 = \mathsf{TC}^0$.~\\
\end{abstract}

\newpage

\tableofcontents

\newpage

\section{Introduction}
\label{s:introduction}

Which classes of functions can be efficiently learned? Answering this question has been a major research direction in computational learning theory since the seminal work of Valiant \cite{Valiant84} formalizing efficient learnability.

For concreteness, consider the model of learning with membership queries under the uniform distribution. In this model, the learner is given oracle access to a target Boolean function and aims to produce, with high probability, a hypothesis that approximates the target function well on the uniform distribution. Say that a circuit class $\mathfrak{C}$ is {\it learnable} in time $T$ if there is a learner running in time $T$ such that for each function $f \in \mathfrak{C}$, when given oracle access to $f$ the learner outputs the description of a Boolean function $h$ approximating $f$ well under the uniform distribution. The hypothesis $h$ is not required to be from the same class $\mathfrak{C}$ of functions. (This and other learning models that appear in our work are defined in Section \ref{s:preliminaries}.) 

Various positive and conditional negative results are known for natural circuit classes in this model, and here we highlight only a few. Polynomial-time algorithms are known for polynomial-size DNF formulas \citep{DBLP:journals/jcss/Jackson97}. Quasi-polynomial time algorithms are known for polynomial-size constant-depth circuits with AND, OR and NOT gates \citep{DBLP:journals/jacm/LinialMN93} (i.e., $\mathsf{AC}^0$ circuits), and in a recent breakthrough \citep{CIKK16}, for polynomial-size constant-depth circuits which in addition contain MOD[$p$] gates, where $p$ is a fixed prime ($\mathsf{AC}^0[p]$ circuits). In terms of hardness, it is known that under certain cryptographic assumptions, the class of polynomial-size constant-depth circuits with threshold gates ($\mathsf{TC}^0$ circuits) is not learnable in sub-exponential time \citep{DBLP:journals/jacm/NaorR04}. (We refer to Section \ref{s:preliminaries} for a review of the inclusions between standard circuit classes.)

However, even under strong hardness assumptions, it is still unclear how powerful a circuit class needs to be before learning becomes utterly infeasible. For instance, whether non-trivial learning algorithms exist for classes beyond $\mathsf{AC}^0[p]$ remains a major open problem. 

Inspired by \citep{CIKK16}, we show that a general and surprising {\it speedup} phenomenon holds unconditionally for learnability of strong enough circuit classes around the border of currently known learning algorithms. Say that a class is {\it non-trivially learnable} if it is learnable in time $\leq 2^n/n^{w(1)}$, where $n$ is the number of inputs to a circuit in the class, and furthermore the learner is only required to output a hypothesis that is an approximation for the unknown function with inverse polynomial advantage. We show that for ``typical'' circuit classes such as constant-depth circuits with Mod[$m$] gates where $m$ is an arbitrary but fixed composite ($\mathsf{ACC}^0$ circuits), constant-depth threshold circuits, formulas and general Boolean circuits, non-trivial learnability in fact implies high-accuracy learnability in time $2^{n^{o(1)}}$, i.e., in sub-exponential time.

\begin{lemma} [Speedup Lemma, Informal Version]
\label{l:speedup}
Let $\mathfrak{C}$ be a typical circuit class. Polynomial-size circuits from $\mathfrak{C}$ are non-trivially learnable if and only if polynomial-size circuits from $\mathfrak{C}$ are \emph{(}strongly\emph{)} learnable in sub-exponential time. Subexponential-size circuits from $\mathfrak{C}$ are non-trivially learnable if and only if polynomial-size circuits from $\mathfrak{C}$ are \emph{(}strongly\emph{)} learnable in quasi-polynomial time.
\end{lemma}

Note that the class of {\it all} \,Boolean functions is learnable in time $\leq 2^n/n^{\Omega(1)}$ with $\geq 1/n$ advantage simply by querying the function oracle on $2^n/n^{O(1)}$ inputs, and outputting the best constant in $\{0,1\}$ for the remaining (unqueried) positions of the truth-table. Our notion of non-trivial learning corresponds to merely beating this trivial brute-force algorithm -- this is sufficient to obtain much more dramatic speedups for learnability of typical circuit classes.

In order to provide more intuition for this result, we compare the learning scenario to another widely investigated algorithmic framework. Consider the problem of checking if a circuit from a fixed circuit class is \emph{satisfiable}, a natural generalization of the $\mathsf{CNF}$-$\mathsf{SAT}$ problem. Recall that $\mathsf{ACC}^0$ circuits are circuits of constant depth with AND, OR, NOT, and modulo gates. There are non-trivial satisfiability algorithms for $\mathsf{ACC}^0$ circuits of size up to $2^{n^\varepsilon}$, where $\varepsilon > 0$ depends on the depth and modulo gates \citep{DBLP:journals/jacm/Williams14}. On the other hand, if such circuits admitted a non-trivial learning algorithm, it follows from the Speedup Lemma that polynomial size $\mathsf{ACC}^0$ circuits can be learned in quasi-polynomial time (see Figure \ref{f:speedup}).    

The Speedup Lemma suggests new approaches both to designing learning algorithms and to proving hardness of learning results. To design a quasi-polynomial time learning algorithm for polynomial-size circuits from a typical circuit class, it suffices to obtain a minimal improvement over the trivial brute-force algorithm for sub-exponential size circuits from the same class. Conversely, to conclude that the brute-force learning algorithm is essentially optimal for a typical class of polynomial-size circuits, it suffices to use an assumption under which subexponential-time learning is impossible.

We use the Speedup Lemma to show various structural results about learning. These include equivalences between several previously defined learning models, a dichotomy between sub-exponential time learnability and the existence of pseudo-random function generators in the non-uniform setting, and implications from non-trivial learning to circuit lower bounds. 

The techniques we explore have other consequences for complexity theory, such as Karp-Lipton style results for bounded-error exponential time, and results showing hardness of the Minimum Circuit Size Problem for a standard complexity class. In general, our results both exploit and strengthen the rich web of connections between learning, pseudo-randomness and circuit lower bounds, which promises to have further implications for our understanding of these fundamental notions. We now describe these contributions in more detail.

\begin{figure}[H]
\centering
\includegraphics[scale=1, trim={3.5cm 17.9cm 3cm 4.3cm}, clip]{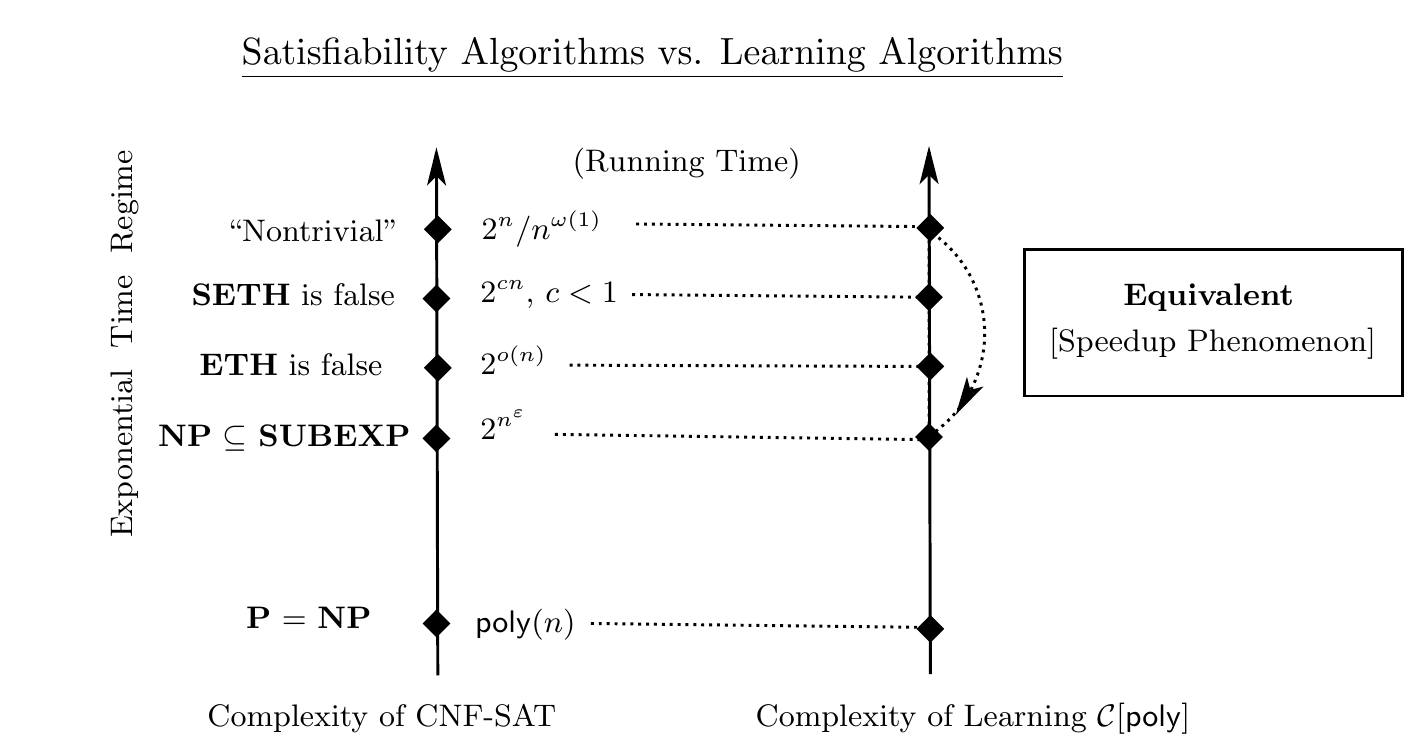}
\caption{A speedup phenomenon in computational learning theory for typical circuit classes for learning under the uniform distribution with membership queries. The speedup procedure simultaneously boosts accuracy and running time.}
\label{f:speedup}
\end{figure}

\subsection{Summary of Results}\label{ss:results}

We state below informal versions of our main results. We put these results in perspective and compare them to previous work in Section \ref{ss:related_work}.\\

{\bf \noindent Equivalences for Learning Models.}\vspace{0.15cm}

The Speedup Lemma shows that learnability of polynomial size circuits for typical circuit classes is not sensitive to the distinction between randomized sub-exponential time algorithms and randomized non-trivial  algorithms. We use the Speedup Lemma to further show that for such classes, learnability for a range of previously defined learning models is equivalent. These include the worst-case and average-case versions of function compression as defined by Chen et al. \citep{DBLP:journals/cc/ChenKKSZ15} (see also \citep{DBLP:journals/eccc/Srinivasan15}), and randomized learning with membership and equivalence queries \citep{DBLP:journals/ml/Angluin87}.\footnote{Our notion of randomized learning with membership and equivalence queries allows the learner's hypothesis to be incorrect on a polynomially small fraction of the inputs.}  The equivalence between function compression and learning in particular implies that accessing the entire truth table of a function represented by the circuit from the class confers no advantage in principle over having limited access to the truth table. 

\begin{theorem} [Equivalences for Learning Models, Informal Version]
\label{t:equivalences}
The following are equivalent for polynomial-size circuits from a typical circuit class $\mathfrak{C}$\emph{:}
\begin{enumerate}

\item[\emph{1.}] Sub-exponential time learning with membership queries. 

\item[\emph{2.}] Sub-exponential time learning with membership and equivalence queries.

\item[\emph{3.}] Probabilistic function compression.

\item[\emph{4.}] Average-case probabilistic function compression.

\item[\emph{5.}] Exponential time distinguishability from random functions.
\end{enumerate}
\end{theorem}

In particular, in the randomized sub-exponential time regime and when restricted to learning under the uniform distribution, Valiant's model \citep{Valiant84} and Angluin's model \citep{DBLP:journals/ml/Angluin87} are equivalent in power with respect to the learnability of typical classes of polynomial size circuits.

\vspace{0.5cm}

{\bf \noindent A Dichotomy between Learning and Pseudorandomness.}\vspace{0.15cm} 

It is well-known that if the class of polynomial-size circuits from a class $\mathfrak{C}$ is learnable, then there are no pseudo-random function generators computable in $\mathfrak{C}$, as the learner can be used to distinguish random functions from pseudo-random ones \citep{DBLP:journals/jacm/KearnsV94}. A natural question is whether the converse is true: can we in general build pseudo-random functions in the class from non-learnability of the class? 
We are able to use the Speedup Lemma in combination with other techniques to show such a result in the {\it non-uniform} setting, where the pseudo-random function generator as well as the learning algorithm are non-uniform. As a consequence, for each typical circuit class $\mathfrak{C}$, there is a {\it dichotomy} between pseudorandomness and learnability -- either there are pseudo-random function generators computable in the class, or the class is learnable, but not both.

\begin{theorem} [Dichotomy between Learning and Pseudorandomness, Informal Version]
\label{t:dichotomy}

Let $\mathfrak{C}$ be a typical circuit class. There are pseudo-random function generators computable by polynomial-size circuits from $\mathfrak{C}$ that are secure against sub-exponential size Boolean circuits if and only if polynomial-size circuits from $\mathfrak{C}$ are learnable non-uniformly in sub-exponential time.
\end{theorem}

\vspace{0.3cm}


{\bf \noindent Nontrivial Learning implies Circuit Lower Bounds.}\vspace{0.15cm}

In the algorithmic approach of Williams \citep{DBLP:journals/siamcomp/Williams13}, non-uniform circuit lower bounds against a class $\mathfrak{C}$ of circuits are shown by designing algorithms for satisfiability of $\mathfrak{C}$-circuits that beat the trivial brute-force search algorithm. Williams' approach has already yielded the result that $\mathsf{NEXP} \not \subseteq \mathsf{ACC^0}$ \citep{DBLP:journals/jacm/Williams14}.

It is natural to wonder if an analogue of the algorithmic approach holds for learning, and if so, what kinds of lower bounds would follow using such an analogue. We establish such a result -- non-trivial learning algorithms yield lower bounds for bounded-error probabilistic exponential time, just as non-trivial satisfiability algorithms yield lower bounds for non-deterministic exponential time. Our connection between learning and lower bounds has a couple of nice features. Our notion of ``non-trivial algorithm'' can be made even more fine-grained than that of Williams -- it is not hard to adapt our techniques to show that it is enough to beat the brute-force algorithm by a super-constant factor for learning algorithms with constant accuracy, as opposed to a polynomial factor in the case of Satisfiability. Moreover, non-trivial learning for bounded-depth circuits yields lower bounds against circuits with the {\it same} depth, as opposed to the connection for Satisfiability where there is an additive loss in depth \citep{oliveirathesis, DBLP:conf/icalp/JahanjouMV15}.
 
\begin{theorem} [Circuit Lower Bounds from Learning and from Natural Proofs, Informal Version]
\label{t:Lbs}
Let $\mathfrak{C}$ be any circuit class closed under projections. 
\begin{itemize}
\item[\emph{(}i\emph{)}] If polynomial-size circuits from $\mathfrak{C}$ are non-trivially learnable, then \emph{(}two-sided\emph{)} bounded-error probabilistic exponential time does not have polynomial-size circuits from $\mathfrak{C}$. 
\item[\emph{(}ii\emph{)}] If sub-exponential size circuits from $\mathfrak{C} = \mathsf{ACC}^0$ are non-trivially learnable, then one-sided error probabilistic exponential time does not have polynomial-size circuits from $\mathsf{ACC}^0$. 
\item[\emph{(}iii\emph{)}] If there are natural proofs useful against sub-exponential size circuits from $\mathfrak{C}$, then zero-error probabilistic exponential time does not have polynomial-size circuits from $\mathfrak{C}$.
\end{itemize}
\end{theorem}

Observe that the existence of natural proofs against sub-exponential size circuits yields stronger lower bounds than learning and satisfiability algorithms. (We refer to Section \ref{s:preliminaries} for a review of the inclusions between exponential time classes.)

\vspace{0.6cm}

{\bf \noindent Karp-Lipton Theorems for Probabilistic Exponential Time.}\vspace{0.15cm}

Our main results are about learning, but the techniques have consequences for complexity theory. Specifically, our use of pseudo-random generators has implications for the question of Karp-Lipton theorems for probabilistic exponential time. A Karp-Lipton theorem for a complexity class gives a connection between uniformity and non-uniformity, by showing that a non-uniform inclusion of the complexity class also yields a uniform inclusion. Such theorems were known for a range of classes such as $\mathsf{NP}$, $\mathsf{PSPACE}$, $\mathsf{EXP}$, and $\mathsf{NEXP}$ \citep{DBLP:conf/stoc/KarpL80, DBLP:journals/cc/BabaiFNW93, DBLP:journals/jcss/ImpagliazzoKW02}, but not for bounded-error probabilistic exponential time. We show the first such theorem for bounded-error probabilistic exponential time. A technical caveat is that the inclusion in our consequent is not completely uniform, but requires a logarithmic amount of advice.

\begin{theorem} [Karp-Lipton Theorem for Probabilistic Exponential Time, Informal Version]
\label{t:KL}~\\
If bounded-error probabilistic exponential time has polynomial-size circuits infinitely often, then bounded-error probabilistic exponential time is infinitely often in deterministic exponential time with logarithmic advice. 
\end{theorem}

\vspace{0.3cm}

{\bf \noindent Hardness of the Minimum Circuit Size Problem.}\vspace{0.15cm}

Our techniques also have consequences for the complexity of the Minimum Circuit Size Problem (MCSP). In MCSP, the input is the truth table of a Boolean function together with a parameter $s$ in unary, and the question is whether the function has Boolean circuits of size at most $s$. MCSP is a rare example of a problem in $\mathsf{NP}$ which is neither known to be in $\mathsf{P}$ or $\mathsf{NP}$-complete. In fact, we don't know much unconditionally about the complexity of this problem. We know that certain natural kinds of reductions cannot establish $\mathsf{NP}$-completeness \citep{DBLP:conf/coco/MurrayW15}, but until our work, it was unknown whether MCSP is hard for \emph{any} standard complexity class beyond $\mathsf{AC}^0$ \citep{DBLP:journals/siamcomp/AllenderBKMR06}. We show the first result of this kind. 

\begin{theorem} [Hardness of the Minimum Circuit Size Problem, Informal Version]
\label{t:MCSP}
The Minimum Circuit Size Problem is hard for polynomial-size formulas under truth-table reductions computable by polynomial-size constant-depth threshold circuits.
\end{theorem}


\noindent \rule{\textwidth}{0.5pt}\vspace{0.1cm}
\noindent \textbf{Remark.} This work contains several related technical contributions to the research topics mentioned above. We refer to the appropriate sections for more details. Finally, in Section \ref{s:open_problems} we highlight some open problems and directions that we find particularly attractive.

\vspace{0.1cm}
  
\subsection{Related Work}\label{ss:related_work}

\subsubsection{Speedups in Complexity Theory}

We are not aware of any unconditional speedup result of this form involving the time complexity of a natural class of computational problems, under a general computational model. In any case, it is instructive to compare Lemma \ref{l:speedup} to a few other speedup theorems in computational complexity.

A classic example is Blum's Speedup Theorem \citep{DBLP:journals/jacm/Blum67}. It implies that there is a recursive function $f \colon \mathbb{N} \to \mathbb{N}$ such that if an algorithm computes this function in time $T(n)$, then there is an algorithm computing $f$ in time $O(\log T(n))$. Lemma \ref{l:speedup} differs in an important way. It refers to a natural computational task, while the function provided by Blum's Theorem relies on an artificial construction. Another well-known speedup result is the Linear Speedup Theorem (cf. \citep[Section 2.4]{DBLP:books/daglib/0072413}). Roughly, it states that if a Turing Machine computes in time $T(n)$, then there is an equivalent Turing Machine that computes in time $T(n)/c$. The proof of this theorem is based on the simple trick of increasing the alphabet size of the machine. It is therefore dependent on the computational model, while Lemma \ref{l:speedup} is not.

Perhaps closer to our result are certain conditional derandomization theorems in complexity theory. We mention for concreteness two of them. In \citep{DBLP:journals/jcss/ImpagliazzoKW02}, it is proved that if $\mathsf{MA} \neq \mathsf{NEXP}$, then $\mathsf{MA} \subseteq \mathtt{i.o.}\mathsf{NTIME}[2^{n^\varepsilon}]/n^\varepsilon$, while in \citep{DBLP:journals/jcss/ImpagliazzoW01}, it is shown that if $\mathsf{BPP} \neq \mathsf{EXP}$, then $\mathsf{BPP} \subseteq \mathtt{i.o.}\mathsf{pseudo}$-$\mathsf{DTIME}[2^{n^\varepsilon}]$. It is possible to interpret these results as computational speedups, but observe that the faster algorithms have either weaker correctness guarantees, or require advice. Lemma \ref{l:speedup} on the other hand transforms a non-trivial learning algorithm into a sub-exponential time learning algorithm of the same type.

Further results have been discovered in more restricted computational models. For instance, in the OPP model, \citep{DBLP:conf/stoc/PaturiP10} proved that if Circuit-SAT has algorithms running in time $2^{(1-\delta)n}$, then it also has OPP algorithms running in time $2^{\varepsilon n}$. In bounded-depth circuit complexity, \citep{DBLP:journals/jacm/AllenderK10} established among other results that if the Formula Evaluation Problem has uniform $\mathsf{TC}^0$-circuits of size $O(n^k)$, then it also has uniform $\mathsf{TC}^0$-circuits of size $O(n^{1+\varepsilon})$.

If one considers other notions of complexity, we can add to this list several results that provide different, and often rather unexpected, forms of speedup. We mention, for instance, depth reduction in arithmetic circuit complexity (see .e.g. \citep{DBLP:conf/focs/AgrawalV08}), reducing the number of rounds in interactive proofs \citep{DBLP:journals/jcss/BabaiM88}, decreasing the randomness complexity of bounded-space algorithms \citep{DBLP:journals/jcss/NisanZ96}, cryptography in constant locality \citep{DBLP:journals/siamcomp/ApplebaumIK06}, among many others. 

\subsubsection{Connections between Pseudorandomness, Learning and Cryptography}

There are well-known connections between learning theory, theoretical cryptography and pseudorandomness (see e.g. \citep{DBLP:books/cu/Goldreich2001}). Indeed, pseudorandom distributions lie at the heart of the definition of semantic security \citep{DBLP:conf/stoc/GoldwasserM82, DBLP:journals/jcss/GoldwasserM84}, which permeates modern cryptography, and to every secure encryption scheme there is a naturally associated hard-to-learn (decryption) problem.

The other direction, i.e., that from a generic hard learning problem it is always possible to construct secure cryptographic schemes and other basic primitives, is much less clear.\footnote{Recall that secure private-key encryption is equivalent to the existence of one-way functions, pseudorandom generators and pseudorandom functions, with respect to polynomial time computations (cf. \citep{DBLP:books/crc/KatzLindell2007}). Nevertheless, not all these equivalences are known to hold when more refined complexity measures are considered, such as circuit depth. In particular, generic constructions of pseudorandom functions from the other primitives are not known in small-depth classes. This can be done under certain \emph{specific} hardness assumptions \citep{DBLP:journals/jacm/NaorR04}, but here we restrict our focus to generic relations between basic cryptographic primitives.} Following a research line initiated in \citep{DBLP:conf/focs/ImpagliazzoL90}, results more directly related to our work were established in \citep{DBLP:conf/crypto/BlumFKL93}. They proved in particular that private-key encryption and pseudorandom generators exist under a stronger \emph{average-case} hardness-of-learning assumption, where one also considers the existence of a hard distribution over the functions in the circuit class $\mathfrak{C}$.

However, these results and subsequent work leave open the question of whether hardness of learning in the usual case, i.e., the mere assumption that any efficient learner fails on some $f \in \mathfrak{C}$, implies the existence of pseudorandom functions computable by $\mathfrak{C}$-circuits. While there is an extensive literature basing standard cryptographic primitives on a variety of conjecturally hard learning tasks (see e.g., \citep{DBLP:journals/jacm/Regev09} and references therein for such a line of work), to our knowledge Theorem \ref{t:dichotomy} is the first result to establish a \emph{general equivalence} between the existence of pseudorandom functions and the hardness of learning, which holds for \emph{any} typical circuit class. A caveat is that our construction requires non-uniformity, and is established only in the exponential security regime.

\subsubsection{Lower Bounds from Learning Algorithms}

While several techniques from circuit complexity have found applications in learning theory in the past (see e.g., \citep{DBLP:journals/jacm/LinialMN93}), Fortnow and Klivans \citep{DBLP:journals/jcss/FortnowK09} were the first to systematically investigate the connection between learning algorithms and lower bounds in a generic setting.\footnote{For a broader survey on connections between algorithms and circuit lower bounds, we refer to \citep{MR3281015}.} 

For \emph{deterministic} learning algorithms using membership and equivalence queries, initial results from \citep{DBLP:journals/jcss/FortnowK09} and \citep{DBLP:journals/toct/HarkinsH13} were strengthened and simplified in \citep{DBLP:conf/coco/KlivansKO13}, where it was shown that non-trivial deterministic learning algorithms for $\mathfrak{C}$ imply that $\mathsf{EXP} \nsubseteq \mathfrak{C}$.

The situation for \emph{randomized} algorithms using membership queries is quite different, and only the following comparably weaker results were known. First, \citep{DBLP:journals/jcss/FortnowK09} proved that randomized polynomial time algorithms imply $\mathsf{BPEXP}$ lower bounds. This result was refined in \citep{DBLP:conf/coco/KlivansKO13}, where a certain connection involving sub-exponential time randomized learning algorithms and $\mathsf{PSPACE}$ was observed. More recently, \citep{DBLP:conf/icalp/Volkovich14} combined ideas from \cite{DBLP:conf/coco/KlivansKO13} and \citep{DBLP:journals/siamcomp/Santhanam09} to prove that efficient randomized learning algorithms imply lower bounds for $\mathsf{BPP}/1$, i.e., probabilistic polynomial time with advice. However, in contrast to the deterministic case, obtaining lower bounds from weaker running time assumptions had been elusive.\footnote{Some connections to lower bounds are also known in the context of learnability of \emph{arithmetic} circuits. We refer to \citep{DBLP:journals/jcss/FortnowK09, DBLP:conf/colt/Volkovich16} for more details.}

Indeed, we are not aware of any connection between two-sided non-trivial randomized algorithms and circuit lower bounds, even when considering different algorithmic frameworks in addition to learning. In particular, Theorem \ref{t:Lbs} (\emph{i}) seems to be the first result in this direction. It can be seen as an analogue of the connection between satisfiability algorithms and lower bounds established by Williams \citep{DBLP:journals/siamcomp/Williams13, DBLP:journals/jacm/Williams14}. But apart from this analogy, the proof of Theorem \ref{t:Lbs} employs significantly different techniques.

\subsubsection{Useful Properties, Natural Properties, and Circuit Lower Bounds}

The concept of natural proofs, introduced by Razborov and Rudich \citep{RR97}, has had a significant impact on research on unconditional lower bounds. Recall that a property $\mathcal{P}$ of Boolean functions is a natural property against a circuit class $\mathfrak{C}$ if it is: (1) efficiently computable (constructivity); (2) rejects all $\mathfrak{C}$-functions, and accepts at least one ``hard'' function (usefulness), and (3) is satisfied by most Boolean functions (denseness). In case $\mathcal{P}$ satisfies only conditions (1) and (2), is it said to be useful against $\mathfrak{C}$.

There are natural properties against $\mathsf{AC}^0[p]$ circuits, when $p$ is prime \citep{RR97}. But under standard cryptographic assumptions, there is no natural property against $\mathsf{TC}^0$ \citep{DBLP:journals/jacm/NaorR04}. Consequently, the situation for classes contained in $\mathsf{AC}^0[p]$ and for those that contain $\mathsf{TC}^0$ is reasonably well-understood. More recently, \citep{DBLP:journals/siamcomp/Williams16} (see also \citep{DBLP:journals/jcss/ImpagliazzoKW02}) proved that if $\mathsf{NEXP} \nsubseteq \mathfrak{C}$ then there are useful properties against $\mathfrak{C}$. This theorem combined with the lower bound from \citep{DBLP:journals/jacm/Williams14} show that $\mathsf{ACC}^0$ admits useful properties. 

Given these results, the existence of natural properties against $\mathsf{ACC}^0$ has become one of the most intriguing problems in connection with the theory of natural proofs.  Theorem \ref{t:Lbs} (\emph{iii}) shows that if there are $\mathsf{P}$-natural properties against sub-exponential size $\mathsf{ACC}^0$ circuits, then $\mathsf{ZPEXP} \nsubseteq \mathsf{ACC}^0$. This would lead to an improvement of Williams' celebrated lower bound which does not seem to be accessible using his techniques alone.\footnote{The result that $\mathsf{P}$-natural  properties against sub-exponential size circuits yield $\mathsf{ZPEXP}$ lower bounds was also obtained in independent work by Russell Impagliazzo, Valentine Kabanets and Ilya Volkovich (private communication).}

\subsubsection{Karp-Lipton Theorems in Complexity Theory}

Karp-Lipton theorems are well-known results in complexity theory relating non-uniform circuit complexity and uniform collapses. A theorem of this form was first established in \citep{DBLP:conf/stoc/KarpL80}, where they proved that if $\mathsf{NP} \subseteq \mathsf{Circuit}[\mathsf{poly}]$, then the polynomial time hierarchy collapses. This result shows that non-uniform circuit lower bounds cannot be avoided if our goal is a complete understanding of uniform complexity theory. 

Since their fundamental work, many results of this form have been discovered for complexity classes beyond $\mathsf{NP}$. In some cases, the proof required substantially new ideas, and the new Karp-Lipton collapse led to other important advances in complexity theory. Below we discuss the situation for two exponential complexity classes around $\mathsf{BPEXP}$, which is connected to Theorem \ref{t:KL}. 

A stronger Karp-Lipton theorem for $\mathsf{EXP}$ was established in \citep{DBLP:journals/cc/BabaiFNW93}, using techniques from interactive proofs and arithmetization. An important application of this result appears  in \citep{DBLP:conf/coco/BuhrmanFT98} in the proof that $\mathsf{MAEXP} \nsubseteq \mathsf{Circuit}[\mathsf{poly}]$. This is still one of the strongest known non-uniform lower bounds. For $\mathsf{NEXP}$, a Karp-Lipton collapse was proved in \citep{DBLP:journals/jcss/ImpagliazzoKW02}. This time the proof employed the easy witness method and techniques from pseudorandomness, and the result plays a fundamental role in Williams' framework \citep{DBLP:journals/siamcomp/Williams13}, which culminated in the proof that $\mathsf{NEXP} \nsubseteq \mathsf{ACC}^0$ \citep{DBLP:journals/jacm/Williams14}. (We mention that a Karp-Lipton theorem for $\mathsf{EXP}^\mathsf{NP}$ has also been established in \citep{DBLP:conf/fsttcs/BuhrmanH92}.) Karp-Lipton collapse theorems are known for a few other complexity classes contained in $\mathsf{EXP}$, and they have found applications in a variety of contexts in algorithms and complexity theory (see e.g., \citep{DBLP:journals/tcs/Yap83, DBLP:journals/jcss/FortnowS11}).

Despite this progress on proving Karp-Lipton collapses for exponential time classes, there is no published work on such for probabilistic classes. Theorem \ref{t:KL} is the first such result for the class $\mathsf{BPEXP}$. 

\subsubsection{The Minimum Circuit Size Problem}

The Minimum Circuit Size Problem (MCSP) and its variants has received a lot of attention in both applied  and theoretical research. Its relevance in practice is clear. From a theoretical point of view, it is one of the few natural problems in $\mathsf{NP}$ that has not been shown to be in $\mathsf{P}$ or $\mathsf{NP}$-complete. The hardness of MCSP is also connected to certain fundamental problems in proof complexity (cf. \citep{DBLP:books/daglib/0033749, MR3275844}).

Interestingly, a well-understood variant of MCSP is the Minimum DNF Size Problem, for which both NP-hardness \citep{masek79} and near-optimal hardness of approximation have been established \citep{DBLP:journals/siamcomp/AllenderHMPS08, DBLP:conf/focs/KhotS08}. However, despite the extensive literature on the complexity of the MCSP problem \citep{DBLP:conf/stoc/KabanetsC00, DBLP:journals/siamcomp/AllenderBKMR06, DBLP:conf/mfcs/AllenderD14, DBLP:conf/fsttcs/HitchcockP15, DBLP:conf/stacs/AllenderHK15, DBLP:conf/coco/MurrayW15, DBLP:conf/fsttcs/HitchcockP15, DBLP:journals/corr/AllenderGM15, DBLP:conf/coco/HiraharaW16}, and the intuition that it must also be computationally hard, there are few results providing evidence of its difficulty. Among these, we highlight the unconditional proof that $\mathsf{MCSP} \notin \mathsf{AC}^0$ \citep{DBLP:journals/siamcomp/AllenderBKMR06}, and the reductions showing  that $\mathsf{Factoring} \in \mathsf{ZPP}^\mathsf{MCSP}$ \citep{DBLP:journals/siamcomp/AllenderBKMR06} and $\mathsf{SZK} \subseteq \mathsf{BPP}^\mathsf{MCSP}$ \citep{DBLP:conf/mfcs/AllenderD14}. The lack of further progress has led to the formulation and investigation of a few related problems, for which some additional results have been obtained (cf. \citep{DBLP:journals/siamcomp/AllenderBKMR06, DBLP:conf/stacs/AllenderHK15, DBLP:journals/corr/AllenderGM15, DBLP:conf/coco/HiraharaW16}). 

More recently, \citep{DBLP:conf/coco/MurrayW15} provided some additional explanation for the difficulty of proving hardness of $\mathsf{MCSP}$. They unconditionally established that a  class of local reductions that have been used for many other NP-completeness proofs cannot work, and that the existence of a few other types of reductions would have significant consequences in complexity theory. Further results along this line appear in \citep{DBLP:conf/fsttcs/HitchcockP15}. 

Theorem \ref{t:MCSP} contributes to our understanding of the difficulty of MCSP by providing the first hardness results for a standard complexity class beyond $\mathsf{AC}^0$.  We hope this result will lead to further progress on the quest to determine the complexity of this elusive problem.\footnote{We have learned from Eric Allender (private communication) that in independent work with Shuichi Hirahara, they have shown some hardness results for the closely related problem of whether a string has high KT complexity. These results do not yet seem to transfer to MCSP and its variants. In addition, we have learned from Valentine Kabanets (private communication) that in recent independent work with Russell Impagliazzo and Ilya Volkovich, they have also obtained some results on the computational hardness of MCSP.}

\subsection{Main Techniques}\label{ss:techniques}

\subsubsection{Overview}

Our results are obtaining via a mixture of techniques from learning theory, computational complexity, pseudo-randomness and
circuit complexity. We refer to Figure \ref{f:techniques} for a web of connections involving the theorems stated in Section \ref{ss:results} and the  methods employed in the proofs. We start with an informal description of most of the techniques depicted in Figure \ref{f:techniques}, with pointers to some relevant references.\footnote{This is not a comprehensive survey of the original use or appearance of each method. It is included here only as a quick guide to help the reader to assimilate the main ideas employed in the proofs.}\\

\noindent \textbf{Nisan-Wigderson Generator} \citep{DBLP:journals/jcss/NisanW94}. The NW-Generator allows us to convert a function $f \colon \{0,1\}^n \to \{0,1\}$ into a family of functions $\mathsf{NW}(f)$. Crucially, if an algorithm $A$ is able to distinguish $\mathsf{NW}(f)$ from a random function, there is a reduction that only needs oracle access to $f$ and $A$, and that can be used to weakly approximate $f$. The use of the NW-Generator in the context of learning, for a function $f$ that is not necessarily hard, appeared recently in \citep{CIKK16}.\footnote{Interestingly, another unexpected and somewhat related use of the NW-generator appears in proof complexity (see e.g., \citep{MR3270137} and references therein).}\\

\begin{figure}[H]
\centering
\includegraphics{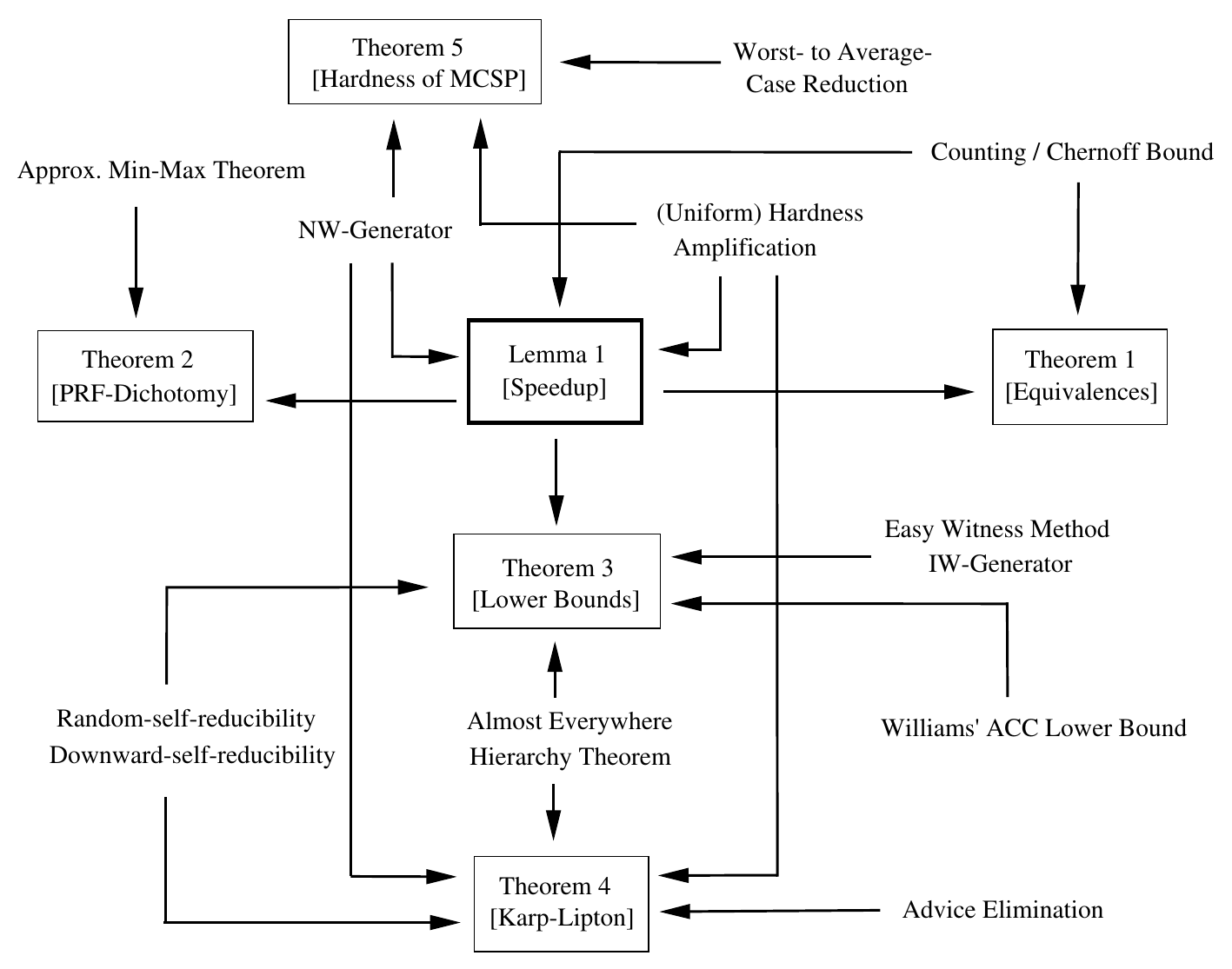}
\caption{An overview of the main techniques employed in the proof of each result discussed in Section \ref{ss:results}. An arrow from $P$ to $Q$ indicates that the proof of $Q$ relies on $P$.}
\label{f:techniques}
\end{figure}

\noindent \textbf{(Uniform) Hardness Amplification}. This is a well-known technique in circuit complexity (cf. \citep{DBLP:books/sp/goldreich2011/GoldreichNW11}), allowing one to produce a not much more complex function $\tilde{g} \colon \{0,1\}^{m(n)} \to \{0,1\}$, given oracle access to some function $g \colon \{0,1\}^n \to \{0,1\}$, that is much harder to approximate than $g$. The uniform formulation of this result shows that a weak approximator for $\tilde{g}$ can be converted into a strong approximator for $g$. The connection to learning was explicitly observed in \citep{DBLP:conf/colt/BonehL93}.\\

\noindent \textbf{Counting and Concentration Bounds}. This is a standard argument which allows one to prove that most Boolean functions on $n$-bit inputs cannot be approximated by Boolean circuits of size $\leq 2^n/n^{\omega(1)}$ (Lemma \ref{l:random_functions}). In particular, learning algorithm running in non-trivial time can only successfully learn a negligible fraction of all Boolean functions.\\

\noindent \textbf{Small-Support Min-Max Theorem} \cite{MR1274423, DBLP:conf/stoc/LiptonY94}. This is an approximate version of the well-known min-max theorem from game theory. It provides a bound on the support size of the mixed strategies. To prove Theorem \ref{t:dichotomy}, we consider a game between a function family generator and a candidate distinguisher, and this result allows us to move from a family of distinguishers against different classes of functions to a single universal distinguisher of bounded complexity.\\

\noindent \textbf{Worst-Case to Average-Case Reduction}. The NW-Generator and hardness amplification can be used to boost a very weak approximation into a strong one. In some circuit classes such as $\mathsf{NC}^1$, a further reduction allows one to obtain a circuit that is correct on every input with high probability (see e.g. \citep{DBLP:conf/approx/AllenderAW10}). This is particularly useful when proving hardness results for $\mathsf{MCSP}$.\\

\noindent \textbf{Easy Witness Method} \citep{DBLP:journals/jcss/Kabanets01} \textbf{and Impagliazzo-Wigderson Generator} \citep{DBLP:conf/stoc/ImpagliazzoW97}. The easy witness method is usually employed as a win-win argument: either a verifier accepts a string encoded by a small circuit, or every accepted string has high worst-case circuit complexity. No matter the case, it can be used to our advantage, thanks to the generator from \citep{DBLP:conf/stoc/ImpagliazzoW97} that transforms a worst-case hard string (viewed as a truth table) into a pseudorandom distribution of strings.\\

\noindent \textbf{(Almost Everywhere) Hierarchy Theorems.} A difficulty when proving Theorems \ref{t:Lbs} and \ref{t:KL} is that there are no known tight hierarchy theorems for randomized time. Our approach is therefore indirect, relying on the folklore result that bounded-space algorithms can diagonalize on every input length against all bounded-size circuits (Lemma \ref{l:iohardness} and Corollary \ref{c:iockthard}).\\

\noindent \textbf{Random-self-reducibility and Downward-self-reducibility.} These are important notions of self-reducibility shared by certain functions. Together, they can be used via a recursive procedure to obtain from a learning algorithm for such a function, which requires oracle access to the function, a standard randomized algorithm computing the same function \citep{DBLP:journals/jcss/ImpagliazzoW01, DBLP:journals/cc/TrevisanV07}.\\

\noindent \textbf{Advice Elimination.} This idea is important in the contrapositive argument establishing Theorem \ref{t:KL}. Assuming that a certain deterministic simulation of a function in $\mathsf{BPEXP}$ is not successful, it is not clear how to determine on each input length a ``bad'' string of that length for which the simulation fails. Such bad strings are passed as advice in our reduction, and in order to eliminate the dependency on them, we use an advice-elimination strategy  from \citep{DBLP:journals/cc/TrevisanV07}.

\subsubsection{Sketch of Proofs}

We describe next in a bit more detail how the techniques described above are employed in the proof of our main results. We stress that the feasibility of all these arguments crucially depend on the parameters associated to each result and technique. However, for simplicity our focus here will be on the qualitative connections.\\

\noindent \textbf{Lemma \ref{l:speedup}} (Speedup Lemma).
Given query access to a function $f \in \mathfrak{C}$ that we would like to learn to high accuracy, the first idea is to notice that if there is a  distinguisher against $\mathsf{NW}(f)$, then we can non-trivially approximate $f$ using membership queries. But since this is not the final goal of a strong learning algorithm, we consider $\mathsf{NW}(\tilde{f})$, the generator applied to the amplified version of $f$. Using properties of the NW-generator and hardness amplification, it follows that if there is a distinguisher against $\mathsf{NW}(\tilde{f})$, it is possible to approximate $\tilde{f}$, which in turn provides a strong approximator for $f$. (A similar strategy is employed in \citep{CIKK16}, where a natural property is used instead of a distinguisher.)

Next we use the assumption that $\mathfrak{C}$ has non-trivial learning algorithms to obtain a distinguisher against $\mathfrak{C}$. (For this approach to work, it is fundamental that the functions in $\mathsf{NW}(\tilde{f}) \subseteq \mathfrak{C}$. In other words, the reductions discussed above should not blow-up the complexity of the involved functions by too much. For this reason, $\mathfrak{C}$ must be a sufficiently strong circuit class.) By a counting argument and a concentration bound, while a non-trivial learning algorithm will weakly learn every function in $\mathfrak{C}$, it must fail to learn a random Boolean function with high probability. We apply this idea to prove that a non-trivial learner can be used as a distinguisher against $\mathsf{NW}(\tilde{f})$. 

These techniques can therefore be combined in order to boost a non-trivial learner for $\mathfrak{C}$ into a high-accuracy learner for $\mathfrak{C}$. This takes care of the accuracy amplification. The running time speedup comes from the efficiency of the reductions involved, and from the crucial fact that each function in $\mathsf{NW}(\tilde{f})$ is a function over $m \ll n$ input bits. In particular, the non-trivial but still exponential time learning algorithm for $\mathfrak{C}$ only needs to be invoked on Boolean functions over $m$ input bits. (This argument only sketches one direction in Lemma \ref{l:speedup}.)\\

\noindent \textbf{Theorem \ref{t:equivalences}} (Learning Equivalences). At the core of the equivalence between all learning and compression models in Theorem \ref{t:equivalences} is the idea that in each case we can obtain a certain distinguisher from the corresponding algorithm. Again, this makes fundamental use of counting and concentration bounds to show that non-trivial algorithms can be used as distinguishers. On the other hand, the Speedup Lemma shows that a distinguisher can be turned into a sub-exponential time randomized learning algorithm that requires membership queries only. 

In some models considered in the equivalence, additional work is necessary. For instance, in the learning model where equivalence queries are allowed, they must be simulated by the distinguisher. For exact compression, a hypothesis output by the sub-exponential time learner might still contain errors, and these need to be corrected by the compression algorithm. A careful investigation of the parameters involved in the proof make sure the equivalences indeed hold.\\

\noindent \textbf{Theorem \ref{t:dichotomy}} (Dichotomy between Learning and PRFs). It is well-known that the existence of learning algorithms for a class $\mathfrak{C}$ implies that $\mathfrak{C}$-circuits cannot compute pseudorandom functions. Using the Speedup Lemma, it follows that the existence of \emph{non-trivial} learning algorithms for $\mathfrak{C}$ implies that $\mathfrak{C}$ cannot compute \emph{exponentially} secure pseudorandom functions. 

For the other direction, assume that every samplable family $\mathcal{F}$ of functions from $\mathfrak{C}$ can be distinguished from a random function by some procedure $D_\mathcal{F}$ of sub-exponential complexity. By introducing a certain two-player game (Section \ref{ss:prf_dist}), we are able to employ the small-support min-max theorem to conclude that there is a single circuit of bounded size that distinguishes every family of functions in $\mathfrak{C}$ from a random function. In turn, the techniques behind the Speedup Lemma imply that every function in $\mathfrak{C}$ can be learned in sub-exponential time. 

We remark that the non-uniformity in the statement of Theorem \ref{t:dichotomy} comes from the application of a non-constructive min-max theorem.\\

\noindent \textbf{Theorem \ref{t:Lbs}} (Lower Bounds from Non-trivial Learning and Natural Proofs). Here we combine the Speedup Lemma with the self-reducibility approach from \citep{DBLP:journals/jcss/ImpagliazzoW01, DBLP:journals/cc/TrevisanV07, DBLP:journals/jcss/FortnowK09, DBLP:conf/coco/KlivansKO13} and other standard arguments. Assuming a non-trivial learning algorithm for $\mathfrak{C}$, we first boost it to a high-accuracy \emph{sub-exponential} time learner. Now if $\mathsf{PSPACE} \nsubseteq \mathfrak{C}$ we are done, since $\mathsf{PSPACE} \subseteq \mathsf{BPEXP}$. Otherwise, using a special \emph{self-reducible} complete function $f \in \mathsf{PSPACE}$ \citep{DBLP:journals/cc/TrevisanV07}, we are able obtain from a sub-exponential time \emph{learning} algorithm for $f$ a \emph{sub-exponential} time \emph{decision} algorithm computing $f$. Using the completeness of $f$ and a strong hierarchy theorem for bounded-space algorithms, standard techniques allow us to translate the hardness of $\mathsf{PSPACE}$ against bounded-size circuits and the non-trivial upper bound on the randomized complexity of $f$ into a non-uniform circuit lower bound for randomized exponential time. A win-win argument is used crucially to establish that no depth blow-up is necessary when moving from a non-trivial algorithm for (depth-$d$)-$\mathfrak{C}$ to a (depth-$d$)-$\mathfrak{C}$ circuit lower bound. For $\mathfrak{C} = \mathsf{ACC}^0$, we combine certain complexity collapses inside the argument with Williams' lower bound \citep{DBLP:journals/jacm/Williams14}.

In order to obtain even stronger lower bounds from natural properties against sub-exponential size circuits, we further combine this approach with an application of the easy witness method. This and other standard techniques lead to the collapse $\mathsf{BPEXP} = \mathsf{ZPEXP}$, which strengthens the final circuit lower bound.\\

\noindent \textbf{Theorem \ref{t:KL}} (Karp-Lipton Collapse for Probabilistic Time). This result does not rely on the Speedup Lemma, but its argument is somewhat more technically involved than the proof of Theorem \ref{t:Lbs}. The result is established in the contrapositive. Assuming that an attempted derandomization of $\mathsf{BPEXP}$ fails, we show that polynomial space can be simulated in sub-exponential randomized time. Arguing similarly to the proof of Theorem \ref{t:Lbs}, we conclude that there are functions in randomized exponential time that are not infinitely often computed by small circuits.

The first difficulty is that the candidate derandomization procedure on $n$-bit inputs requires the use of the NW-generator applied to a function on $n^c$-bit inputs, due to our setting of parameters. However, in order to invoke the self-reducibility machinery, we need to make sure the generator can be broken on \emph{every} input length, and not on infinitely many input lengths. To address this, we introduce logarithmic advice during the simulation, indicating which input length in $[n^c, (n+1)^c]$ should be used in the generator. This amount of advice is reflected in the statement of the theorem.

A second difficulty is that if the derandomization fails on some input string of length $n$, it is important in the reduction to know a ``bad'' string with this property. For each input length, a bad string is passed as advice to the learning-to-decision reduction (this is the second use of advice in the proof). This time we are able to remove the advice using an advice-elimination technique, which makes use of self-correctability as in \citep{DBLP:journals/cc/TrevisanV07}. Crucially, the advice elimination implies that randomized exponential time \emph{without advice} is not infinitely often contained in $\mathfrak{C}$, which completes the proof of the contrapositive of Theorem \ref{t:KL}.\\

\noindent \textbf{Theorem \ref{t:MCSP}} (Hardness of MCSP). Recall that this result states that MCSP is hard for $\mathsf{NC}^1$ with respect to non-uniform $\mathsf{TC}^0$ reductions. The proof of Theorem \ref{t:MCSP} explores the fine-grained complexity of the Nisan-Wigderson reconstruction procedure and of the hardness amplification reconstruction algorithm. In order words, the argument depends on the combined circuit complexity of the algorithm that turns a distinguisher for $\mathsf{NW}(\tilde{f})$ into a high-accuracy approximating circuit for $f$, under the notation of the proof sketch for Lemma \ref{l:speedup}. This time we obtain a distinguisher using an oracle to  MCSP. It is possible to show that this reduction can be implemented in non-uniform $\mathsf{TC}^0$. 

Observe that the argument just sketched only provides a randomized reduction that approximates the initial Boolean function $f$ under the uniform distribution. But Theorem \ref{t:MCSP} requires a \emph{worst-case} reduction from $\mathsf{NC}^1$ to MCSP. In other words, we must be able to compute any  $\mathsf{NC}^1$ function correctly on every input. This can be achieved using that there are functions in $\mathsf{NC}^1$ that are $\mathsf{NC}^1$-hard under $\mathsf{TC}^0$-reductions, and that in addition admit randomized worst-case to average-case reductions computable in $\mathsf{TC}^0$. Using non-uniformity, randomness can be eliminated by a standard argument. Altogether, this completes the proof that $\mathsf{NC}^1$ reduces to MCSP via a non-uniform $\mathsf{TC}^0$ computation.\\

\vspace{0.2cm}

These proofs provide a few additional examples of the use of pseudorandomness in contexts where this notion is not intrinsic to the result under consideration. For instance, the connection between non-trivial learning algorithms and lower bounds (Theorem \ref{t:Lbs}), the Karp-Lipton collapse for probabilistic exponential time (Theorem \ref{t:KL}), and the hardness of the Minimum Circuit Size Problem (Theorem \ref{t:MCSP}) are statements that do not explicitly refer to pseudorandomness. Nevertheless, the arguments discussed above rely on this concept in fundamental ways. This motivates a further investigation of the role of pseudorandomness in complexity theory, both in terms of finding more applications of the ``pseudorandom method'', as well as in discovering alternative proofs relying on different techniques.

\section{Preliminaries and Notation}
\label{s:preliminaries}

\subsection{Boolean Function Complexity}

We use $\mathcal{F}_m$ to denote the set of all Boolean functions $f \colon \{0,1\}^m \to \{0,1\}$. If $W$ is a probability distribution, we use $w \sim W$ to denote an element sampled according to $W$. Similarly, for a finite set $A$, we use $a \sim A$ to denote that $a$ is selected uniformly at random from $A$. Under this notation, $f \in \mathcal{F}_m$ represents a fixed function, while $f \sim \mathcal{F}_m$ is a uniformly random function. For convenience, we let $\mathcal{U}_n \eqdef \{0,1\}^n$. Following standard notation, $X \equiv Y$ denotes that random variables $X$ and $Y$ have the same distribution. We use standard asymptotic notation such as $o(\cdot)$ and $O(\cdot)$, and it always refer to a parameter $n \to \infty$, unless stated otherwise.

We say that $f, g \in \mathcal{F}_n$ are $\varepsilon$-close if $\Pr_{x\sim\mathcal{U}_n}[f(x) = g(x)]\geq 1-\varepsilon$. We say that $h \in \mathcal{F}_n$ computes $f$ with advantage $\delta$ if $\Pr_{x\sim\mathcal{U}_n}[f(x) = h(x)]\geq 1/2 + \delta$. It will sometimes be convenient to view a Boolean function $f \in \mathcal{F}_m$ as a subset of $\{0,1\}^m$ in the natural way.

We often represent Boolean functions as strings via the truth table mapping. Given a Boolean function $f \in \mathcal{F}_n$, $\mathtt{tt}(f)$ is the $2^n$-bit string which represents the truth table of $f$ in the standard way, and conversely, given a string $y \in \{0,1\}^{2^n}$, $\mathtt{fn}(y)$ is the Boolean function in $\mathcal{F}_n$ whose truth table is represented by $y$.

Let $\mathfrak{C} = \{\mathcal{C}_n\}_{n \in \mathbb{N}}$ be a class of Boolean functions, where each $\mathcal{C}_n \subseteq \mathcal{F}_n$. Given a language $L \subseteq \{0,1\}^*$, we write $L \in \mathfrak{C}$ if for every large enough $n$ we have that $L_n \eqdef \{0,1\}^n \cap L$ is in $\mathcal{C}_n$. Often we will abuse notation and view $\mathfrak{C}$ as a class of Boolean circuits. For convenience, we use number of wires to measure circuit size. We denote by $\mathfrak{C}[s(n)]$ the set of $n$-variable $\mathfrak{C}$-circuits of size at most $s(n)$. As usual, we say that a uniform complexity class $\Gamma$ is contained in $\mathfrak{C}[\mathsf{poly}(n)]$ if for every $L \in \Gamma$ there exists $k\geq 1$ such that $L \in \mathfrak{C}[n^k]$.

We say that $\mathfrak{C}$ is \emph{typical} if $\mathfrak{C} \in \{\mathsf{AC}^0, \mathsf{AC}^0[p], \mathsf{ACC}^0, \mathsf{TC}^0, \mathsf{NC}^1, \mathsf{Formula}, \mathsf{Circuit}\}$. Recall that 
$$
\mathsf{CNF}, \mathsf{DNF} \,\subsetneq\, \mathsf{AC}^0 \,\subsetneq\, \mathsf{AC}^0[p] \,\subsetneq\, \mathsf{ACC}^0 \,\subseteq\, \mathsf{TC}^0 \,\subseteq\, \mathsf{NC}^1 = \mathsf{Formula}[\mathsf{poly}] \,\subseteq\, \mathsf{Circuit}[\mathsf{poly}].
$$
\noindent We assume for convenience that $\mathsf{TC}^0$ is defined using (unweighted) majority gates instead of weighted threshold gates. Also, while $\mathsf{NC}^1$ typically refers to circuits of polynomial size and logarithmic depth, we consider the generalized version where $\mathsf{NC}^{1}[s]$ is the class of languages computed by circuits of size $\leq s$ and depth $\leq \log s$. 

While we restrict our statements to typical classes, it is easy to see that they generalize to most circuit classes of interest. When appropriate we use $\mathfrak{C}_d$ to restrict attention to $\mathfrak{C}$-circuits of depth at most $d$. In this work, we often find it convenient to suppress the dependence on $d$, which is implicit for instance in the definition of a circuit family from a typical bounded-depth circuit class, such as the first four typical classes in the list above. It will be clear from the context whether the quantification over $d$ is existential or universal. 

Given a sequence of Boolean functions $\{f_n\}_{n \in \mathbb{N}}$ with $f_n \colon \{0,1\}^n \to \{0,1\}$, we let $\mathfrak{C}^f$ denote the extension of $\mathfrak{C}$ that allows $\mathcal{C}_n$-circuits to have oracle gates computing $f_n$.

For a complexity class $\Gamma$ and a language $L \subseteq \{0,1\}^*$, we say that $L \in \mathtt{i.o.}\Gamma$ if there is a language $L' \in \Gamma$ such that $L_n = L'_n$ for infinitely many values of $n$. Consequently, if $\Gamma_1 \nsubseteq \mathtt{i.o.}\Gamma_2$ then there is a language in $\Gamma_1$ that disagrees with each language in $\Gamma_2$ on every large enough input length. 

Recall the following diagram of class inclusions involving standard complexity classes:\footnote{Non-uniform lower bounds against unrestricted polynomial size circuits are currently known only for $\mathsf{MAEXP}$, the exponential time analogue of $\mathsf{MA}$ \citep{DBLP:conf/coco/BuhrmanFT98}.}
$$
\mathsf{ZPP} \subseteq \diamondinclusion{\mathsf{RP}}{\mathsf{NP}\\\mathsf{BPP}}{\mathsf{MA}} \subseteq \mathsf{PSPACE} \subseteq \mathsf{EXP} \subseteq \mathsf{ZPEXP} \subseteq \diamondinclusion{\mathsf{REXP}}{\mathsf{NEXP}\\\mathsf{BPEXP}}{\mathsf{MAEXP}} \subseteq \mathsf{EXPSPACE}.
$$
\noindent In order to avoid confusion, we fix the following notation for exponential complexity classes. $\mathsf{E}$ refers to languages computed in time $2^{O(n)}$. $\mathsf{EXP}$ refers to languages computed with bounds of the form $2^{n^c}$ for some $c \in \mathbb{N}$. $\mathsf{SUBEXP}$ denotes complexity $2^{n^{\varepsilon}}$ for a fixed but arbitrarily small $\varepsilon > 0$. Finally, $\mathsf{ESUBEXP}$ refers to a bound of the form $2^{2^{n^\varepsilon}}$, again for a fixed but arbitrarily small $\varepsilon > 0$. These conventions are also used for the $\mathsf{DSPACE}(\cdot)$ and $\mathsf{BPTIME}(\cdot)$ variants, such as $\mathsf{BPE}$, $\mathsf{BPSUBEXP}$ and $\mathsf{EXPSPACE}$. For instance, a language $L \subseteq \{0,1\}^*$ is in $\mathsf{BPSUBEXP}$ if for every $\varepsilon > 0$ there is a bounded-error randomized algorithm that correctly computes $L$ in time $\leq 2^{n^\varepsilon}$ on every input of length $n$, provided that $n$ is sufficiently large. For quasi-polynomial time classes such as $\mathsf{RQP}$ and $\mathsf{BPQP}$, the convention is that for each language in the class there is a constant $c \geq 1$ such that the corresponding algorithm runs in time at most $O(n^{(\log n)^c})$.

We will use a few other standard notions, and we refer to standard textbooks in computational complexity and circuit complexity for more details.

\subsection{Learning and Compression Algorithms}

The main learning model with which we concern ourselves is PAC learning under the uniform distribution with membership queries.

\begin{definition}[Learning Algorithms]
\label{d:learning}
Let $\mathfrak{C}$ be a circuit class. Given a size function $s \colon \mathbb{N} \rightarrow \mathbb{N}$ and a time function $T \colon \mathbb{N} \rightarrow \mathbb{N}$, we say that $\mathfrak{C}[s]$ has $(\varepsilon(n), \delta(n))$-learners running in time $T(n)$ if there is a randomized oracle algorithm $A^f$ \emph{(}the learner\emph{)} such that for every large enough $n\in \mathbb{N}$\emph{:}
\begin{itemize}
\item For every function $f \in \mathfrak{C}[s(n)]$, given oracle access to $f$, with probability at least $1 - \delta(n)$ over its internal randomness, $A^f(1^n)$ outputs a Boolean circuit $h$ such that $\Pr_{x \sim \mathcal{U}_n}[f(x) \neq h(x)] \leq \varepsilon(n)$.
\item For every function $f$, $A^f(1^n)$ runs in time at most $T(n)$.
\end{itemize}
\end{definition}

It is well-known that the confidence of a learning algorithm can be amplified without significantly affecting the running time (cf. \citep{KearnsVazirani:94}), and unless stated otherwise we assume that $\delta(n) = 1/n$.

A \emph{weak learner} for $\mathfrak{C}[s(n)]$ is a $(1/2-1/n^c, 1/n)$-learner, for some fixed $c > 0$ and sufficiently large $n$. We say $\mathfrak{C}[s]$ has \emph{strong learners} running in time $T$ if for each $k \geq 1$ there is a $(1/n^k, 1/n)$-learner for $\mathfrak{C}[s]$ running in time $T$. Different values for the accuracy parameter $k$ can lead to different running times, but we will often need only a fixed large enough $k$ when invoking the learning algorithm. On the other hand, when proving that a class has a strong learner, we show that the claimed asymptotic running time holds for all fixed $k \in \mathbb{N}$. For simplicity, we may therefore omit the dependence of $T$ on $k$. We say that $\mathfrak{C}[s]$ has \emph{non-trivial learners} if it has $(1/2- 1/n^k,1/n)$-learners running in time $T(n)= 2^n/n^{\omega(1)}$, for some fixed $k \in \mathbb{N}$.

We also discuss randomized learning under the uniform distribution with membership queries and \emph{equivalence queries} \citep{DBLP:journals/ml/Angluin87}. In this stronger model, the learning algorithm is also allowed to make queries of the following form: Is the unknown function $f$ computed by the Boolean circuit $C$? Here $C$ is an efficient representation of a Boolean circuit produced be the learner. The oracle answers ``yes'' if the Boolean function computed by $C$ is $f$; otherwise it returns an input $x$ such that $C(x) \neq f(x)$.  

\begin{definition}[Compression Algorithms]
\label{d:compression}
Given a circuit class $\mathfrak{C}$ and a size function $s\colon \mathbb{N} \rightarrow \mathbb{N}$, a compression algorithm for $\mathfrak{C}[s]$ is an algorithm $A$ for which the following hold\emph{:}
\begin{itemize}
\item Given an input $y \in \{0,1\}^{2^n}$, $A$ outputs a circuit $D$ \emph{(}not necessarily in $\mathfrak{C}$\emph{)} of size $o(2^n/n)$ such that if $\mathtt{fn}(y) \in \mathfrak{C}[s(n)]$ then $D$ computes $\mathtt{fn}(y)$.
\item $A$ runs in time polynomial in $|y| = 2^n$.
\end{itemize}
We say $\mathfrak{C}[s]$ admits compression if there is a \emph{(}polynomial time\emph{)} compression algorithm for $\mathfrak{C}[s]$.
\end{definition} 

We will also consider the following variations of compression. If the algorithm is probabilistic, producing a correct circuit with probability $\geq 2/3$, we say $\mathfrak{C}[s]$ has \emph{probabilistic compression}. If the algorithm produces a circuit $D$ which errs on at most $\varepsilon(n)$ fraction of inputs for $\mathtt{fn}(y)$ in $\mathfrak{C}[s]$, we say that $A$ is an \emph{average-case compression algorithm with error} $\varepsilon(n)$. We define correspondingly what it means for a circuit class to have average-case compression or probabilistic average-case compression.

\subsection{Natural Proofs and the Minimum Circuit Size Problem}

We say that $\mathfrak{R} = \{\mathcal{R}_n\}_{n \in \mathbb{N}}$ is a \emph{combinatorial property} (of Boolean functions) if $\mathcal{R}_n \subseteq \mathcal{F}_n$ for all $n$. We use $L_{\mathfrak{R}}$ to denote the language of truth-tables of functions in $\mathfrak{R}$. Formally, $L_{\mathfrak{R}} = \{y \mid y = \mathtt{tt}(f)~\text{for some}~f \in \mathcal{R}_n~\text{and}~n \in \mathbb{N}\}$.

\begin{definition}[Natural Properties \citep{RR97}]\label{d:natural_property}
Let $\mathfrak{R} = \{\mathcal{R}_n\}$ be a combinatorial property,  $\mathfrak{C}$ a circuit class, and $\mathfrak{D}$ a \emph{(}uniform or non-uniform\emph{)} complexity class. We say that $\mathfrak{R}$ is a $\mathfrak{D}$-natural property useful against $\mathfrak{C}[s(n)]$ if there is $n_0 \in \mathbb{N}$ such that the following holds\emph{:}
\begin{itemize}
\item[\emph{(}i\emph{)}] \emph{Constructivity.} $L_{\mathfrak{R}} \in \mathfrak{D}$.
\item[\emph{(}ii\emph{)}] \emph{Density.} For every $n \geq n_0$, $\Pr_{f \sim \mathcal{F}_n}[f \in \mathcal{R}_n] \geq 1/2$.
\item[\emph{(}iii\emph{)}] \emph{Usefulness.} For every $n \geq n_0$, we have $\mathcal{R}_n \cap \mathcal{C}_n[s(n)] = \emptyset$.
\end{itemize}
\end{definition}

\begin{definition} [Minimum Circuit Size Problem] \label{d:mcsp}
Let $\mathfrak{C}$ be a circuit class. The Minimum Circuit Size Problem for $\mathfrak{C}$, abbreviated as $\mathsf{MCSP}$-$\mathfrak{C}$, is defined as follows\emph{:}
\begin{itemize}
\item \emph{Input.} A pair $(y,s)$, where $y \in \{0,1\}^{2^n}$ for some $n \in \mathbb{N}$, and $1 \leq s \leq 2^n$ is an integer \emph{(}inputs not of this form are rejected\emph{)}.
\item \emph{Question.} Does $\mathtt{fn}(y)$ have $\mathfrak{C}$-circuits of size at most $s$?
\end{itemize}
\end{definition}

We also define a variant of this problem, where the circuit size is not part of the input.

\begin{definition} [Unparameterized Minimum Circuit Size Problem] \label{d:umcsp}
Let $\mathfrak{C}$ be a circuit class, and $s \colon \mathbb{N} \rightarrow \mathbb{N}$ be a function. The Minimimum Circuit Size Problem for $\mathfrak{C}$ with parameter $s$, abbreviated as $\mathsf{MCSP}$-$\mathfrak{C}[s]$, is defined as follows\emph{:}
\begin{itemize}
\item \emph{Input.} A string $y \in \{0,1\}^{2^n}$, where $n \in \mathbb{N}$ \emph{(}inputs not of this form are rejected\emph{)}.
\item \emph{Question.} Does $\mathtt{fn}(y)$ have $\mathfrak{C}$-circuits of size at most $s(n)$?
\end{itemize}
\end{definition}

Note that a dense property useful against $\mathfrak{C}[s(n)]$ is a dense subset of the complement of MCSP-$\mathfrak{C}[s]$.

\subsection{Randomness and Pseudorandomness}

\begin{definition}[Pseudorandom Generators]\label{d:PRG}
Let $\ell \colon \mathbb{N} \to \mathbb{N}$, $h \colon \mathbb{N} \to \mathbb{N}$ and $\varepsilon \colon \mathbb{N} \rightarrow [0,1]$ be functions, and let $\mathfrak{C}$ be a circuit class. A sequence $\{G_n\}$ of functions
$G_n \colon \{0,1\}^{\ell(n)} \to \{0,1\}^n$ is an $(\ell, \varepsilon)$ pseudorandom generator
\emph{(PRG)} against $\mathfrak{C}[h(n)]$ if for any sequence of circuits $\{D_n\}$ with $D_n \in \mathfrak{C}[h(n)]$ and for all large enough $n$,
$$
\left | \Pr_{w \sim U_n}[D_n(w)=1] - \Pr_{x \sim \mathcal{U}_{\ell(n)}}[D_n(G_n(x))=1] \right | \;\leq\;
\varepsilon(n).
$$ 
\noindent The pseudorandom generator is called quick if its range is computable in time $2^{O(\ell(n))}$. 
\end{definition}

\begin{theorem} [PRGs from computational hardness \citep{DBLP:journals/jcss/NisanW94, DBLP:conf/stoc/ImpagliazzoW97}]
\label{t:hardness_PRG}
Let $s \colon \mathbb{N} \rightarrow \mathbb{N}$ be a time-constructible function such that $n \leq s(n) \leq 2^n$ for every $n \in \mathbb{N}$. There is a constant $c > 0$ and an algorithm which, given as input $n$ in unary and the truth table of a Boolean function on $s^{-1}(n)$ bits which does not have circuits of size $n^c$, computes the range of a $(\ell(n), 1/n)$ pseudorandom generator against $\mathsf{Circuit}[n]$ in time $2^{O(\ell(n))}$, where $\ell(n) = O( (s^{-1}(n))^2/\log n)$.
\end{theorem}

\begin{definition}[Distinguishers and Distinguishing Circuits]\label{d:distinguisher}
Given a probability distribution $W_n$ with  $\mathsf{Support}(W_n) \subseteq \{0,1\}^n$ and a Boolean function $h_n \colon \{0,1\}^n \to \{0,1\}$, we say that $h_n$ is a distinguisher for $W_n$ if
$$
\left | \Pr_{w \sim W_n}[h_n(w)=1] - \Pr_{x \sim \mathcal{U}_n}[h_n(x)=1] \right | \;\geq\;1/4.
$$
We say that a circuit $D_n$ is a circuit distinguisher for $W_n$ if $D_n$ computes a function $h_n$ that is a distinguisher for $W_n$. A function $f \colon \{0,1\}^{*} \rightarrow \{0,1\}$ is a distinguisher for a sequence of distributions $\{W_n\}$ if for each large enough $n$, $f_n$ is a distinguisher for $W_n$, where $f_n$ is the restriction of $f$ to $n$-bit inputs.
\end{definition}

The following is a slight variant of a definition in \citep{CIKK16}.

\begin{definition}[Black-Box Generator]\label{d:blackbox_generator}
Let $\ell \colon \mathbb{N} \to \mathbb{N}$, $\gamma(n) \in [0,1]$, and $\mathfrak{C}$ be a circuit class. A black-box $(\gamma, \ell)$-function generator within $\mathfrak{C}$ is a mapping that associates to any $f\colon \{0,1\}^n \to \{0,1\}$ a family $\mathtt{GEN}(f) = \{g_z\}_{z \in \{0,1\}^m}$ of functions $g_z \colon \{0,1\}^\ell \to \{0,1\}$, for which the following conditions hold\emph{:}
\begin{itemize}
\item[\emph{(}i\emph{)}] \emph{Family size.} The parameter $m \leq \mathsf{poly}(n,1/\gamma)$. 
\item[\emph{(}ii\emph{)}] \emph{Complexity.} For every $z \in \{0,1\}^m$, we have $g_z \in \mathfrak{C}^f[\mathsf{poly}(m)]$.
\item[\emph{(}iii\emph{)}] \emph{Reconstruction.} Let $L = 2^{\ell}$ and $W_L$ be the distribution supported over $\{0,1\}^{L}$ that is generated by $\mathtt{tt}(g_z)$, where $z \sim \mathcal{U}_m$. There is a randomized algorithm $A^f$, taking as input a circuit $D$ and having oracle access to $f$, which when $D$ is a distinguishing circuit for $W_L$, with probability at least $1 - 1/n$ outputs a circuit of size $\mathsf{poly}(n, 1/\gamma, \mathtt{size}(D))$ that is $\gamma$-close to $f$. Furthermore, $A^f$ runs in time at most $\mathsf{poly}(n, 1/\gamma, L(n))$.
\end{itemize}
\end{definition}

This definition is realized by the following result.

\begin{theorem}
[Black-Box Generators for Restricted Classes \citep{CIKK16}]\label{t:blackbox_ACzerop}
Let $p$ be a fixed prime, and $\mathfrak{C}$ be a typical circuit class containing $\mathsf{AC}^0[p]$. For every $\gamma \colon \mathbb{N} \to [0,1]$ and $\ell \colon \mathbb{N} \to \mathbb{N}$ there exists a black-box $(\gamma, \ell)$-function generator within $\mathfrak{C}$.
\end{theorem}

\begin{definition}[Complexity Distinguisher]\label{d:probcompdist}
Let $\mathfrak{C}$ be a circuit class and consider functions $s,T \colon \mathbb{N} \rightarrow \mathbb{N}$. We say that a probabilistic oracle algorithm $A^g$ is a complexity distinguisher for $\mathfrak{C}[s(n)]$ running in time $T$ if $A^g(1^n)$ always halts in time $T(n)$ with an ouput in $\{0,1\}$, and the following hold\emph{:}
\begin{itemize}
\item For every $g \in \mathfrak{C}[s(n)]$, $\;\Pr_A[A^g(1^n) = 1] \leq 1/3$.
\item $\mathbb{E}_{g \sim \mathcal{F}_{n},A} [A^g(1^n)] \geq 2/3$.
\end{itemize}
\end{definition}

\begin{definition}[Zero-Error Complexity Distinguisher]\label{d:zerocompdist}
Let $\mathfrak{C}$ be a circuit class and $s,T \colon \mathbb{N} \rightarrow \mathbb{N}$ be functions. We say that a probabilistic oracle algorithm $A^g$ is a zero-error complexity distinguisher for $\mathfrak{C}[s(n)]$ running in time $T$ if $A^g(1^n)$ always halts in time $T(n)$ with an output in $\{0, 1, ?\}$, and the following hold\emph{:}
\begin{itemize}
\item If $g \in \mathfrak{C}[s(n)]$, $A^g(1^n)$ always outputs $0$ or $?$, and $\Pr_A[A^g(1^n) = \;?] 
\leq 1/3$. 
\item For every $n \geq 1$ there exists a family of functions $\mathcal{S}_n \subseteq \mathcal{F}_n$ with $|\mathcal{S}_n|/|\mathcal{F}_n| \geq 1 - o(1)$  such that for every $f \in \mathcal{S}_n$, $A^f(1^n)$ always outputs $1$ or $?$, and $\Pr_A[A^f(1^n) = \;?] \leq 1/3$.
\end{itemize}
\end{definition}

We will make use of the following standard concentration of measure result.

\begin{lemma}[Chernoff Bound, cf. {\citep[Theorem 2.1]{JLR_randomgraphs}}]\label{l:concentration}
Let $X\sim\mathsf{Bin}(m, p)$ and $\lambda = mp$. For any $t \geq 0$,
$$
\Pr[|X-\E[X]| \geq t] \leq \exp\! \left( -\frac{t^2}{2 (\lambda + t/3)} \right).
$$
\end{lemma}

\section{Learning Speedups and Equivalences}\label{s:proof_gap}

\subsection{The Speedup Lemma}

We start with the observation that the usual upper bound on the number of small Boolean circuits also holds for unbounded fan-in circuit classes with additional types of gates.

\begin{lemma}[Bound on the number of functions computed by small circuits]\label{l:number_circuits}
Let $\mathfrak{C}$ be a typical circuit class. For any $s \colon \mathbb{N} \to \mathbb{N}$ satisfying $s(n) \geq n$ there are at most $2^{50 s(n) \log s(n)}$ functions in $\mathcal{F}_n$ computed by $\mathfrak{C}$-circuits of size at most $s(n)$. 
\end{lemma}

\begin{proof}
A circuit over $n$ input variables and of size at most $s(n)$ can be represented by its underlying directed graph together with information about the type of each gate. A node of the graph together with its gate type can be described using $O(\log s(n))$ bits, since for a typical circuit class there are finitely many types of gates. In addition, each input variable can be described as a node of the graph using $O(\log n) = O(\log s(n))$ bits, since by assumption $s(n) \geq n$. Finally, using this indexing scheme, each wire of the circuit corresponding to a directed edge in the circuit graph can be represented with $O(\log s(n))$ bits. Consequently, as we measure circuit size by number of wires, any circuit of size at most $s(n)$ can be represented using at most $O(s(n) \log s(n))$ bits. The lemma follows from the trivial fact that a Boolean circuit computes at most one function in $\mathcal{F}_n$ and via a conservative estimate for the asymptotic notation.
\end{proof}

Lemmas \ref{l:concentration} and \ref{l:number_circuits} easily imply the following (folklore) result.

\begin{lemma}[Random functions are hard to approximate] \label{l:random_functions} 
Let $\mathfrak{C}$ be a typical circuit class, $s \geq n$, and $\delta \in [0,1/2]$. Then,
$$
\Pr_{f \sim \mathcal{F}_n}[\exists\,\mathfrak{C}\text{-circuit of size}\,\leq s(n)~\text{computing}\,f\,\text{with advantage}~\delta(n)] \,\leq\,  \exp\!\left (- \delta^2 2^{n-1} + 50s \log s \right ).
$$
\end{lemma}

\begin{proof}
Let $g \in \mathcal{F}_n$ be a fixed function. It follows from Lemma \ref{l:concentration} with $p = 1/2$, $m = 2^n$, $t = \delta2^n$, and using $\delta \leq 1/2$ that
$$
\Pr_{f \sim \mathcal{F}_n}[g~\text{computes}~f~\text{with advantage}~\delta(n)] \leq  \exp\!\left (- \frac{\delta^2 2^n}{2} \right )\!.
$$
The claim follows immediately from this estimate, Lemma \ref{l:number_circuits}, and a union bound.
\end{proof}

\begin{lemma} [Non-trivial learners imply distinguishers] 
\label{l:learningdist}
Let $\mathfrak{C}$ be a typical circuit class, $s \colon \mathbb{N} \rightarrow \mathbb{N}$ be a size bound, and $T \colon \mathbb{N} \rightarrow \mathbb{N}$ be a time bound such that $T(n) = 2^n/n^{\omega(1)}$. If $\mathfrak{C}[s]$ has weak learners running in time $T$, then $\mathfrak{C}[s(n)]$ has complexity distinguishers running in time $T(n) \cdot \mathsf{poly}(n)$.
\end{lemma}

\begin{proof} 

By the assumption that $\mathfrak{C}$ is weakly learnable, there is a probabilistic oracle algorithm $A_{\mathsf{learn}}^{f}$, running in time $T(n)$ on input $1^n$, which when given oracle access to $f \in \mathfrak{C}[s]$, outputs with probability at least $1-1/n$ a Boolean circuit $h$ which agrees with $f$ on at least a $1/2+1/n^k$ fraction of inputs of length $n$, for some universal constant $k$. We show how to construct from $A_{\mathsf{learn}}^{f}$ an oracle algorithm $A_{\mathsf{dist}}^{f}$ which is a complexity distinguisher for $\mathfrak{C}[s]$.

$A_{\mathsf{dist}}^{f}$ operates as follows on input $1^n$.  It runs $A_{\mathsf{learn}}^{f}$ on input $1^n$. If $A_{\mathsf{learn}}^{f}$ does not output a hypothesis, $A_{\mathsf{dist}}^{f}$ outputs `1'. Otherwise $A_{\mathsf{dist}}^{f}$ estimates the agreement between the hypothesis $h$ output by the learning algorithm and the function $f$ by querying $f$ on $n^{5k}$ inputs of length $n$ chosen uniformly at random, and checking for each such input whether $f$ agrees with $h$. The estimated agreement is computed to be the fraction of inputs on which $f$ agrees with $h$. If it is greater than $1/2+1/n^{2k}$, $A_{\mathsf{dist}}^{f}$ outputs `0', otherwise it outputs `1'.

By the assumption on efficiency of the learner $A_{\mathsf{learn}}^{f}$, it follows that $A_{\mathsf{dist}}^{f}$ runs in time $T(n) \cdot \mathsf{poly}(n)$. Thus we just need to argue that $A_{\mathsf{dist}}^{f}$ is indeed a complexity distinguisher.

For a uniformly random $f$, the probability that $A_\mathsf{learn}^f$ outputs a hypothesis $h$ that has agreement greater than $1/2+1/n^{4k}$ with $f$ is exponentially small. This is because $A_\mathsf{learn}^f$ runs in time $T(n) = 2^n/n^{\omega(1)}$, and hence if it outputs a hypothesis, it must be of size at most $2^n/n^{\omega(1)}$. By Lemma \ref{l:random_functions}, only an exponentially small fraction of functions can be approximated by circuits of such size. Also, given that a circuit $h$ has agreement at most $1/2+1/n^{4k}$ with $f$, the probability that the estimated agreement according to the procedure above is greater than $1/2+1/n^{2k}$ is exponentially small by Lemma \ref{l:concentration}. Thus, for a uniformly random $f$, the oracle algorithm $A_{\mathsf{dist}}^{f}$ outputs `0' with exponentially small probability, and hence for large enough $n$, it outputs `1' with probability at least $2/3$.

For $f \in \mathfrak{C}[s(n)]$, by the correctness and efficiency of the learning algorithm, $A_{\mathsf{learn}}^{f}$ outputs a hypothesis $h$ with agreement at least $1/2+1/n^k$ with $f$, with probability at least $1-1/n$. For such a hypothesis $h$, using Lemma \ref{l:concentration} again, the probability that the estimated agreement is smaller than $1/2+1/n^{2k}$ is exponentially small. Thus, for $n$ large enough, with probability at least $2/3$, $A_{\mathsf{dist}}^{f}$ outputs `0'.
\end{proof}

\begin{lemma} [Faster learners from distinguishers]
\label{l:distfastlearning}
Let $\mathfrak{C}$ be a typical circuit class. If $\mathfrak{C}[\mathsf{poly}(n)]$ has complexity distinguishers running in time $2^{O(n)}$, then for every $\varepsilon > 0$, $\mathfrak{C}[\mathsf{poly}(n)]$ has strong learners running in time $O(2^{n^{\varepsilon}})$. If for some $\varepsilon > 0$, $\mathfrak{C}[2^{n^{\varepsilon}}]$ has complexity distinguishers running in time $2^{O(n)}$, then $\mathfrak{C}[\mathsf{poly}(n)]$ has strong learners running in time $2^{\mathsf{log}(n)^{O(1)}}$.
\end{lemma}

\begin{proof} 
We prove the first part of the Lemma, and the second part follows analogously using a different parameter setting.

Let $\mathfrak{C}$ be a typical circuit class. If $\mathfrak{C} = \mathsf{AC}^0$, the lemma holds unconditionally since this class can be learned in quasi-polynomial time \citep{DBLP:journals/jacm/LinialMN93}. Assume otherwise that $\mathfrak{C}$ contains $\mathsf{AC}^0[p]$, for some fixed prime $p$. By assumption, $\mathfrak{C}[\mathsf{poly}(n)]$ has a complexity distinguisher $A_0^g$ running in time $2^{O(n)}$. We show that for every $\varepsilon > 0$ and every $k > 0$, $\mathfrak{C}[\mathsf{poly}(n)]$ has $(1/n^k,1/n)$-learners running in time $O(2^{n^{\varepsilon}})$. Let $\varepsilon' > 0$ be any constant such that $\varepsilon' < \varepsilon$. By Theorem \ref{t:blackbox_ACzerop} there exists a black-box $(\gamma, \ell)$-function generator $\mathtt{GEN}$ within $\mathfrak{C}$, where $\gamma = 1/n^k$ and $\ell = n^{\varepsilon'}$. For this setting of $\gamma$ and $\ell$ we have that the parameter $m$ for $\mathtt{GEN}(f)$ in Definition \ref{d:blackbox_generator} is $\mathsf{poly}(n)$, and that for each $z \in \{0,1\}^m$, we have $g_z \in \mathfrak{C}^f[\mathsf{poly}(n)]$. Let $A_1^f$ be the randomized reconstruction algorithm for $\mathtt{GEN}(f)$.

We define a $(1/n^k, 0.99)$-learner $A^f$ for $\mathfrak{C}[\mathsf{poly}(n)]$ running in time $O(2^{n^{\varepsilon}})$; the confidence can then be amplified to satisfy the definition of a strong learner while not increasing the running time of the learner by more than a polynomial factor. The learning algorithm operates as follows. It interprets the oracle algorithm $A_0^{g}$ on input $1^{\ell}$ as a probabilistic polynomial-time algorithm $D(\cdot, \vec{r})$ which is explicitly given
the truth table of $g$, of size $L = 2^{\ell}$, as input, with $\vec{r}$ the randomness for this algorithm. It guesses $\vec{r}$ at random and then computes a circuit $D_L$ of size $2^{O(\ell)}$ which is equivalent to $D(\cdot, \vec{r})$ on inputs of size $2^{\ell}$, using the standard transformation of polynomial-time algorithms into circuits. It then runs $A_1^{f}$ on input $D_L$, and halts with the same output as $A_1^{f}$. Observe that the queries made by the reconstruction algorithm can be answered by the learner, since it also has query access to $f$.

Using the bounds on running time of $A_0$ and $A_1$, it is easy to see that $A^f$ can be implemented to run in time $2^{O(\ell)}$, which is at most $2^{n^{\epsilon}}$ for large enough $n$. We need to argue that $A^f$ is a correct strong learner for $\mathfrak{C}[\mathsf{poly}(n)]$. The critical point is that when $f \in \mathfrak{C}[\mathsf{poly}(n)]$, with noticeable probability, $D_L$ is a distinguishing circuit for $W_L$ (using the terminology of Definition \ref{d:blackbox_generator}), and we can then take advantage of the properties of the reconstruction algorithm. We now spell this out in more detail.

When $f \in \mathfrak{C}[\mathsf{poly}(n)]$, using the fact that $\mathfrak{C}$ is typical and thus closed under composition with itself, and that it contains $\mathsf{AC}^0[p]$, we have that for each $z \in \{0,1\}^m$, $g_z \in \mathfrak{C}[\mathsf{poly}(n)]$. Note that the input size for $g_z$ is $\ell = n^{\varepsilon'}$, and hence also $g_z \in \mathfrak{C}[\mathsf{poly}(\ell)]$. Using now that $A_0$ is a complexity distinguisher, we have that for any $z \in \{0,1\}^{m}$, $\Pr_A[A^{g_z}(1^\ell) = 1] \leq 1/3$, while $\mathbb{E}_{g \sim \mathcal{F}_{\ell},A} [A^g(1^\ell)] \geq 2/3$. By a standard averaging argument and the fact that probabilities are bounded by $1$, this implies that with probability at least $0.05$ over the choice of $\vec{r}$, $D_L$ is a distinguishing circuit for $W_L$. Under the properties of the reconstruction algorithm $A_1^{f}$, when given as input such a circuit $D_L$, with probability at least $1-1/n$, the output of $A_1^{f}$ is $1/n^k$-close to $f$. Hence with probability at least $0.05 \cdot(1-1/n) > 0.01$ over the randomness of $A$, the output of $A^f$ is $1/n^k$-close to $f$, as desired. As observed before, the success probability of the learning algorithm can be amplified by standard techniques (cf. \citep{KearnsVazirani:94}).

The second part of the lemma follows by the same argument with a different choice of parameters, using a black-box $(\gamma, \ell)$-function generator with $\gamma = 1/n^k$ and $\ell = (\log n)^c$, where $c$ is chosen large enough as a function of $\varepsilon$. Again, the crucial point is that the relative circuit size of each $g_z$ compared to its number of input bits is within the size bound of the distinguisher.
\end{proof}

\begin{remark}
While Lemma \emph{\ref{l:distfastlearning}} is sufficient for our purposes, we observe that the same argument shows in fact that the conclusion holds under the weaker assumption that the complexity distinguisher runs in time $2^{n^c}$, for a fixed $c \in \mathbb{N}$. In other words, it is possible to obtain faster learners from complexity distinguishers running in time that is quasi-polynomial in the length of the truth-table of its oracle function.
\end{remark}

\begin{lemma} [Speedup Lemma]
\label{l:learningspeedup}
Let $\mathfrak{C}$ be a typical circuit class. The following hold\emph{:}

\begin{itemize}
\item \emph{(Low-End Speedup)} $\mathfrak{C}[\mathsf{poly}(n)]$ has non-trivial learners if and only if for each $\varepsilon > 0$, $\mathfrak{C}[\mathsf{poly}(n)]$ has strong learners running in time $O(2^{n^{\varepsilon}})$.
\item \emph{(High-End Speedup)} There exists $\varepsilon > 0$ such that $\mathfrak{C}[2^{n^{\varepsilon}}]$ has non-trivial learners if and only if $\mathfrak{C}[\mathsf{poly}(n)]$ has strong learners running in time $2^{\mathsf{log}(n)^{O(1)}}$.
\end{itemize}
\end{lemma}

\begin{proof} 
First we show the Low-End Speedup result. The ``if'' direction is trivial, so we only need to consider the ``only if'' case. This follows from Lemma \ref{l:distfastlearning} and Lemma \ref{l:learningdist}. Indeed, by Lemma \ref{l:learningdist}, if $\mathfrak{C}[\mathsf{poly}(n)]$ has non-trivial learners, it has complexity distinguishers running in time $2^n/n^{\omega(1)}$. By Lemma \ref{l:distfastlearning}, the existence of such complexity distinguishers implies that for each $\varepsilon > 0$, $\mathfrak{C}[\mathsf{poly}(n)]$ has strong learners running in time $O(2^{n^{\varepsilon}})$, and we are done.

Next we show the High-End Speedup result. The proof for the ``only if'' direction is completely analogous to the corresponding proof for the Low-End Speedup result. The ``if'' direction, however, is not entirely trivial. We employ a standard padding argument to establish this case, thus completing the proof of Lemma \ref{l:learningspeedup}. 

Suppose that $\mathfrak{C}[\mathsf{poly}(n)]$ has a strong learner running in time $2^{\mathsf{log}(n)^c}$, for some constant $c > 0$. Let $A_{\mathsf{low}}$ be a learning algorithm witnessing this fact. We show how to use $A_{\mathsf{low}}$ to construct a learning algorithm $A_{\mathsf{high}}$ which is a $(1/n, 1/\mathsf{poly}(n))$-learner for $\mathfrak{C}[2^{n^{1/3c}}]$, and runs in time $\leq 2^{\sqrt{n}}$. As usual, confidence can be boosted without a significant increase of running time, and it follows that $\mathfrak{C}[2^{n^{1/3c}}]$ has non-trivial learners according to our definition.

On input $1^n$ and with oracle access to some function $f\colon \{0,1\}^n \to \{0,1\}$, $A_{\mathsf{high}}^{f}(1^n)$ simulates $A_{\mathsf{low}}^{f'}(1^{n'})$, where $n' \eqdef n + 2^{\lceil n^{1/3c} \rceil}$, and $f' \colon \{0,1\}^{n'} \to \{0,1\}$ is the (unique) Boolean function satisfying the following properties. For any input $x' = xy \in \{0,1\}^{n'}$, where $|y| = 2^{\lceil n^{1/3c} \rceil}$ and $|x| = n$, $f'(x')$ is defined to be $f(x)$. Note that if $f \in \mathfrak{C}[2^{n^{1/3c}}]$ then $f' \in \mathfrak{C}[O(n')]$: the linear-size $\mathfrak{C}$-circuit for $f'$ on an input $x'$ of length $n'$ just simulates the smallest $\mathfrak{C}$-circuit for $f$ on its $n$-bit prefix. During the simulation, whenever $A_{\mathsf{low}}$ makes an oracle call $x'$ to $f'$, $A_{\mathsf{high}}$ answers it using an oracle call $x$ to $f$, where $x$ is the prefix of $x'$ of length $n$. By definition of $f'$, this simulation step is always correct. When $A_{\mathsf{low}}^{f'}$ completes its computation and outputs a hypothesis $h'$ on $n'$ input bits, $A_{\mathsf{high}}$ outputs a modified hypothesis $h$ as follows: it chooses a random string $r$ of length $2^{\lceil n^{1/3c} \rceil}$, and outputs the circuit $h_r$ defined by $h_r(x) \eqdef h'(xr)$. Note that by the assumed efficiency of $A_{\mathsf{low}}$, $A_{\mathsf{high}}$ halts in time $\leq 2^{\sqrt{n}}$ on large enough $n$.

By the discussion above, it is enough to argue that the hypothesis $h$ output by $A_{\mathsf{high}}$ is a good hypothesis with probability at least $1/\mathsf{poly}(n)$. Since $A_{\mathsf{low}}$ is a strong learner and since the $f'$ used as oracle to $A_{\mathsf{low}}$ in the simulation has linear size, for all large enough $n$, with probability at least $1-1/n'$, $h'$ disagrees with $f'$ on at most a $1/(n')^k$ fraction of inputs of length $n'$, where $k$ is a large enough constant fixed in the construction above. Consider a randomly chosen $r$ of length $n'-n$. By a standard Markov-type argument, when $h'$ is good, for at least a $1/\mathsf{poly}(n)$ fraction of the strings $r$, $h_r(x)$ disagrees with $f(x)$ on at most a $1/n$ fraction of inputs. This completes the argument.
\end{proof}

\subsection{Equivalences for Learning, Compression, and Distinguishers} 

\begin{theorem} [Algorithmic Equivalences]
 \label{t:bdederrorequiv}

Let $\mathfrak{C}$ be a typical circuit class. The following statements are equivalent\emph{:}
\begin{enumerate}
\item[\emph{1.}] $\mathfrak{C}[\mathsf{poly}(n)]$ has non-trivial learners.
\item[\emph{2.}] For each $\varepsilon > 0$ and $k \in \mathbb{N}$, $\mathfrak{C}[\mathsf{poly}(n)]$ can be learned to error $\leq n^{-k}$ in time $O(2^{n^{\epsilon}})$.
\item[\emph{3.}] $\mathfrak{C}[\mathsf{poly}(n)]$ has probabilistic \emph{(}exact\emph{)} compression.
\item[\emph{4.}] $\mathfrak{C}[\mathsf{poly}(n)]$ has probabilistic average-case compression with error $o(1)$.
\item[\emph{5.}] $\mathfrak{C}[\mathsf{poly}(n)]$ has complexity distinguishers running in time $2^{O(n)}$.
\item[\emph{6.}] For each $\varepsilon > 0$, $\mathfrak{C}[\mathsf{poly}(n)]$ has complexity distinguishers running in time $O(2^{n^{\epsilon}})$.
\item[\emph{7.}] $\mathfrak{C}[\mathsf{poly}(n)]$ can be learned using membership and equivalence queries to sub-constant error in non-trivial time.
\end{enumerate}
\end{theorem}

\begin{proof} We establish these equivalences via the following complete set of implications:\vspace{0.15cm}

$(5) \Rightarrow (2)$: This follows from Lemma \ref{l:distfastlearning}.
 
$(2) \Rightarrow (1)$ and $(6) \Rightarrow (5)$: These are trivial implications.

 $(2) \Rightarrow (4)$: Probabilistic compression for $\mathfrak{C}[\mathsf{poly}(n)]$ follows from simulating a $(1/n^3, 1/n)$-learner for the class running in time $O(2^{\sqrt{n}})$, and answering any oracle queries by looking up the corresponding bit in the truth table of the function, which is given as input to the compression algorithm. The compression algorithm returns as output the hypothesis of the strong learner, and by assumption this agrees on a $(1-1/n^3)$ fraction of inputs of length $n$ with the input function, with probability at least $1-1/n$. Moreover, since the simulated learner runs in time $O(2^{\sqrt{n}})$, the circuit that is output has size at most $O(2^{\sqrt{n}})$. It is clear that the simulation of the learner can be done in time $2^{O(n)}$, as required for a compression algorithm.

 $(2) \Rightarrow (3)$: This follows exactly as above, except that there is an additional step after the simulation of the learner.  Once the learner has output a hypothesis $h$, the compression algorithm compares this hypothesis with its input truth table entry by entry, simulating $h$ whenever needed. If $h$ differs from the input truth table on more than a $1/n^3$ fraction of inputs, the compression algorithm rejects -- this happens with probability at most $1/n$ by assumption on the learner. If $h$ and the input truth table differ on at most $1/n^3$ fraction of inputs of length $n$, the compression algorithm computes by brute force a circuit of size at most $2^n/n^2$ which computes the function $h'$ that is the XOR of $h$ and the input truth table. The upper bound on size follows from the fact that $h'$ has at most $2^n/n^3$ 1's. Finally, the compression algorithm outputs $h \oplus h'$. For any typical circuit class, the size of the corresponding circuit is $O(2^n/n^2)$. Note that $h \oplus h'$ computes the input truth table exactly.

$(2) \Rightarrow (6)$: This follows from Lemma \ref{l:learningdist}.

$(1) \Rightarrow (5)$, $(3) \Rightarrow (5)$, and $(4) \Rightarrow (5)$: The distinguisher runs the circuit output by the learner or compression algorithm on every input of length $n$, and computes the exact agreement with its input $f$ on length $n$ by making $2^n$ oracle queries to $f$. If the circuit agrees with $f$ on at least a $2/3$ fraction of inputs, the distinguisher outputs $0$, otherwise it outputs $1$. By the assumption on the learner/compression algorithm, for $f \in \mathfrak{C}[\mathsf{poly}(n)]$, the distinguisher outputs $0$ with probability at least $2/3$. Using Lemma \ref{l:random_functions}, for a random function, the probability that the distinguisher outputs $1$ is at least $2/3$.

$(7) \Rightarrow (5)$: The complexity distingisher has access to the entire truth-table, and can answer the membership and equivalence queries of the learner in randomized time $2^{O(n)}$. Randomness is needed only to simulate the random choices of the learning algorithm, while the answer to each query can be computed in deterministic time. Since the learner runs in time $2^n/n^{\omega(1)}$, whenever it succeeds it outputs a hypothesis circuit of at most this size. The complexity distinguisher can compare this hypothesis to its input truth-table, and similarly to the arguments employed before, is able to distinguish random functions from functions in $\mathfrak{C}[\mathsf{poly}(n)]$.

$(2) \Rightarrow (7)$: This is immediate since the algorithm from $(2)$ is faster, has better accuracy, and makes no equivalence queries.
\end{proof}

\begin{theorem} [Equivalences for zero-error algorithms]
\label{t:zeroerrorequiv}
Let $\mathfrak{C}[\mathsf{poly}(n)]$ be a typical circuit class. The following statements are equivalent\emph{:}
\begin{enumerate}
\item[\emph{1.}] There are $\mathsf{P}$-natural proofs useful against $\mathfrak{C}[\mathsf{poly}(n)]$.

\item[\emph{2.}] There are $\mathsf{ZPP}$-natural proofs useful against $\mathfrak{C}[\mathsf{poly}(n)]$.

\item[\emph{3.}] For each $\varepsilon > 0$, there are $\mathsf{DTIME}(O(2^{(\log N)^{\varepsilon}}))$-natural proofs useful against $\mathfrak{C}[\mathsf{poly}(n)]$, where $N = 2^n$ is the truth-table size.

\item[\emph{4.}] For each $\varepsilon > 0$, there are zero-error complexity distinguishers for $\mathfrak{C}[\mathsf{poly}(n)]$ running in time $O(2^{n^{\varepsilon}})$.
\end{enumerate}
\end{theorem}
\begin{proof} 
We establish these equivalences via the following complete set of implications:\vspace{0.15cm}

$(1) \Rightarrow (2)$: This is a trivial implication.

$(3) \Rightarrow (4)$: This is almost a direct consequence of the definitions, except that the density of the natural property has to be amplified to $1 - o(1)$ before converting the algorithm into a zero-error complexity distinguisher. This is a standard argument, and can be achieved by defining a new property from the initial one. More details can be found, for instance, in the proof of \citep[Lemma 2.7]{CIKK16}. 

$(2) \Rightarrow (1)$:\footnote{This argument is folklore. It has also appeared in more recent works, such as \citep{DBLP:journals/siamcomp/Williams16}.}  Let $A$ be an algorithm running in zero-error probabilistic time $m^k$ on inputs of length $m$ and with error probability $\leq 1/4$, for $m$ large enough and $k$ an integer, and deciding a combinatorial property $\mathfrak{R}$ useful against $\mathfrak{C}[\mathsf{poly}(n)]$. We show how to define a combinatorial property $\mathfrak{R'}$ useful against $\mathfrak{C}[\mathsf{poly}(n)]$ such that $\mathfrak{R'} \in \mathsf{P}$, and such that at least a $1/8$ fraction of the truth tables of any large enough input length belong to $L_{\mathfrak{R'}}$. This fraction can be amplified by defining a new natural property $\mathfrak{R''}$ such that
any string $yz$ with $|y| = |z|$ belongs to $L_{\mathfrak{R''}}$ if and only if either $y \in L_{\mathfrak{R'}}$ or $z \in L_{\mathfrak{R'}}$ (see e.g. \citep{CIKK16}). 

We define $\mathfrak{R'}$ via a deterministic polynomial-time algorithm $A'$ deciding $L_{\mathfrak{R'}}$. Given an input truth table $y$ of size $2^{n'}$, $A'$ acts as follows: it determines the largest integer $n$ such that $n (k+1) < n'$. It decomposes the input truth table as $y=xzw$, where $|x|=2^n$, $|z|=2^{kn}$, and the remaining part $w$ is irrelevant. It runs $A$ on $x$, using $z$ as the randomness for the simulation of $A$. If $A$ accepts, it accepts; if $A$ rejects or outputs `?', it rejects.

It should be clear that $A'$ runs in polynomial time. The fact that $A'$ accepts at least a $1/8$ fraction of truth tables of any large enough input length follows since for any $x \in L_{\mathfrak{R}}$, $A$ outputs `?' with probability at most $1/3$, and at least a $1/2$ fraction of strings of length $2^n$ are in $L_{\mathfrak{R}}$. It only remains to argue that the property $\mathfrak{R'}$ is useful against $\mathfrak{C}[\mathsf{poly}(n)]$. But any string $y$ of length $2^{n'}$ accepted by $A'$ has as a substring the truth table of a function on $n = \Omega(n')$ bits which is accepted by $A$ and hence is in $L_{\mathfrak{R}}$. Since $\mathfrak{R}$ is useful against $\mathfrak{C}[\mathsf{poly}(n)]$, this implies that $\mathfrak{R'}$ is useful against $\mathfrak{C}[\mathsf{poly}(n)]$.

$(4) \Rightarrow (3)$: The proof is analogous to $(2) \Rightarrow (1)$.

$(3) \Rightarrow (1)$: This is a trivial direction since $N = 2^n$.

$(1) \Rightarrow (3)$: This implication uses an idea of Razborov and Rudich \cite{RR97}. Suppose there are $\mathsf{P}$-natural proofs useful against $\mathfrak{C}[\mathsf{poly}(n)]$. This means in particular that for every $c \geq 1$, there is a polynomial-time algorithm $A_{c}$, which on inputs of length $2^n$, where $n \in \mathbb{N}$, accepts at least a $1/2$ fraction of inputs, and rejects all inputs $y$ such that $\mathtt{fn}(y) \in \mathfrak{C}[n^c]$. Consider an input $y$ to $A_c$ of length $2^n$, and let $\varepsilon > 0$ be fixed. Let $y'$ be the substring of $y$ such that $\mathtt{fn}(y')$ is the subfunction of $\mathtt{fn}(y)$ obtained by fixing the first $n-n^{\varepsilon}$ bits of the input to $\mathtt{fn}(y)$ to $0$. It is easy to see that if $\mathtt{fn}(y) \in \mathfrak{C}[n^c]$, then $\mathtt{fn}(y') \in \mathfrak{C}[(n')^{c/\varepsilon}]$, where $n'$ denotes the number of input bits of $\mathsf{fn}(y')$.

Let $d \geq 1$ be any constant, and $\varepsilon > 0$ be fixed. We show how to define an algorithm $B_d$ which runs in time $O(2^{\mathsf{log}(N)^{\varepsilon}})$ on an input of length $N = 2^n$, deciding a combinatorial property which is useful against $\mathfrak{C}$-circuits of size $n^d$. (Using the same approach, it is possible to design a single algorithm that works for any fixed $d$ whenever $n$ is large enough, provided that we start with a natural property that is useful in this stronger sense.) On input $y$ of length $N$, $B_d$ computes $y'$ of length $2^{\mathsf{log}(N)^{\varepsilon}}$, as defined in the previous paragraph. For the standard encoding of truth tables, $y'$ is a prefix of $y$, and can be computed in time $O(|y'|)$. $B_d$ then simulates $A_{\lceil d/\varepsilon \rceil}$ on $y'$, accepting if and only if the simulated algorithm accepts. The simulation halts in time $\mathsf{poly}(|y'|)$, as $A_{\lceil d/\varepsilon \rceil}$ is a poly-time algorithm. For a random input $y$, $B_d$ accepts with probability at least $1/2$, using that $y'$ is uniformly distributed, and the assumption that $A_{\lceil d/\varepsilon \rceil}$ witnesses natural proofs against a circuit class. For an input $y$ such that $\mathtt{fn}(y) \in \mathfrak{C}[n^d]$, $B_d$ always rejects, as in this case, $\mathtt{fn}(y') \in \mathfrak{C}[(n')^{d/\varepsilon}]$, and so $A_{\lceil d/\varepsilon \rceil}$ rejects $y'$, using the assumption that $A_{\lceil d/\varepsilon \rceil}$ decides a combinatorial property useful against $n$-bit Boolean functions in $\mathfrak{C}[n^{\lceil d/\varepsilon \rceil}]$. 
\end{proof}

\section{Learning versus Pseudorandom Functions} \label{s:learning_prfs}

\subsection{The PRF-Distinguisher Game}\label{ss:prf_dist}

In this section we consider (non-uniform) randomized oracle circuits $B^\mathcal{O}$ from $\mathsf{Circuit}^\mathcal{O}[t]$, where $t$ is an upper bound on the number of wires in the circuit. Recall that a circuit from this class has a special  gate type that computes according to the oracle $\mathcal{O}$, which will be set to some fixed Boolean function $f \colon \{0,1\}^m \to \{0,1\}$ whenever we discuss the computation of the circuit. 

We will view such circuits either as distinguishers or learning algorithms, where the oracle is the primary input to the circuit. For this reason and because our results are stated in the non-uniform setting, we assume from now on that such circuits have no additional input except for variables $y_1, \ldots, y_\ell$ representing the random bits, where $\ell \leq t$. If $w \in \{0,1\}^\ell$ is a fixed sequence of bits, we use $B^\mathcal{O}_w$ to denote the deterministic oracle circuit obtaining from the circuit $B^\mathcal{O}$ by setting its randomness to $w$. Observe that (non-uniform) learning algorithms can be naturally described by randomized oracle circuits from $\mathsf{Circuit}^\mathcal{O}$ with multiple output bits. The output bits describe the output hypothesis, under some reasonable encoding for Boolean circuits.\footnote{In this non-uniform framework it is possible to derandomize a learning circuit with some blow-up in circuit size, but we will not be concerned with this matter here.} 

We will consider pairs $(G_n,\mathcal{D}_n)$ where $G_n \subseteq \mathcal{F}_n$ and $\mathcal{D}_n$ is a distribution with $\mathsf{Support}(\mathcal{D}_n) \subseteq G_n$. This notation is convenient when defining samplable function families and pseudorandom function families. 

\begin{definition}[Pseudorandom Function Families]\label{d:prfs}
We say that a pair $(G_n, \mathcal{D}_n)$ is a $(t(n),\varepsilon(n))$-\emph{pseudorandom function family} \emph{(PRF)} in $\mathcal{C}[s(n)]$ if $G_n \subseteq \mathcal{C}[s(n)]$ and for every randomized oracle circuit $B^\mathcal{O} \in \mathsf{Circuit}^\mathcal{O}[t(n)]$,
$$
\left | \Pr_{g \sim \mathcal{D}_n, \,w}[B^g(w) = 1] - \Pr_{f \sim \mathcal{F}_n, \,w}[B^f(w) = 1]  \right  | \;\leq \; \varepsilon.
$$
\end{definition}

This definition places no constraint on the complexity of generating the pair $(G_n, \mathcal{D}_n)$. In order to capture this, we restrict attention to $G_n \subseteq \mathcal{C}_n$ for some typical circuit class $\mathfrak{C} = \{\mathcal{C}_n\}$, and assume a fixed encoding of circuits from $\mathfrak{C}$ by strings of length polynomial in the size of the circuit. We say that a circuit $A \in \mathsf{Circuit}[S]$ is a $\mathcal{C}_n$-\emph{sampler} if $A$ outputs valid descriptions of circuits from $\mathcal{C}_n$.  

\begin{definition}[Samplable Function Families]\label{d:samplable_functions}
We say that a pair $(G_n, \mathcal{D}_n)$ with $G_n \subseteq \mathcal{C}_n$ is $S$-\emph{samplable} if there exists a $\mathcal{C}_n$-sampler $A \in \mathsf{Circuit}[S]$ on $\ell \leq S$ input variables such that $A(\mathcal{U}_\ell) \equiv \mathcal{D}_n$, where we associate each output string of $A$ to its corresponding Boolean function.
\end{definition}

It is well-known that the existence of learning algorithms for a circuit class $\mathcal{C}_n[s(n)]$ implies that there are no secure pseudorandom function families in $\mathcal{C}_n[s(n)]$. Moreover, this remains true even for function families that are not efficiently samplable. Following the notation from Definition \ref{d:learning}, we can state a particular form of this observation as follows.

\begin{proposition}[Learning $\mathfrak{C}$ implies no PRFs in $\mathfrak{C}$]\label{p:learning_implies_no_prfs}
Assume there is a randomized oracle circuit in $\mathsf{Circuit}^\mathcal{O}[t(n)]$ that $(1/3,1/n)$-learns every function in $\mathcal{C}_n[s(n)]$, where $n \leq t(n) \leq 2^n/n^2$. Then for large enough $n$ there are no $(\mathtt{poly}(t(n)),1/10)$-pseudorandom function families in $\mathcal{C}_n[s(n)]$.
\end{proposition} 

Our goal for the rest of this section is to establish a certain converse of Proposition \ref{p:learning_implies_no_prfs} (Theorem \ref{t:no_samp_PRFs_implies_learning} and Corollary \ref{c:no_PRFs_implies_learning}). An important technical tool will be a ``small-support'' version of the min-max theorem, described next.\\ 

\noindent \textbf{Small-Support Approximate Min-Max Theorem for Bounded Games \citep{MR1274423, DBLP:conf/stoc/LiptonY94}}. We follow the notation from \citep{DBLP:conf/stoc/LiptonY94}. Let $M$ be an $r \times c$ real-valued matrix, $p$ be a probability distribution over its rows, and $q$ be a probability distribution over its columns. The classic min-max theorem \citep{MR1512486} states that
\begin{equation}\label{eq:min_max}
\min_p \max_{j \in [c]} \;\mathbb{E}_{i \sim p}[M(i,j)] \; = \; \max_q \min_{i \in [r]} \;\mathbb{E}_{j \sim q}[M(i,j)].
\end{equation}
The distributions $p$ and $q$ are called \emph{mixed strategies}, while individual indexes $i$ and $j$ are called \emph{pure strategies}. We use $v(M)$ to denote the value in Equation \ref{eq:min_max}. (Recall that this can be interpreted as a game between a row player, or \emph{Minimizer}, and a column player, or \emph{Maximizer}. The min-max theorem states that the order in which the players reveal their strategies does not change the value of the game. It is easy to see that the second player can be restricted to pure strategies.)

We will consider a game played on a matrix of exponential size, and will be interested in near-optimal mixed strategies with succinct descriptions. This motivates the following definitions. A mixed strategy is $k$-\emph{uniform} if it is selected uniformly from a multiset of at most $k$ pure strategies. We use $P_k$ and $Q_k$ to denote the set of $k$-uniform strategies for the row player and the column player, respectively. For convenience, given a mixed row strategy $p$, we let $v(p) = v_M(p) = \max_{j \in [c]} \mathbb{E}_{i \sim p}[M(i,j)]$. Similarly, we use $v(q) = v_M(q) =  \min_{i \in [r]} \mathbb{E}_{j \sim q}[M(i,j)]$ for a column mixed strategy $q$. We say that a mixed strategy $u$ is $\delta$-\emph{optimal} if $|v(u) - v(M)| \leq \delta$. 

We will need the following ``efficient'' version of the min-max theorem.

\begin{theorem}[Small-Support Min-Max Theorem \citep{MR1274423, DBLP:conf/stoc/LiptonY94}]\label{t:efficient_min_max}
Let $M$ be a $r \times c$ real-valued matrix with entries in the interval $[-1,1]$. For every $\delta > 0$, if $k_r \geq 10 \ln (c)/\delta^2$ and $k_c \geq 10 \ln (r)/\delta^2$ then
$$
\min_{p \in P_{k_r}} v(p)  \;\leq\;  v(M) + \delta, \quad \text{and} \quad
\max_{q \in Q_{k_c}} v(q)  \;\geq\;  v(M) - \delta.
$$
\end{theorem} 
In other words, there are $\delta$-optimal strategies for the row and column players with relatively small support size.\\

\noindent \textbf{The PRF-Distinguisher Game}. Let $\mathcal{C}_n[s]$ be a circuit class and $\mathsf{Circuit}^\mathcal{O}[t]$ be an oracle circuit class,  with size parameters $s(n)$ and $t(n)$, respectively. We consider a $[-1,1]$-valued matrix $M = M^{\mathcal{C}_n[s], \mathsf{Circuit}^\mathcal{O}[t]}$, defined as follows. The rows of $M$ are indexed by Boolean functions in $\mathcal{C}_n[s]$, and the columns of $M$ are indexed by (single-output) \emph{deterministic} oracle circuits from $\mathsf{Circuit}^\mathcal{O}[t]$. In other words, such circuit have access to constants $0$ and $1$, compute according to the values of the oracle gates, and produce an output value in $\{0,1\}$. In order not to introduce further notation, we make the simplifying assumption that the negation of every circuit from $\mathsf{Circuit}^\mathcal{O}[t]$ is also in $\mathsf{Circuit}^\mathcal{O}[t]$. For $h \in \mathcal{C}_n[s]$ and $C^\mathcal{O} \in \mathsf{Circuit}^\mathcal{O}[t]$, we let 
$$
M(h,C^\mathcal{O}) \;\eqdef
\; C^h - \Pr_{f \sim \mathcal{F}_n}[C^f = 1],
$$
where $C^g \in \{0,1\}$ denotes the output of $C^\mathcal{O}$ when computing with oracle $\mathcal{O} = g$, for a fixed $g \colon \{0,1\}^n \to \{0,1\}$. We say that the matrix $M$ is the \emph{PRF-Distinguisher game} for $\mathcal{C}_n[s]$ and $\mathsf{Circuit}^\mathcal{O}[t]$. Observe that this is a finite matrix, for every choice of $n$.

Following our notation, we use $v(M)$ to denote the value of the game corresponding to $M$, which can be interpreted as follows. Let $p$ be a mixed strategy for the row player. In other words, $p$ is simply a distribution over functions from $\mathcal{C}_n[s]$. Consequently, to each row strategy $p$ we can associate a pair $(G_p,\mathcal{D}_p)$, where $p = \mathcal{D}_p$ and $G_p = \mathsf{Support}(\mathcal{D}_p)$, as in Definition \ref{d:prfs}.  On the other hand, a mixed strategy $q$ over the columns is simply a distribution over deterministic oracle circuits from $\mathsf{Circuits}^\mathcal{O}[t]$, which can be interpreted as a (non-constructive) randomized circuit $B^\mathcal{O}$. Under this interpretation, the value of the game when played with strategies $p$ and $q$ is given by
\begin{eqnarray}
\mathbb{E}_{h \sim p,\, C^\mathcal{O} \sim q}[M(h,C^\mathcal{O})] & = & \mathbb{E}_{h,C^\mathcal{O}}[C^h - \Pr_{f \sim \mathcal{F}_n}[C^f = 1]]  \nonumber \\
&  = &  \mathbb{E}_{h,C^\mathcal{O}}[C^h] - \Pr_{f, C^\mathcal{O} \sim q}[C^f = 1]  \nonumber \\
& = & \Pr_{g \sim \mathcal{D}_p, \,B^\mathcal{O}}[B^g = 1] - \Pr_{f \sim \mathcal{F}_n, \,B^\mathcal{O}}[B^f = 1], \nonumber 
\end{eqnarray}
which corresponds to the distinguishing probability in Definition \ref{d:prfs} without taking absolute values. But since we assumed that the circuits indexing the columns of $M$ are closed under complementation, it follows that the (global) value $v(M)$ of this game captures the security of PRFs from $\mathcal{C}_n[s]$ against $\mathsf{Circuit}^\mathcal{O}[t]$-distinguishers. (Notice though that this value does not take into account the samplability of the function families involved, nor the constructivity of the ensemble of distinguishers corresponding to a ``randomized'' oracle distinguisher in the argument above.)

\subsection{A (Non-Uniform) Converse to ``Learning Implies no PRFs''}\label{ss:no_prfs_learning}

We proceed with our original goal of establishing a converse of Proposition \ref{p:learning_implies_no_prfs}. Roughly speaking, we want to show that if every samplable function family from $\mathcal{C}_n$ can be distinguished from a random function (possibly by different distinguishers), then there is a single algorithm that learns every function in $\mathcal{C}_n$. Formally, what we get is a sequence of subexponential size (non-uniform) circuits learning $\mathfrak{C}$.

The proofs of Lemmas \ref{l:no_samplable_then_no_PRF} and \ref{l:no_PRF_then_distinguisher} below rely on Theorem \ref{t:efficient_min_max}.

\begin{lemma}[$\cancel{\exists}$ samplable PRF $\rightarrow$~$\cancel{\exists}$ PRF against ensembles of circuits] \label{l:no_samplable_then_no_PRF}
There exists a universal constant $c \in \mathbb{N}$ for which the following holds. Let $t(n) \geq n$, $s(n) \geq n$, $\delta(n) > 0$, and $\varepsilon(n) > 0$ be arbitrary functions. If there is no $O(t \cdot s \cdot 1/\delta)^c$-samplable pair $(\widetilde{G_n}, \widetilde{\mathcal{D}_n})$ that is a $(t(n),\varepsilon(n) + \delta(n))$-\emph{PRF} in $\mathcal{C}_n[s(n)]$, then there is no pair $(G_n, \mathcal{D}_n)$ with $G_n \subseteq \mathcal{C}_n[s(n)]$ that $\varepsilon(n)$-fools every ensemble of deterministic $\mathsf{Circuit}^\mathcal{O}[t(n)]$-circuits.  
\end{lemma}

\begin{proof}
We use Theorem \ref{t:efficient_min_max} to establish the contrapositive. Assume there exists a pair $(G_n, \mathcal{D}_n)$ where $\mathcal{D}_n$ is distributed over $G_n \subseteq \mathcal{C}_n[s(n)]$ such that for every distribution $q$ over $\mathsf{Circuit}^\mathcal{O}[t(n)]$ we have
$$
\left | \Pr_{g \sim \mathcal{D}_n, C^\mathcal{O} \sim q}[C^g = 1] - \Pr_{f \sim \mathcal{F}_n, C^\mathcal{O} \sim q}[C^f = 1]  \right  | \;\leq \; \varepsilon(n).
$$
Let $p = \mathcal{D}_n$, and observe that in the corresponding PRF-Distinguisher game we get $v_M(p) \leq \varepsilon(n)$. Consequently, $v(M) \leq \min_p v_M(p) \leq \varepsilon(n)$. It follows from Theorem \ref{t:efficient_min_max} and a bound on the number of columns of $M$ (similar to Lemma \ref{l:number_circuits}) that there exists a $k$-uniform distribution $\tilde{p}$ over functions in $\mathcal{C}_n[s(n)]$ with $k \leq O(\ln 2^{O(t \log t)}/\delta(n)^2) = O((t \log t)/\delta(n)^2)$ such that $v_M(\tilde{p}) \leq \varepsilon(n) + \delta(n)$.

In other words, each $f \in \mathsf{Support}(\tilde{p})$ is in $\mathcal{C}_n[s(n)]$, the support of this distribution contains at most $O((t \log t)/\delta(n)^2)$ different functions, and each such function can be encoded by a string of length $\mathtt{poly}(s(n))$ that describes the corresponding circuit. Using that $\tilde{p}$ is a $k$-uniform distribution, it is not hard to see that there exists a circuit $A \in \mathsf{Circuit}[S]$ with $A(\mathcal{U}_\ell) \equiv \tilde{p}$ for some $\ell \leq S$, where $S \leq \mathtt{poly}(t,s,1/\delta)$. Since every randomized circuit $B^{\mathcal{O}}$ can be seen as a distribution over deterministic oracle circuits, it follows that there is an $S$-samplable pair $(\widetilde{G_n}, \widetilde{\mathcal{D}_n})$ that is a $(t(n), \varepsilon(n) + \delta(n))$-PRF in $\mathcal{C}_n[s(n)]$. This completes the proof.
 \end{proof}

\begin{lemma}[$\cancel{\exists}$ PRF against ensembles of circuits $\rightarrow$~$\exists$ universal distinguisher] \label{l:no_PRF_then_distinguisher}
There exists a universal constant $c \in \mathbb{N}$ for which the following holds. Let $s(n) \geq n$, $t(n) \geq n$, $\varepsilon(n) > 0$, and $\gamma(n) > 0$ be arbitrary functions. If there is no pair $(G_n, \mathcal{D}_n)$ with $G_n \subseteq \mathcal{C}_n[s(n)]$ that $\varepsilon(n)$-fools every ensemble of deterministic $\mathsf{Circuit}^\mathcal{O}[t(n)]$-circuits, then there is a randomized oracle circuit $B^\mathcal{O} \in \mathsf{Circuit}^\mathcal{O}[O(t\cdot s \cdot 1/\gamma)^c]$ that distinguishes every such pair from a random function with advantage at least $\varepsilon(n) - \gamma(n)$. 
\end{lemma}

\begin{proof}
We rely on the classical min-max theorem and on Theorem \ref{t:efficient_min_max}. It follows from the assumption of the lemma that the corresponding PRF-Distinguisher game has value $v(M) \geq \varepsilon(n)$. By the min-max theorem, there is an ensemble of $\mathsf{Circuit}^\mathcal{O}[t(n)]$-circuits that distinguishes every pair  $(G_n, \mathcal{D}_n)$ satisfying $G_n \subseteq \mathcal{C}_n[s(n)]$ with advantage at least $\varepsilon(n)$. Applying Theorem \ref{t:efficient_min_max}, we obtain a $k$-uniform distribution $q$ over deterministic $\mathsf{Circuit}^\mathcal{O}[t(n)]$-circuits with distinguishing probability at least $\varepsilon(n) - \gamma(n)$ and support size at most $k = O(\ln 2^{O(s\log s)}/\gamma(n)^2) = O((s \log s)/\gamma(n)^2)$. Similarly to the proof of Lemma \ref{l:no_samplable_then_no_PRF}, this ensemble of circuits implies the existence of a single randomized oracle circuit $B^\mathcal{O} \in \mathsf{Circuit}^\mathcal{O}[O(s \cdot t \cdot 1/\gamma)^c]$ that distinguishes every pair $(G_n, \mathcal{D}_n)$ with $G_n \subseteq \mathcal{C}_n[s(n)]$ from a random function with advantage at least $\varepsilon(n) - \gamma(n)$. This completes the proof.
\end{proof}

Lemmas \ref{l:no_samplable_then_no_PRF} and \ref{l:no_PRF_then_distinguisher} hold for each value of $n$. The next lemma is a reduction involving different values of this parameter.

\begin{lemma}[$\exists$ universal distinguisher\underline{s} $\rightarrow$~$\exists$ learning circuit\underline{s}] \label{l:distinguisher_then_learning}
Assume that for every $k \geq 1$ and large enough $n$ there exists a randomized oracle circuit $B^\mathcal{O}_n$ in $\mathsf{Circuit}^\mathcal{O}[2^{O(n)}]$ that distinguishes every pair $(G_n, \mathcal{D}_n)$ with $G_n \subseteq \mathcal{C}_n[n^k]$ from a random function with advantage $\geq 1/40$. Then for every $\ell \geq 1$ and $\varepsilon > 0$ there is a non-uniform sequence of randomized oracle circuits in $\mathsf{Circuit}^\mathcal{O}[2^{n^{\varepsilon}}]$ that learn every function $f \in \mathcal{C}_n[n^\ell]$ to error at most $n^{-\ell}$.
\end{lemma}

\begin{proof}
This lemma is simply a (weaker) non-uniform version of the proof of Lemma \ref{l:distfastlearning} from Section \ref{s:proof_gap}. It is enough to use the sequence of randomized oracle circuits $B^\mathcal{O}_n$ as distinguishing circuits, and to observe that the statement of Theorem \ref{t:blackbox_ACzerop} holds with an arbitrarily small constant in the distinguishing probability. 
\end{proof}

Recall that $\mathfrak{C} = \{\mathcal{C}_n\}$ is an arbitrary typical circuit class. The main technical result of this section follows from Lemmas \ref{l:no_samplable_then_no_PRF}, \ref{l:no_PRF_then_distinguisher}, and \ref{l:distinguisher_then_learning} together with an appropriate choice of parameters. 

\begin{theorem}[No samplable PRFs in $\mathfrak{C}$ implies Learning $\mathfrak{C}$] \label{t:no_samp_PRFs_implies_learning}
If $t(n) \leq 2^{O(n)}$ and $c' \geq 1$ is a large enough constant, the following holds. Suppose that for every $k \geq 1$ each $O((t(n)\cdot n^k)^{c'})$-samplable pair $(G_n, \mathcal{D}_n)$ with $G_n \subseteq \mathcal{C}_n[n^k]$ can be distinguished from a random function by some randomized oracle circuit from $\mathsf{Circuit}^\mathcal{O}[t(n)]$ with advantage at least $1/10$. Then, for every $k\geq 1$, $\varepsilon >0$, and large enough $n$, there is a randomized oracle circuit from $\mathsf{Circuit}^\mathcal{O}[2^{n^\varepsilon}]$ that learns every function in $\mathcal{C}_n[n^k]$ to error at most $n^{-k}$. 
\end{theorem}

\begin{proof}
The existence of the learning circuit will follow if we can prove that the hypothesis of Lemma \ref{l:distinguisher_then_learning} is satisfied. Thus it is enough to argue that, for every $k \geq 1$ and large enough $n$, there is a (single) randomized oracle circuit $B^\mathcal{O}$ from $\mathsf{Circuit}^\mathcal{O}[2^{O(n)}]$ that distinguishes with advantage $\geq 1/40$ every pair $(G_n, \mathcal{D}_n)$ with $G_n \subseteq \mathcal{C}_n[n^k]$. In turn, this follows from Lemma \ref{l:no_PRF_then_distinguisher} for $s(n) = n^k$, $\varepsilon(n) = 1/20$, and $\gamma(n) = 1/40$ if there is no pair $(G_n, \mathcal{D}_n)$ with $G_n \subseteq \mathcal{C}_n[n^k]$ that $1/20$-fools every ensemble of deterministic oracle circuits from $\mathsf{Circuit}^\mathcal{O}[2^{O(n)}]$, for a slightly smaller constant in the latter exponent. But this is implied by the hypothesis of Theorem \ref{t:no_samp_PRFs_implies_learning} together with Lemma \ref{l:no_samplable_then_no_PRF}, instantiated with our value $t(n) \leq 2^{O(n)}$, $s(n) = n^k$,  $\varepsilon(n) = 1/20$, and $\delta(n) = 1/20$, provided that we take $c'$ sufficiently large. This completes the proof. 
\end{proof}

Dropping the samplability condition, we get the following weaker statement, which provides a converse of Proposition \ref{p:learning_implies_no_prfs} in the regime where $t(n)$ is exponential and $s(n)$ is polynomial.

\begin{corollary}[No PRFs in $\mathfrak{C}$ implies Learning $\mathfrak{C}$] \label{c:no_PRFs_implies_learning}
Let $t(n) \leq 2^{O(n)}$. If for every $k \geq 1$ and large enough $n$ there are no $(\mathsf{poly}(t(n)),1/10)$-pseudorandom function families in $\mathcal{C}_n[n^k]$, then for every $\varepsilon > 0$, $k \geq 1$, and large enough $n$, there is a randomized oracle circuit in $\mathsf{Circuit}^\mathcal{O}[2^{n^\varepsilon}]$ that $(n^{-k},1/n)$-learns every function in $\mathcal{C}_n[n^k]$.  
\end{corollary}

We observe that smaller time bounds $t(n)$ do not necessarily lead to smaller learning circuits, due to the running time of the black-box generator in Definition \ref{d:blackbox_generator} and Theorem \ref{t:blackbox_ACzerop}. However, a smaller $t(n)$ implies a weaker samplability condition in the statement of Theorem \ref{t:no_samp_PRFs_implies_learning}, which makes it stronger. A natural question is whether a more efficient distinguisher implies that larger circuits can be distinguished by subexponential size oracle circuits, in analogy to Lemma \ref{l:learningspeedup}. We mention that no simple reduction via padding seems to work, since a random function on $n$ bits mapped into a larger domain via projections is no longer a uniformly random function. Finally, the distinguishing advantage $1/10$ is arbitrary. Indeed, it can be assumed to be much lower, by following the estimates in the proof of Theorem \ref{t:blackbox_ACzerop}.

\begin{remark}
In order to prove Theorem \emph{\ref{t:no_samp_PRFs_implies_learning}}, we have made essential use of the ``efficient'' min-max theorem from \emph{\citep{MR1274423, DBLP:conf/stoc/LiptonY94}}, which guarantees the existence of near-optimal mixed strategies with simple descriptions. Unfortunately, this result does not provide an efficient algorithm to produce such strategies, which would lead to an equivalence between learning algorithms and the nonexistence of pseudorandom functions with respect to uniform computations. While there are more recent works that explore uniform versions of the min-max theorem \emph{(cf. \citep{DBLP:conf/crypto/VadhanZ13})}, they assume the existence of certain auxiliary algorithms in order to construct the near-optimal strategies, and it is unclear to us if they can be applied in the context of Theorem \emph{\ref{t:no_samp_PRFs_implies_learning}}. 
\end{remark}

\section{Lower Bounds from Nontrivial Algorithms}\label{s:proof_learning_lbs}

\begin{theorem} [Circuit lower bounds from nontrivial learning algorithms] \label{t:learning_lbs}~\\
Let $\mathfrak{C}$ be any typical circuit class. If for each $k$, $\mathfrak{C}[n^k]$ has non-trivial learning algorithms, then for each $k$, $\mathsf{BPTIME}(2^{O(n)}) \not \subseteq \mathfrak{C}[n^k]$.
\end{theorem}

Our proof of Theorem \ref{t:learning_lbs} relies on previous results relating randomized learning algorithms and lower bounds. The following connection was established  in \citep{DBLP:conf/coco/KlivansKO13}, using ideas from \citep{DBLP:journals/jcss/ImpagliazzoW01,DBLP:journals/jcss/FortnowK09}) and most crucially the construction of a
downward self-reducible and random self-reducible $\mathsf{PSPACE}$-complete language in
\citep{DBLP:journals/cc/TrevisanV07}.

\begin{theorem}[Connection between learning and lower bounds \citep{DBLP:conf/coco/KlivansKO13, DBLP:journals/jcss/FortnowK09, DBLP:journals/jcss/ImpagliazzoW01}]\label{t:KKO}~\\
There is a $\mathsf{PSPACE}$-complete language $L^\star \in \mathsf{DSPACE}(n)$ and a constant $b \in \mathbb{N}$ for which the following holds. Let $\mathfrak{C}$ be any typical circuit class, and $s \colon \mathbb{N} \to \mathbb{N}$ be any function with $s(n) \geq n$. If $\mathfrak{C}[s(n)]$ is learnable to error $\leq n^{-b}$ in time $T(n) \geq n$, then at least one of the following conditions hold:
\begin{itemize}
\item[\emph{(}i\emph{)}] $L^\star \notin \mathfrak{C}[s(n)]$.
\item[\emph{(}ii\emph{)}] $L^\star \in \mathsf{BPTIME}(\mathsf{poly}(T(n)))$. 
\end{itemize}
\end{theorem}

A self-contained proof of a generalization of Theorem \ref{t:KKO} is presented in Section \ref{s:karp_lipton_collapses}. We will also need a consequence of the following diagonalization lemma.

\begin{lemma}[A nonuniform almost everywhere hierarchy for space complexity]
\label{l:iohardness}~\\
Let $S, S' \colon \mathbb{N} \rightarrow \mathbb{N}$ be space-constructible functions such that $S(n) = o(S'(n)), S(n) = \Omega(\log n)$ and $S'(n) = o(2^n)$. There is a language $L \in \mathsf{DSPACE}(S')$ such that $L \not \in \mathtt{i.o.}\mathsf{DSPACE}(S)/S$.
\end{lemma}

\begin{proof}
This is a folklore argument. We define a space-bounded Turing machine $M$ operating in space $S'$ such that $L(M) \not \in \mathtt{i.o.}\mathsf{DSPACE}(S)/S$. On inputs of length $n$, $M$ uses the space-constructibility of $S'$ to compute $S'(n)$ in unary using space $O(S'(n))$. It marks out $S'(n)$ cells on each of its tapes, and if at any point in its computation, it reads an unmarked cell, it halts and rejects. Thus, on any input of length $n$, $M$ uses space $O(S'(n))$. $M$ also computes and stores $S(n)$ on one of its tapes.

The high-level intuition is that $M$ diagonalizes against machine $M_i$ with advice $z$, for each $1 \leq i \leq \log n$ and advice $z \in \{0,1\}^{S(n)}$. In particular, for any fixed $i$ and large enough $n$, $M$ diagonalizes against $M_i$ with any advice $z \in \{0,1\}^{S(n)}$, and hence $L(M)$ satisfies the conclusion of the Lemma.

 By a counting argument, there are at most $\log n \cdot 2^{S(n)}$ truth tables of Boolean functions $f$ on $n$ bits such that $f$ is computed by a machine $M_i$ with $1 \leq i \leq \log n$ operating in space $S(n)$ and using $S(n)$ bits of advice. Thus, since $S(n) = o(2^n)$, for large enough $n$, by the pigeon-hole principle there exists a Boolean function $f'\colon \{0,1\}^n \to \{0,1\}$ which is $0$ on all but the first $\log n + S(n)$ inputs of length $n$, such that $f'$ is not computed by machine $M_i$ with advice $z$ for any $i$ with $1 \leq i \leq \log n$ and $z \in \{0,1\}^{S(n)}$. 

$M$ computes such a function iteratively as follows. It processes the inputs of length $n$ in lexicographic order. At stage $i+1$, where $i \geq 0$, $M$ has stored a binary string $y_i$ of length $i$ representing the values of $f'$ on the first $i$ inputs of length $n$, and $M$ is trying to determine $f'$ on the $(i+1)$-th input of length $n$. For each machine $M_i$, $1 \leq i \leq \log n$, and each advice string $z$ for $M_i$ of length $S(n)$, by simulating those $M_i$'s with advice $z$ which do not use space more than $S(n)$ on any of the first $i$ inputs, $M$ determines if the truth table of $M_i$ with advice $z$ is consistent with $y_i$ on the first $i$ inputs. Call such a pair $(i,z)$ a consistent machine-advice pair at stage $i+1$. $M$ sets $f'$ to $0$ for the $(i+1)$-th string if a minority of consistent machine-advice pairs halt with $0$ on the $(i + 1)$-th string, and to $1$ otherwise. Determining whether a minority of consistent machine-advice pairs halt with $0$ on the $(i+1)$-th string can be done by merely keeping a count of how many consistent machine-advice pairs halt with $0$, and how many halt with $1$, which only requires space $O(S(n))$. Note that using the minority value cuts down the number of consistent machine-advice pairs for the next stage by at least a factor of half. This implies that at stage $\log n + S(n)$, there are no consistent machine-advice pairs left, and hence $M$ has successfully diagonalized. It is not hard to see that the overall simulation can be carried out in space $O(S(n))$, using the fact that $S(n) = \Omega(\log n)$. 
\end{proof}

\begin{corollary} [Diagonalizing in uniform space against non-uniform circuits]
\label{c:iockthard}~\\Let $S_1, S_2 \colon \mathbb{N} \rightarrow \mathbb{N}$ be space-constructible functions such that $S_2(n)^2 = o(S_1(n))$, $S_2(n) = \Omega(\log n)$ and $S_1(n) = o(2^n)$. There is a language $L \in \mathsf{DSPACE}(S_1)$ such that $L \not \in \mathtt{i.o.}\mathsf{Circuit}[S_2]$. In particular, for each $k$, there is a language $L_k \in \mathsf{PSPACE}$ such that $L_k \not \in \mathsf{Circuit}[n^k]$.
\end{corollary}

\begin{proof}
Corollary \ref{c:iockthard} follows from Lemma \ref{l:iohardness} using the fact that $\mathsf{Circuit}[S] \subseteq \mathsf{DSPACE}(S^2)/S^2$. 
\end{proof}

In fact, a tighter simulation holds, and therefore a tighter separation in Corollary \ref{c:iockthard}, but we will not need this for our purposes. We are now ready to prove Theorem \ref{t:learning_lbs}.

\begin{proof}[Proof of Theorem \emph{\ref{t:learning_lbs}}]
Let $\mathfrak{C}$ be a typical circuit class. By assumption, $\mathfrak{C}[n^k]$ has a non-trivial learner for each $k > 0$. Since $\mathfrak{C}$ is typical, we can use Lemma \ref{l:learningspeedup} to conclude that for each $\varepsilon > 0$ and for each $k > 0$, $\mathfrak{C}[n^k]$ is strongly learnable in time $2^{n^{\varepsilon}}$. 

Let $L^{\star}$ be the $\mathsf{PSPACE}$-complete language in the statement of Theorem \ref{t:KKO}. Using Theorem \ref{t:KKO} and the conclusion of the previous paragraph, we have that at least one of the following is true: (1) For all $k$, $L^{\star} \not \in \mathfrak{C}[n^k]$, or (2) For all $\varepsilon > 0$, $L^{\star} \in \mathsf{BPTIME}(2^{n^{\varepsilon}})$.

In case (1), since $L^{\star} \in \mathsf{DSPACE}(n) \subseteq \mathsf{DTIME}(2^{O(n)})$, we have that for each $k > 0$, $\mathsf{DTIME}(2^{O(n)}) \not \subseteq \mathfrak{C}[n^k]$, and hence also $\mathsf{BPTIME}(2^{O(n)}) \not \subseteq \mathfrak{C}[n^k]$.

In case (2), we have that $L^{\star} \in \mathsf{BPTIME}(2^{n^{\varepsilon}})$ for every $\varepsilon > 0$. Since $L^{\star}$ is $\mathsf{PSPACE}$-complete, this implies that the language $L_k$ in the statement of Corollary \ref{c:iockthard} is also in
$\mathsf{BPTIME}(2^{n^{\varepsilon}})$, for every fixed $\varepsilon > 0$ and $k \in \mathbb{N}$. (Here the polynomial blowup of instance size in the reduction from $L_k$ to $L^{\star}$ is taken care of by the universal quantification over $\varepsilon$.) In particular, we have $L_k \in \mathsf{BPTIME}(2^n)$, for every $k$. Since for any typical circuit class we have $\mathfrak{C}[n^k] \subseteq \mathsf{Circuit}[n^c]$ for a large enough $c = c(k)$, there is a language $L_c \in \mathsf{BPTIME}[2^n]$ such that $L_c \notin \mathfrak{C}[n^k]$. This establishes the desired result.
\end{proof}

We mention for completeness that the same approach yields a trade-off involving the running time of the learning algorithm and its accuracy in the hypothesis of Theorem \ref{t:learning_lbs}.

\begin{theorem}[Trade-off between error and running time] \label{t:tradeoff_learning_lbs}~\\
Let $\mathfrak{C}$ be a typical circuit class, and $\gamma \colon \mathbb{N} \to (0,1/2] \,\cap\, \mathbb{Q}$ be a polynomial time computable function. If for each $k$, $\mathfrak{C}[n^k]$ can be learned with advantage at least $\gamma(n)$ in time $\gamma(n)^2 \cdot 2^n/n^{\omega(1)}$, then for each $k$, $\mathsf{BPTIME}[2^{O(n)}] \nsubseteq \mathfrak{C}[n^k]$.
\end{theorem}

\begin{proof}(Sketch) The proof is entirely analogous to the argument in Theorem \ref{t:learning_lbs}. It is enough to observe that such learning algorithms yield the complexity distinguishers required in Lemma \ref{l:learningspeedup} via a natural generalization of the proof of Lemma \ref{l:learningdist}. The quantitative trade-off between accuracy and running time is a consequence of Lemma \ref{l:random_functions}.
\end{proof}

\begin{remark}
Observe that as the advantage $\gamma(n)$ approaches $2^{-n/2}$ from above, the running time required in Theorem \emph{\ref{t:tradeoff_learning_lbs}} becomes meaningless. This quantitative connection between $\gamma(n)$ and the running time is not entirely unexpected. On the one hand, it is a consequence of the concentration bound, which is essentially optimal. But also note that every function $g \colon \{0,1\}^n \to \{0,1\}$ can be approximated with advantage $\geq 2^{-n/2}$ by a parity function \emph{(}or its negation\emph{)}, and that heavy fourier coefficients corresponding to such parity functions can be found using membership queries by the Goldreich-Levin Algorithm \emph{(}see e.g. \emph{\citep{DBLP:books/daglib/0033652}}\emph{)}. 
\end{remark}

We can expand the scope of application of Theorem \ref{t:learning_lbs}, using a win-win argument. The more general result below applies to subclasses of Boolean circuits satisfying the very weak requirement that they are closed under projections, rather than just to the more specialized ``typical'' classes.

\begin{theorem} [Lower bounds from non-trivial learning algorithms for subclasses of circuits]
\label{t:learning_lbs_general}~\\
Let $\mathfrak{C}$ be any subclass of Boolean circuits closed under projections. If for each $k$, $\mathfrak{C}[n^k]$ has non-trivial learning algorithms, then for each $k$, $\mathsf{BPTIME}(2^{O(n)}) \not \subseteq \mathfrak{C}[n^k]$.
\end{theorem}

\begin{proof}
Consider the Circuit Value Problem (CVP), which is complete for $\mathsf{Circuit}[\mathsf{poly}]$ under polynomial size projections. Either CVP is in $\mathfrak{C}[n^c]$ for some fixed $c$, or it is not. If it is not, then we have the desired lower bound for CVP and hence also for the class $\mathsf{BPTIME}(2^{O(n)})$, which contains this problem. If CVP is in $\mathfrak{C}[n^c]$, then since CVP is closed under poly-size projections, we have by completeness and the assumption of the theorem that for each $k$, $\mathsf{Circuit}[n^k]$ has non-trivial learning algorithms. Now applying Theorem \ref{t:learning_lbs}, we have that for each $k$, $\mathsf{BPTIME}(2^{O(n)}) \not \subseteq \mathsf{Circuit}[n^k]$, which implies that $\mathsf{BPTIME}(2^{O(n)}) \not \subseteq \mathfrak{C}[n^k]$, since $\mathfrak{C}$ is a subclass of Boolean circuits.
\end{proof}

\begin{remark}
Observe that it is possible to instantiate Theorem \emph{\ref{t:learning_lbs_general}} for very particular classes such as $\mathsf{AND} \circ \mathsf{OR} \circ \mathsf{THR}$ circuits, and that the lower bound holds for exactly the same circuit class. In particular, there is no circuit depth blow-up. 
\end{remark}

We get an improved lower bound consequence for the circuit class $\mathsf{ACC}^0$, but under the assumption that subexponential-size circuits are non-trivially learnable. (Recall that there are satisfiability algorithms for such circuits with non-trivial running time \citep{DBLP:journals/jacm/Williams14}.)

\begin{theorem} [Improved lower bounds from non-trivial learning algorithms for $\mathsf{ACC}^0$]
\label{t:learning_lbs_ACC}~\\
If for every depth $d \in \mathbb{N}$ and modulo $m \in \mathbb{N}$ there is $\varepsilon > 0$ such that $\mathsf{ACC}^0_{d,m}[2^{n^{\varepsilon}}]$ has non-trivial learning algorithms, then $\mathsf{REXP} \not \subseteq \mathsf{ACC}^0[\mathsf{poly}]$.
\end{theorem}

\begin{proof}
Under the assumption on learnability, using Lemma \ref{l:learningspeedup}, we have that for each $k > 0$, $\mathsf{ACC}^0[n^k]$ has strong learners running in time $2^{\mathsf{polylog}(n)}$. Now applying Theorem \ref{t:KKO}, we have that at least one of the following is true for the $\mathsf{PSPACE}$-complete language $L^{\star}$ in the statement of the theorem: (1) $L^{\star} \not \in \mathsf{ACC}^0[n^k]$ for any $k$, or (2) $L^{\star} \in \mathsf{BPQP}$, where $\mathsf{BPQP}$ is bounded error probabilistic quasi-polynomial time.

In case (1), we have that $L^{\star} \not \in \mathsf{ACC}^0[\mathsf{poly}]$, and are done as in the proof of Theorem \ref{t:learning_lbs}.

In case (2), by $\mathsf{PSPACE}$-completeness of $L^{\star}$, we have that $\mathsf{PSPACE} \subseteq \mathsf{BPQP}$. This implies that $\mathsf{NP} \subseteq \mathsf{BPQP}$, and hence that $\mathsf{NP} \subseteq \mathsf{RQP}$, where $\mathsf{RQP}$ is probabilistic quasi-polynomial time with one-sided error. The second implication follows using downward self-reducibility to find a witness for SAT given the assumption that SAT is in $\mathsf{BPQP}$, thus eliminating error on negative instances. Now $\mathsf{NP} \subseteq \mathsf{RQP}$ implies $\mathsf{NEXP} = \mathsf{REXP}$, using a standard translation argument. Williams showed that $\mathsf{NEXP} \not \subseteq \mathsf{ACC}^0[\mathsf{poly}]$, and so it follows that $\mathsf{REXP} \not \subseteq \mathsf{ACC}^0[\mathsf{poly}]$, as desired.
\end{proof}

More generally, the same argument combined with the connection between non-trivial satisfiability algorithms and circuit lower bounds \citep{DBLP:journals/jacm/Williams14} imply the following result.

\begin{corollary}[Lower bounds from learning and satisfiability]\label{c:combination}
Let $\mathfrak{C}$ be any typical circuit class. Assume that for each $k$, $\mathfrak{C}[n^k]$ admits a non-trivial satisfiability algorithm, and that for some $\varepsilon > 0$, $\mathfrak{C}[2^{n^\varepsilon}]$ admits a non-trivial learning algorithm. Then $\mathsf{REXP} \nsubseteq \mathfrak{C}[\mathsf{poly}]$.
\end{corollary}

Recall that randomized learning algorithms and $\mathsf{BPP}$-natural properties are strongly related by results of \citep{CIKK16}. We can give still stronger lower bound conclusions from assumptions about $\mathsf{P}$-natural proofs. The idea is to combine the arguments above with an application of the easy witness method of Kabanets \citep{DBLP:journals/jcss/Kabanets01}.

\begin{theorem} [Improved lower bounds from natural proofs]
\label{t:natural_lbs}~\\
Let $\mathfrak{C}$ be any subclass of Boolean circuits closed under projections. If there are $\mathsf{P}$-natural proofs useful against $\mathfrak{C}[2^{n^{\varepsilon}}]$ for some $\varepsilon > 0$, then $\mathsf{ZPEXP} \not \subseteq \mathfrak{C}[\mathsf{poly}]$. 
\end{theorem}

The following immediate consequence is of particular interest.

\begin{corollary}[$\mathsf{ACC}^0$ lower bounds from natural proofs]\label{c:williams_natural_lb}~\\
If for some $\delta > 0$ there are $\mathsf{P}$-natural proofs against $\mathsf{ACC}^0[2^{n^\delta}]$ then $\mathsf{ZPEXP} \nsubseteq \mathsf{ACC}^0[\mathsf{poly}]$.
\end{corollary}

In order to prove Theorem \ref{t:natural_lbs}, we will need the following lemma.

\begin{lemma} [Simulating bounded error with zero error given natural proofs]
\label{l:natural_easywitness}
Suppose there is a constant $\delta > 0$ such that there are $\mathsf{P}$-natural proofs against $\mathsf{Circuit}[2^{n^{\delta}}]$. Then $\mathsf{BPEXP} = \mathsf{ZPEXP}$.
\end{lemma}

\begin{proof}
Note that zero-error probabilistic time is trivially contained in bounded-error probabilistic time, so we only need to show that $\mathsf{BPEXP} \subseteq \mathsf{ZPEXP}$ under the assumption. We will in fact show that $\mathsf{BPP} \subseteq \mathsf{ZPQP}$, where $\mathsf{ZPQP}$ is zero-error bounded probabilistic quasi-polynomial time. The desired conclusion follows from this using a standard translation argument.

By assumption, there is a natural property $\mathfrak{R}$ useful against $\mathsf{Circuit}[2^{n^{\delta}}]$ for some constant $\delta > 0$, such that $L_{\mathfrak{R}} \in \mathsf{P}$. Let $M$ be any machine operating in bounded-error probabilistic time $n^d$ for some $d > 0$. We define a zero-error machine $M'$ deciding $L(M)$ in quasi-polynomial time as follows. On input $x$ of length $n$, $M'$ guesses a random string $r$ of size $2^{\mathsf{log}(n)^{d'}}$, where $d'$ is a large enough constant to be defined later. It then checks if $r \in L_{\mathfrak{R}}$ or not, using the polynomial-time decision procedure for the natural property $\mathfrak{R}$. If not, it outputs `?' and halts. If it does, it runs the procedure of Theorem \ref{t:hardness_PRG} on input $n^{2d}$ in unary and $r$, to obtain the range of a $(\mathsf{polylog}(n), 1/n^{2d})$ PRG against $\mathsf{Circuit}[n^{2d}]$. Since $r \in L_{\mathfrak{R}}$, Theorem \ref{t:hardness_PRG} applies, and the output of the procedure is guaranteed to be the range of such a PRG. $M'$ then runs $M$ on $x$ independently with each element of the range of the PRG used as randomness, and takes the majority vote. This is guaranteed to be correct when $r \in L_{\mathfrak{R}}$, which happens with probability at least $1/2$ by the density property of $\mathfrak{R}$. Thus $M'$ is a zero-error machine, and it is clear that $M'$ can be implemented in quasi-polynomial time.
\end{proof}

\begin{proof}[Proof of Theorem \emph{\ref{t:natural_lbs}}]
We proceed as in the proof of Theorem \ref{t:learning_lbs_general}. Either CVP is in $\mathfrak{C}[\mathsf{poly}]$, or it is not. If not, we have the desired lower bound for CVP, and hence for $\mathsf{ZPEXP}$, which contains this problem. 

On the other hand, if CVP is in $\mathfrak{C}[\mathsf{poly}]$, we have that CVP is in $\mathfrak{C}[n^k]$ for some $k > 0$. By the completeness of CVP for poly-size circuits under poly-size projections, and the closure of $\mathfrak{C}$ under projections, we have that $\mathsf{Circuit}[n] \subseteq \mathfrak{C}[n^k]$ for some $k > 0$, and hence by a standard translation argument, we have that $\mathsf{Circuit}[2^{n^\delta}] \subseteq \mathfrak{C}[2^{n^{\varepsilon}}]$ for any $\delta < \varepsilon$. By assumption, we have $\mathsf{P}$-natural properties useful against $\mathfrak{C}[2^{n^{\varepsilon}}]$ and hence we also have $\mathsf{P}$-natural properties useful against $\mathsf{Circuit}[2^{n^{\delta}}]$ for any $\delta < \varepsilon$. Now, applying Lemma \ref{l:natural_easywitness}, we get $\mathsf{BPEXP} = \mathsf{ZPEXP}$.

We argue next that under the existence of $\mathsf{P}$-natural properties useful against $\mathsf{Circuit}[2^{n^{\delta}}]$ for a fixed $\delta > 0$, we also have $\mathsf{EXPSPACE} = \mathsf{BPEXP}$. The mentioned hypothesis implies that there exist complexity distinguishers against $\mathsf{Circuit}[2^{n^{\delta}}]$ running in deterministic time $2^{O(n)}$ (the acceptance probability can be amplified using truth-table concatenation). As a consequence, Lemma \ref{l:distfastlearning} provides strong learning algorithms for $\mathsf{Circuit}[\mathsf{poly}]$ running in quasi-polynomial time. By Theorem \ref{t:KKO}, either $\mathsf{PSPACE} \nsubseteq \mathsf{Circuit}[\mathsf{poly}]$, and we are done, or $\mathsf{PSPACE} \subseteq \mathsf{BPQP}$. Now a standard upward translation gives $\mathsf{EXPSPACE} \subseteq \mathsf{BPEXP}$, which shows that $\mathsf{EXPSPACE} = \mathsf{BPEXP}$.

Altogether, we have $\mathsf{EXPSPACE} = \mathsf{ZPEXP}$. Now this collapse and Corollary \ref{c:iockthard} with $S_1(n) = 2^{\sqrt{n}}$ and $S_2(n) = n^{\log n}$ yield a language $L \in \mathsf{ZPEXP}$ such that $L \notin \mathsf{Circuit}[\mathsf{poly}]$, which completes the proof of Theorem \ref{t:natural_lbs}.
\end{proof}

Recall that the existence of \emph{useful} properties against a circuit class $\mathfrak{C}$ is essentially equivalent to the existence of non-deterministic exponential time lower bounds against $\mathfrak{C}$ \citep{DBLP:journals/siamcomp/Williams16, oliveirathesis}. We do not expect a similar equivalence in the case of \emph{natural} properties and lower bounds for probabilistic exponential time. The results described in this section show that natural properties imply such lower bounds. However, if the other direction were true, then any lower for $\mathfrak{C}$ with respect to probabilistic exponential time classes would also provide a non-trivial learning algorithm for $\mathfrak{C}$. In particular, since we believe in separations such as $\mathsf{EXP} \nsubseteq \mathsf{Circuit}[\mathsf{poly}]$, this would imply via the Speedup Lemma that polynomial size circuits can be learned in sub-exponential time, which seems unlikely.

\section{Karp-Lipton Collapses for Probabilistic Classes}\label{s:karp_lipton_collapses}

\subsection{A Lemma About Learning with Advice}

In this section we will need some notions of computability with advice. While this is a standard notion, we provide some definitions, as bounded-error randomized algorithms taking advice can be defined in different ways. 

Recall that an {\it advice-taking} Turing machine is a Turing machine equipped with an extra tape, the advice tape. At the start of any computation of an advice-taking Turing machine, the input is present on the input tape of the machine and a string called the ``advice'' on the advice tape of the machine, to both of which the machine has access.

\begin{definition} [Probabilistic time with advice]
\label{d:BPTIMEAdv}
Let $T \colon \mathbb{N} \rightarrow \mathbb{N}$ and $a \colon \mathbb{N} \rightarrow \mathbb{N}$ be functions. $\mathsf{BPTIME}(T)/a$ is the class of languages $L \subseteq \{0,1\}^*$ for which there is an advice-taking probabilistic Turing machine $M$ which always halts in time $T(n)$ and a sequence $\{z_n\}_{n \in \mathbb{N}}$ of strings such that\emph{:}
\begin{enumerate}
\item[\emph{1.}] For each $n$, $|z_n| \leq a(n).$
\item[\emph{2.}] For any input $x \in L$ such that $|x| = n$, $M$ accepts $x$ with probability at least $2/3$ when using advice string $z_n$.
\item[\emph{3.}] For any input $x \not \in L$ such that $|x| = n$, $M$ rejects $x$ with probability at least $2/3$ when using advice string $z_n$.
\end{enumerate}
\end{definition}

Note that in the above definition, there are no guarantees on the behaviour of the machine for advice strings other than the ``correct'' advice string $z_n$. In particular, for an arbitrary advice string, the machine does not have to satisfy the bounded-error condition on an input, though it does have to halt within time $T$.

The notion of resource-bounded computation with advice is fairly general and extends to other models of computation, such as deterministic computation and computation of non-Boolean functions. These extensions are natural, and we will not define them formally.

A slightly less standard notion of computation with advice is learnability with advice. We extend Definition \ref{d:learning} to capture learning with advice by giving the learning algorithm an advice string, and only requiring the learning algorithm to work correctly for a ``correct'' advice string of the requisite length.

We will also need the standard notions of downward self-reducibility and random self-reducibility.

\begin{definition} [Downward self-reducibility]
\label{d:dsr}
A function $f \colon \{0,1\}^* \to \{0,1\}$ is said to be downward self-reducible if there is a polynomial-time oracle procedure $A^f(x)$ such that\emph{:}
\begin{enumerate}
 \item[\emph{1.}] On any input $x$ of length $n$, $A^f(x)$ only makes queries of length $< n$.
 \item[\emph{2.}] For every input $x$, $A^f(x) = f(x)$.
\end{enumerate}
\end{definition} 
 
\begin{definition} [Random self-reducibility]
\label{d:rsr}
A function $f \colon \{0,1\}^* \to \{0,1\}$ is said to be random self-reducible if there are constants $k, \ell \geq 1$ and polynomial-time computable functions $g \colon \{0,1\}^* \to \{0,1\}^*$ and $h \colon \{0,1\}^* \to \{0,1\}$ satisfying the following conditions\emph{:}

\begin{enumerate}
\item[\emph{1.}] For large enough $n$, for every $x \in \{0,1\}^n$ and for each $i \in \mathbb{N}$ such that $1 \leq i \leq n^k$, $g(i,x,r) \sim U_n$ when $r \sim U_{n^{\ell}}$.
\item[\emph{2.}] For large enough $n$ and for every function $\tilde{f}_n \colon \{0,1\}^n \to \{0,1\}$ that is $(1/n^k)$-close to $f$ on $n$-bit strings, for every $x \in \{0,1\}^n$\emph{:}
$$f(x) = h(x, r, \tilde{f}_n(g(1,x,r)), \tilde{f}_n(g(2,x,r)), \ldots , \tilde{f}_n(g(n^k, x, r)))$$ 
with probability $\geq 1-2^{-2n}$ when $r \sim U_{n^{\ell}}$.
\end{enumerate}
\end{definition}

\begin{theorem} [A special $\mathsf{PSPACE}$-complete function \citep{DBLP:journals/cc/TrevisanV07}]
\label{t:PSPACECompRsrDsr}~\\
There is a $\mathsf{PSPACE}$-complete function $f_{\mathsf{TV}} \colon \{0,1\}^* \to \{0,1\}$ such that $f_{\mathsf{TV}}$ is downward self-reducible and random self-reducible.
\end{theorem}

Below we consider the learnability of the class of Boolean functions $\{f_{\mathsf{TV}}\}$ that contains only the function $f_{\mathsf{TV}}$.

\begin{lemma} [Learnability with advice for $\mathsf{PSPACE}$ implies randomized algorithms]
\label{l:PSPACELearnSolve}~\\
For any polynomial-time computable non-decreasing function $a \colon \mathbb{N} \rightarrow \mathbb{N}$ with $a(n) \leq n$, and for any non-decreasing function $T \colon \mathbb{N} \rightarrow \mathbb{N}$ such that $n \leq T(n) \leq 2^n$, if $\{f_{\mathsf{TV}}\}$ is strongly learnable in time $T$ with $a$ bits of advice, then $f_{\mathsf{TV}}$ is computable in bounded-error probabilistic time $T(n)^2 \cdot 2^{a(n)} \cdot n^{O(1)}$, and hence $\mathsf{PSPACE} \subseteq \mathsf{BPTIME}(T(\mathsf{poly}(n))^2 \cdot 2^{a(\mathsf{poly}(n))} \cdot \mathsf{poly}(n))$.
\end{lemma}

\begin{proof} 
The argument is based on and extends ideas from \citep{DBLP:conf/coco/KlivansKO13, DBLP:journals/jcss/FortnowK09, DBLP:journals/cc/TrevisanV07, DBLP:journals/jcss/ImpagliazzoW01}. Recall that $f_{\mathsf{TV}}$ is the same Boolean function as in the statement of Theorem \ref{t:PSPACECompRsrDsr}. As stated there, this function is downward self-reducible and random self-reducible. Now suppose $\{f_{\mathsf{TV}}\}$ is learnable in time $T$ with $a$ bits of advice. We design a probabilistic machine $M$ solving $f_{\mathsf{TV}}$ on inputs of length $n$ with bounded error in time $T(n)^2 \cdot 2^{a(n)} \cdot n^{O(1)}$. The addition inclusion in the statement of Lemma \ref{l:PSPACELearnSolve} follows from the completeness of $f_\mathsf{TV}$.

Let $x \in \{0,1\}^n$ be the input to $M$. Let $A_{\mathsf{learn}}$ be a $(1/n^{4k}, 1/2^{2n})$-learner for $\{f_{\mathsf{TV}}\}$ that takes $a(n)$ bits of advice and runs in time $T(n)$.\footnote{In this argument, we do not care about $\mathsf{poly}(n)$ multiplicative factors applied to the final running time, so we can assume the failure probability of the learner to be exponentially small by amplification. This is a standard argument, and we refer to \citep{KearnsVazirani:94} for more details.} Here $k$ is the exponent in the number of queries in the random self-reduction for $f_{\mathsf{TV}}$ given by Theorem \ref{t:PSPACECompRsrDsr}.\\

\noindent \textbf{Overview.} The plan of the proof is that $M$ will use the advice-taking learner to inductively produce, with high probability, circuits computing $f_{\mathsf{TV}}$ correctly on inputs of length $1 \ldots n$. The crucial aspect is not to allow the size of these circuits to grow too large. There will be $n$ phases in the operation of $M$ -- during Phase $i$, $M$ will produce with high probability a randomized circuit computing $f_{\mathsf{TV}}$ on inputs of length $i$. 

Each phase consists of 2 parts. In Part 1 of Phase $i$, $M$ computes, for each possible advice string $z$ of length $a(i)$ that can be fed to the advice-taking learner on input $1^i$, a {\it candidate deterministic circuit} $C_i^z$ on $i$-bit inputs of size at most $T(i)$. In order for $M$ to do this, it uses the properties of the learner, as well as the circuits for smaller lengths that have already been computed. The only guarantee on the candidate circuits is that at least one of them is a approximately correct circuit for $f_{\mathsf{TV}}$ at length $i$, in the sense that it is correct on most inputs of this length. In Part 2 of Phase $i$, $M$ uses the random self-reducibility and downward self-reducibility of $f_{\mathsf{TV}}$ to select the ``best-performing'' candidate among these circuits and compute a ``correction'' $C_i$ of the best-performing circuit. The circuit $C_i$ will have size $T(i) \cdot \mathsf{poly}(n)$, and with high probability, it will be a randomized circuit that computes $f_{\mathsf{TV}}$ correctly on all $i$-bit inputs, in the sense that on each such string it is correct with overwhelming probability over its internal randomness. At the end of Phase $n$, $M$ evaluates the circuit $C_n$ on $x$ and outputs the answer.\\

We now give the details of how Part 1 and Part 2 work for each phase. We will then need to argue that $M$ is correct, and that it is as efficient as claimed. Phase $1$, which is the base case for $M$'s inductive operation, is trivial. The circuit $C_1$ computing $f_{\mathsf{TV}}$ correctly on inputs of length $1$ is simply hard-coded into $M$.

Now let $i > 1$ be an integer. We describe how Part 1 and Part 2 of Phase $i$ work, assuming inductively that $M$ already has stored in memory a sequence of circuit $\{C_j\}$, for $1 \leq j \leq i-1$, such that for each such $j$, $C_j$ has size at most $T(j) \cdot \mathsf{poly}(n)$, and with all but exponentially small probability, computes $f_{\mathsf{TV}}$ correctly on each input of length $j$.\\

\noindent \textbf{Part 1.} $M$ first uses the polynomial-time computability of $a$ to compute $a(i)$. It then cycles over strings $z \in \{0,1\}^{a(i)}$, and for each string $z$ it does the following.
It simulates $A_{\mathsf{learn}}(1^i)$ with advice $z$. Each time $A_{\mathsf{learn}}$ makes a membership query of length $i$, $M$ answers the membership query using the downward self-reducibility of $f_{\mathsf{TV}}$ as follows. If the downward self-reduction makes a query of length $j < i$, $M$ answers it by running the stored circuit $C_{j}$ on the corresponding query.

If $A_{\mathsf{learn}}(1^i)$ with advice $z$ does not halt with an output that is a circuit on $i$ bits, $M$ sets $C_i^z$ to be a trivial circuit on $i$ bits, say the circuit that always outputs 0. Otherwise $M$ sets $C_i^z$ to be the circuit output by the learning algorithm.  Since $A_{\mathsf{learn}}$ is guaranteed to halt in time $T(i)$ for every advice string, the circuit $C_i^z$ has size at most $T(i)$.\\

\noindent \textbf{Part 2.} $M$ samples strings $y_1, \ldots, y_t$, where $t = n^{10k}$, uniformly and independently at random amongst $i$-bit strings. It computes ``guesses'' $b_1, \ldots, b_t \in \{0,1\}$ for the values of $f_{\mathsf{TV}}$ on these inputs by running the downward self-reducibility procedure for $f_{\mathsf{TV}}$, and answering any queries of length $j < i$ using the stored circuits $C_j$. Then, for each advice string $z$, it simulates $C_i^z$ on each input $y_{\ell}$, where  $1 \leq \ell \leq t$, and computes the fraction $\rho_z$ of inputs $y_{\ell}$ for which $C_i^z(y_{\ell}) = b_{\ell}$. Let $z_{\mathsf{max}}$ be the advice string $z$ for which $\rho_{z_{\mathsf{max}}}$ is maximum among all such advice strings. Let $D_i$ be the (deterministic) circuit $C_i^{z_{\mathsf{max}}}$. $M$ produces a randomized circuit $C_i$ from $D_i$ as follows. $C_i$ applies the random self-reduction procedure for $f_{\mathsf{TV}}$ $O(n)$ times independently, using the circuit $D_i$ to answer the random queries to $f_{\mathsf{TV}}$, and outputs the majority answer of these runs. Note that $C_i$ can easily be implemented in size $T(i) \cdot \mathsf{poly}(n)$, using the fact that the random self-reduction procedure runs in polynomial time. (We stress that $C_i$ is a randomized circuit even though $D_i$ is deterministic.)\\

It is sufficient to argue that $M$ halts in time $T(n) \cdot \mathsf{poly}(n) \cdot 2^{a(n)}$, and that the final circuit $C_n$ computed by $M$ is a correct randomized circuit for $f_{\mathsf{TV}}$ on inputs of length $n$ with high probability over the random choices of $M$.\\

\noindent \textbf{Complexity of $M$.} We will show that $M$ uses time at most $T(i)^2 \cdot \mathsf{poly}(n) \cdot 2^{a(i)}$ in Phase $i$, and computes a circuit $C_i$ of size at most $T(i) \cdot \mathsf{poly}(n)$. Since $a$ and $T$ are non-decreasing, this implies that $M$ uses time at most $T(n)^2 \cdot \mathsf{poly}(n) \cdot 2^{a(n)}$ in total. We will analyze Part 1 and Part 2 separately.

The first step in Part 1, which is computing $a(i)$, can be done in time $\mathsf{poly}(n)$. For each $z$, simulating the learner and computing the circuit $C_i^z$ can be done in time at most $T(i) \cdot T(i-1) \cdot \mathsf{poly}(n)$, since the learner runs in time $T(i)$ and makes at most that many oracle queries, each of which can be answered by simulating a circuit $C_{j}$ of size at most $T(j) \cdot \mathsf{poly}(n)$, where $j \leq i-1$. There are $2^{a(i)}$ advice strings $z$ which $M$ cycles over, hence the total time taken by $M$ in Part 1 of Phase $i$ is at most $T(i)^2 \cdot 2^{a(i)} \cdot \mathsf{poly}(n)$ by the non-decreasing property of $T$.

In Part 2 of Phase $i$, computing the bits $b_1, \ldots, b_{t}$ takes time at most $T(i) \cdot \mathsf{poly}(n)$, since the downward self-reducibility procedure runs in time $\mathsf{poly}(n)$, and every query can be answered by simulation of a circuit $C_j$ with $j < i$ in time at most $T(i) \cdot \mathsf{poly}(n)$. For each $C_i^z$, computing the fraction $\rho_z$ takes time at most $T(i) \cdot \mathsf{poly}(n)$, since it involves simulating $C_i^z$ on $\mathsf{poly}(n)$ inputs, and $C_i^z$ is of size at most $T(i)$. Doing this for each $z$ takes time at most $T(i) \cdot 2^{a(i)} \cdot \mathsf{poly}(n)$ time, as there are $2^{a(i)}$ possible advice strings of length $a(i)$. Computing $z_{\mathsf{max}}$ takes time $\mathsf{poly}(n)$, and then computing the ``corrected'' circuit $C_i$ takes time at most $T(i) \cdot \mathsf{poly}(n)$, since the random self-reducibility procedure runs in polynomial time and can therefore be simulated using polynomial-size circuits.\\

\noindent \textbf{Correctness of $M$.} Clearly Phase 1 concludes with a correct circuit $C_1$ for $f_{\mathsf{TV}}$ on $1$-bit inputs. We will argue inductively that, given that the randomized circuit $C_{i-1}$ computed at the end of Phase $i-1$ is a correct circuit for $f_{\mathsf{TV}}$ on inputs of length $i-1$ such that its error probability is at most $2^{-2n}$ on any input, with all but exponentially small probability over the random choices of $M$ in Phase $i$, the randomized circuit $C_i$ computed at the end of Phase $i$ is a correct circuit for $f_{\mathsf{TV}}$ on inputs of length $i$, with error probability at most $2^{-2n}$ on any input. By a union bound over the phases, it follows from this that with all but exponentially small probability, the final circuit $C_n$ is a correct randomized circuit for $f_{\mathsf{TV}}$ on inputs of length $n$ (with error probability at most $2^{-2n}$), and hence that carrying out all the phases and then simulating $C_n$ on $x$ yields the correct value $f_{\mathsf{TV}}(x)$ with overwhelming probability.

Therefore our task reduces to arguing the correctness of Phase $i$ given the correctness of Phase $i-1$, for an arbitrary $i$ such that $1 < i \leq n$. We discuss the correctness of Part 1 and Part 2 separately. 

In Part 1, we argue that with all but exponentially small probability, at least one of the circuits $C_i^z$ computes $f_{\mathsf{TV}}$ correctly on all but a $1/i^{3k}$ fraction inputs of length $i$. Consider the string $z_i$ of length $a(i)$ that is the ``correct'' advice string for $A_{\mathsf{learn}}$ on input $1^i$. We only analyze Part 1 for the advice string $z_i$ -- the other advice strings are irrelevant to our analysis of correctness for this part. $A_{\mathsf{learn}}$ with advice $z_i$ is a correct learner for $\{f_{\mathsf{TV}}\}$; hence with probability at least $1 - 2^{-2i}$, it outputs a circuit that computes $f_{\mathsf{TV}}$ on at least a $1-1/i^{4k}$ fraction of inputs of length $i$, when it is given access to a \emph{correct} oracle for $f_{\mathsf{TV}}$. By running the learner $\mathsf{poly}(n)$ times independently and doing standard amplification, the success probability can be boosted to $1 - 2^{-2n}$, while keeping the agreement of the hypothesis with $f_{\mathsf{TV}}$ at least $1-1/i^{3k}$, and not affecting the efficiency of $M$ by more than a polynomial factor. $M$ might not be able to answer queries to $f_{\mathsf{TV}}$ with perfect accuracy, however by the inductive  hypothesis that $C_{j}$ has error at most $2^{-2n}$ on any specific input for $j < i$, it follows by a union bound that with probability at least $1 - T(i) 2^{-2n} \geq 1 - 2^{-n}$ over the internal randomness of $M$, the simulation of the learner is correct. Hence with probability at least $1 - 2^{-n}$, $M$ outputs a circuit $C_i^{z_i}$ during Phase $i$, Part 1, such that $C_i^{z_i}$ agrees with $f_{\mathsf{TV}}$ on at least a $1-1/i^{3k}$ fraction of inputs of length $i$. 

Next we analyze Part 2 of Phase $i$. By a union bound, with probability at least $1 - \mathsf{poly}(n)/2^{2n}$, the ``guesses'' $b_1, \ldots, b_{t} \in \{0,1\}$ are all the correct values for $f_{\mathsf{TV}}$ on inputs $y_1, \ldots , y_{t} \in \{0,1\}^{i}$, where by construction $t = n^{10k}$. By using a standard concentration bound such as Lemma \ref{l:concentration}, we have that the estimate $\rho_{z_i}$ is at least $1-1/i^{2k}$ with probability at least $1-2^{-4n}$, and that with probability at least $1-2^{-4n}$ any $z$ such that the agreement $\rho_z$ is at least $1-1/i^{2k}$ must be such that $C_i^z$ agrees with $f_{\mathsf{TV}}$ on at least a $1-1/i^{3k/2}$ fraction of inputs of length $i$. Thus with probability at least $1 - \mathsf{poly}(n)/2^{2n}$, we have that $C_i^{z_{\mathsf{max}}}$ has agreement at least $1-1/i^{3k/2}$ with $f_{\mathsf{TV}}$ on inputs of length $i$. By again using a union bound and a standard concentration bound such as Lemma \ref{l:concentration}, we have that with all but exponentially small probability, the corrected circuit $C_i$ is a randomized circuit which computes $f_{\mathsf{TV}}$ correctly on all inputs of length $i$, making error $< 2^{-2n}$ on any single input. This completes the inductive argument for correctness.
\end{proof}

\subsection{Karp-Lipton Results for Bounded-Error Exponential Time}

\begin{lemma} [Learnability with advice from distinguishability]
\label{l:GenericDistLearn}
Let $f \in \mathsf{EXP}$ be a Boolean function and $a \colon \mathbb{N} \rightarrow \mathbb{N}$ be a advice function.
\begin{enumerate}

\item[\emph{1.}] \emph{(High-End Generator)} There is a constant $c \geq 1$ such that for any $\varepsilon \in (0,1]$, there is a sequence of functions $\{G^{\mathsf{HE}}_n\}_{n \in \mathbb{N}}$ with $G^{\mathsf{HE}}_n \colon \{0,1\}^{n^c} \rightarrow \{0,1\}^{2^{n^{\varepsilon}}}$ computable in deterministic time $2^{O(n^c)}$ such that if there is a probabilistic procedure $A(1^n)$ taking $a(n)$ bits of advice and running in time $2^{O(n^{\varepsilon})}$, and outputting a circuit distinguisher for $G^{\mathsf{HE}}_n(U_{n^c})$ with constant probability for all large enough $n$, then $\{f\}$ is strongly learnable in time $2^{O(n^{\varepsilon})}$ with $a(n)$ bits of advice.

\item[\emph{2.}] \emph{(Low-End Generator)} There is a constant $c \geq 1$ such that for any $d \geq 1$, there is a sequence of functions $\{G^{\mathsf{LE}}_n\}_{n \in \mathbb{N}}$ with $G^{\mathsf{LE}}_n \colon \{0,1\}^{n^c} \rightarrow \{0,1\}^{2^{(\log n)^d}}$ computable in deterministic time $2^{O(n^c)}$ such that if there is a probabilistic quasipolynomial-time procedure $A(1^{n})$ taking $a(n)$ bits of advice and outputting a circuit distinguisher for $G^{\mathsf{LE}}_n(U_{n^c})$ with constant probability for all large enough $n$, then $\{f\}$ is strongly learnable in quasi-polynomial time with $a(n)$ bits of advice.
\end{enumerate}
\end{lemma}

\begin{proof} (Nutshell) This follows from the reconstruction procedure for the Nisan-Wigderson generator together with hardness amplification. We refer to \citep{DBLP:journals/jcss/NisanW94} for more details.
\end{proof}

\begin{theorem} [Low-end Karp-Lipton Theorem for bounded-error exponential time]
\label{t:BPEXP_KarpLipton}~\\
If there is a $k \geq 1$ such that $\mathsf{BPE} \subseteq \mathtt{i.o.}\mathsf{Circuit}[n^k]$, then $\mathsf{BPEXP} \subseteq \mathtt{i.o.}\mathsf{EXP}/O(\log n)$.
\end{theorem}

\begin{proof}
We will prove the contrapositive. For each bounded-error probabilistic exponential time machine
$M$, we will define for each rational $\varepsilon > 0$ a deterministic exponential-time machine $M_{\varepsilon}$ taking logarithmic advice which attempts to simulate it. If all of the attempted simulations $M_{\varepsilon}$ fail almost everywhere, we will show that $\mathsf{PSPACE} \subseteq \mathsf{BPSUBEXP}$, and we will then use a translation argument and Corollary \ref{c:iockthard} to conclude that $\mathsf{BPE} \not \subseteq \mathtt{i.o.}\mathsf{Circuit}[n^k]$, thus establishing the contrapositive. 

Let $M$ be any bounded-error probabilistic machine running in time $2^{m^j}$, where $m$ is the input length, and $j$ is a constant. We assume without loss of generality that $j \geq 1$, and that $M$ has error $< 1/4$ on any input. Let $\varepsilon > 0$ be any rational. We define the deterministic exponential-time machine $M_{\varepsilon}$ taking $O(\log m)$ bits of advice on inputs of length $m$ below. It uses the generators $\{G^{\mathsf{HE}}_n\}$ given by Lemma \ref{l:GenericDistLearn} corresponding to the $\mathsf{PSPACE}$-complete language $f_{\mathsf{TV}}$ in the statement of Theorem \ref{t:PSPACECompRsrDsr}, which is clearly in exponential time.

On input $x$ of length $m$, $M_{\varepsilon}$ first uses the advice on its tape to determine an integer $n$ such that $2^{2m^j} \leq 2^{n^{\varepsilon}} < 2^{2(m+1)^j}$. Note that any $n \in \mathbb{N}$ satisfying these conditions is such that $n = \Theta(m^{j/\varepsilon})$. Hence there are $\mathsf{poly}(m)$ possibilities for $n$, and any of these possibilities can be encoded using $O(\log m)$ bits on the advice tape. Given a number $i$ on the advice tape, $M_{\varepsilon}$ can decode the relevant $n$ by determining the $i$-th number in increasing order satisfying both inequalities. This can be done in $\mathsf{poly}(m)$ time since we can assume $\varepsilon$ is hard-coded into $M_{\varepsilon}$, and any single inequality verification can be done in $\mathsf{poly}(m)$ time. $M_{\varepsilon}$ then computes $R(y) = G^{\mathsf{HE}}_{n}(y)$ for every string $y \in \{0,1\}^{n^c}$. It simulates $M$ on $x$ using each string $R(y)$ in turn as the randomness for $M$, and outputs the majority result of these simulations. It is easy to see that $M_{\varepsilon}$ can be implemented to run in $2^{O(n^c)} = 2^{O(m^{cj/\varepsilon})}$ time, i.e, in time that is exponential on its input length $m$.

If any of the simulations $M_{\varepsilon}$ succeeds on infinitely many input lengths $m$, we have that $L(M) \in \mathtt{i.o.}\mathsf{EXP}/O(\log m)$. Suppose, contrariwise, that all of the simulations $M_{\varepsilon}$ fail almost everywhere. We will argue that $f_{\mathsf{TV}} \in \mathsf{BPSUBEXP}$ and hence, by completeness of $f_{\mathsf{TV}}$, $\mathsf{PSPACE} \subseteq \mathsf{BPSUBEXP}$.

For any $x \in \{0,1\}^m$, let $C_x$ be the circuit of size at most $2^{2m^j}$ defined as follows: the input of $C_x$ is the sequence of random bits $r$ used by $M$ in its computation on $x$. $C_x(r)$ accepts iff $M$ accepts on $x$ using the sequence $r$ of random bits. By the standard translation of deterministic computations into circuits, $C_x$ can be implemented in size at most $2^{2m^j}$, using the fact that $M$ halts in time $2^{m^j}$.

Fix any $\varepsilon > 0$. Let $n$ be an arbitrary positive integer, and let $m(n)$ be the unique $m$ such that $2^{2m^j} \leq 2^{n^{\varepsilon}} < 2^{2(m+1)^j}$ (observe that $h(a) \eqdef 2^{2a^j}$ is an increasing function, so this $m$ is indeed unique if $n$ is not too small). We claim that for every large enough $n$, there is an input $x$ of length $m(n)$ such that $C_x$ is a distinguisher for $G^{\mathsf{HE}}_n(U_{n^c})$. Indeed, if not, there are infinitely many $n$ such that for all inputs $x$ of length $m(n)$, $C_x$ is not a distinguisher, but this implies that the simulation $M_{\varepsilon}$ on inputs of length $m(n)$ would succeed infinitely often with advice encoding the input length $n$. Since for each $m$, there are only finitely many $n$ such that $m = m(n)$, it follows that the simulation $M_{\varepsilon}$ succeeds on infinitely many input lengths with logarithmic advice. But this  contradicts our assumption that the simulation $M_{\varepsilon}$ fails almost everywhere.

Now that our claim is established, we define a deterministic procedure $A(1^n)$ taking $O(n^{\varepsilon})$ bits of advice and running in time $2^{O(n^{\varepsilon})}$, which for each large enough $n$ produces a circuit distinguisher for $G^{\mathsf{HE}}_n$. The procedure $A$ computes $m(n)$ in polynomial time. Note that $m(n) = O(n^{\varepsilon/j}) = O(n^{\varepsilon})$, by our assumption that $j \geq 1$. $A$ then interprets its advice as an string $x$ of length $m(n)$. It computes $C_x$, which it can do given $x$ in time polynomial in the size of $C_x$, and outputs $C_x$. The time taken by $A$ is dominated by the time required to compute $C_x$, which is $2^{O(n^{\varepsilon})}$, and the advice used by $A$ is of size $O(n^{\varepsilon})$.

By applying Lemma \ref{l:GenericDistLearn}, we get that $\{f_{\mathsf{TV}}\}$ is strongly learnable in time $2^{O(n^{\varepsilon})}$ with $O(n^{\varepsilon})$ bits of advice. By applying Lemma \ref{l:PSPACELearnSolve}, we get that $f_{\mathsf{TV}}$ is computable in bounded-error probabilistic time $2^{O(n^{\varepsilon})}$. Note that this is the case for every $\varepsilon > 0$, since our choice of $\varepsilon$ was arbitrary. Thus we have $f_{\mathsf{TV}} \in \mathsf{BPSUBEXP}$, and hence by completeness that $\mathsf{PSPACE} \subseteq \mathsf{BPSUBEXP}$. Using a standard upward translation argument and applying Corollary \ref{c:iockthard}, we get that for every $k > 0$, $\mathsf{BPE} \not \subseteq \mathtt{i.o.}\mathsf{Circuit}[n^k]$, which is the desired conclusion.
\end{proof}

\begin{theorem} [High-end Karp-Lipton Theorem for bounded-error exponential time]
\label{t:BPEXP_KarpLipton_LowEnd}~\\
If $\mathsf{BPEXP} \subseteq \mathtt{i.o.}\mathsf{Circuit}[2^{n/3}]$, then for each $\varepsilon > 0$, $\mathsf{BPEXP} \subseteq \mathtt{i.o.}\mathsf{DTIME}(2^{2^{n^{\varepsilon}}})/n^{\varepsilon}$.
\end{theorem}

\begin{proof} (Sketch) The proof is entirely analogous to the proof of Theorem \ref{t:BPEXP_KarpLipton}, except that we use generators $G^{\mathsf{LE}}_n$ rather than the generators $G^{\mathsf{HE}}_n$, adjusting other parameters accordingly. We get that either $\mathsf{PSPACE} \subseteq \mathsf{BPQP}$, or that for every $\varepsilon > 0$, $\mathsf{BPEXP} \subseteq \mathtt{i.o.}\mathsf{DTIME}(^{2^{n^{\varepsilon}}})/n^{\varepsilon}$. In the first case, by upward translation, we get that $\mathsf{EXPSPACE} = \mathsf{BPEXP}$, and then by using Corollary \ref{c:iockthard}, we conclude that $\mathsf{BPEXP} \not \subseteq \mathtt{i.o.}\mathsf{Circuit}[2^{n/3}]$. 
\end{proof}

\begin{theorem} [Low-end fully uniform Karp-Lipton style theorem for probabilistic time]
\label{t:REXP_KarpLipton}~\\
If there is a $k \geq 1$ such that $\mathsf{BPE} \subseteq \mathtt{i.o.}\mathsf{Circuit}[n^k]$, then $\mathsf{REXP} \subseteq \mathtt{i.o.}\mathsf{EXP}$.
\end{theorem}

\begin{proof} (Sketch) We use the crucial fact that the union of hitting sets is also a hitting set to eliminate the advice in the simulation. The argument is the same as in the proof of Theorem \ref{t:BPEXP_KarpLipton}, except that the simulating machine $M_{\varepsilon}$ runs $M$ on $x$ using as randomness $R$ every element in turn that is in the range of $G^{\mathsf{HE}}_n$ for every $n$ such that $2^{2m^j} \leq 2^{n^{\varepsilon}} < 2^{2(m+1)^j}$, accepting if and only if any of these runs accepts. Note that $M_{\varepsilon}$ does not take advice. We do not need to give the ``correct'' $n$ as advice to the machine because, if any $n$ in the interval produces an accepting path (corresponding to a string in the range of the generator), then $\bigcup_n G_n^{\mathsf{HE}}(U_{n^c})$ for $n$ as above contains an accepting path for $M$ on $x$. Finally, we observe that computing the range of the generator for every such $n$ does not blow-up the complexity of the simulation by more than a polynomial factor.
\end{proof}

\begin{theorem} [High-end fully uniform Karp-Lipton style theorem for probabilistic time]
\label{t:REXP_KarpLipton_LowEnd}~\\
If $\mathsf{BPEXP} \subseteq \mathtt{i.o.}\mathsf{Circuit}[2^{n/3}]$, then $\mathsf{REXP} \subseteq \mathtt{i.o.}\mathsf{ESUBEXP}$.
\end{theorem}

\begin{proof} (Sketch)
The proof is entirely analogous to the proof of Theorem \ref{t:REXP_KarpLipton}, except that we use generators $G^{\mathsf{LE}}_n$ rather than the generators $G^{\mathsf{HE}}_n$, adjusting other parameters accordingly.
\end{proof}

These results can be combined with a Karp-Lipton collapse for deterministic exponential time. For instance, the following holds. 

\begin{corollary}
If there is $k \in \mathbb{N}$ such that $\mathsf{BPE} \subseteq \mathsf{Circuit}[n^k]$, then $\mathsf{REXP} \subseteq \mathtt{i.o.}\mathsf{MA}$.
\end{corollary} 

\begin{proof}
It follows from the hypothesis that $\mathsf{E} \subseteq \mathsf{Circuit}[n^k]$, and hence $\mathsf{EXP} \subseteq \mathsf{Circuit}[\mathsf{poly}]$ by translation. This in turn implies that $\mathsf{EXP} = \mathsf{MA}$ \citep{DBLP:journals/cc/BabaiFNW93}. Moreover, the hypothesis gives $\mathsf{REXP} \subseteq \mathtt{i.o.}\mathsf{EXP}$ using Theorem \ref{t:REXP_KarpLipton}. Consequently, we get $\mathsf{REXP} \subseteq \mathtt{i.o.}\mathsf{MA}$, which completes the proof.
\end{proof}

\subsection{Karp-Lipton Results for Zero-Error Exponential Time}

\begin{lemma} [Fully uniform simulations using easy witness and truth-table concatenation]
\label{l:EasyWitness}~\\
Either $\mathsf{BPP} \subseteq \mathsf{ZPQP}$, or $\mathsf{ZPEXP} \subseteq \mathtt{i.o.}\mathsf{ESUBEXP}$.
\end{lemma}

\begin{proof} We use the ``easy witness'' method of Kabanets \citep{DBLP:journals/jcss/Kabanets01}. Let $M$ be any probabilistic Turing machine with zero error running in time $2^{m^j}$ for some $j \geq 1$, such that on each random computation path of $M$ on any input $x$, the output is either the correct answer for $M$ on $x$ or `?', and moreover the probability of outputting `?' is less than $2^{-2m}$ for any input $x \in \{0,1\}^m$. For each $\varepsilon > 0$, we define the following attempted deterministic simulation $M_{\varepsilon}$ for $M$. 
On input $x$ of length $m$, $M_{\varepsilon}$ cycles over all circuits $C$ of size $2^{m^{\varepsilon/2}}$ on $m^j$ inputs. For each such circuit, it explicitly computes the truth table $\mathsf{tt}(C)$ of the circuit $C$, and runs $M$ on $x$ with $\mathsf{tt}(C)$ as randomness. If the run accepts, it accepts; if the run rejects, it rejects. If the run outputs `?', it moves on to the next circuit $C$ in lexicographic order of circuit encodings. If all runs output `?', the machine rejects. It should be clear that the simulation $M_{\varepsilon}$ runs in deterministic time $\leq 2^{2^{m^{\varepsilon}}}$ on inputs of length $m$ for any sufficiently large $m$, and only accepts inputs $x \in L(M)$. 

If for each $\varepsilon > 0$, we have that the simulation $M_{\varepsilon}$ solves $L(M)$ correctly on all inputs of length $m$ for infinitely many input lengths $m$, we have that $L(M) \subseteq \mathtt{i.o.}\mathsf{ESUBEXP}$. 

Suppose, on the contrary, that there is some $\varepsilon > 0$ such that the simulation $M_{\varepsilon}$ fails on at least one input $x$ of each large enough input length $m$. We show how to use this to decide every language in $\mathsf{BPP}$ in $\mathsf{ZPQP}$.

Let $N$ be any bounded-error probabilistic machine running in time at most $n^k$ for some constant $k$ and large enough $n$. Assume without loss of generality that $N$ has error $\leq 1/10$ on any input of length $n$. We use $M$ to give a zero-error simulation $N'$ of $N$ on all inputs of large enough length. Given an input $y$ of length $n$, $N'$ simulates $M$ on each input $x$ of length $m(n) \eqdef (\lceil \log n \rceil)^d$ in turn, for some constant $d \geq 1$ to be specified later. If $M$ outputs `?', $N'$ outputs `?', otherwise it moves on to the next input in lexicographic order. If running $M$ gives `?' outputs for every input $x$ of length $m(n)$, $N'$ outputs `?'. Otherwise, $N'$ concatenates the random strings used on the computation paths of $M$ for each input of length $m(n)$ into a single string $R_n$ of length $O(2^{\mathsf{polylog}(n)})$. It then uses $R_n$ as the truth-table of the hard function for the generator in Theorem \ref{t:hardness_PRG}, setting parameters so that at least $n^{2k}$ pseudorandom bits are produced by the generator. It cycles over all possible seeds of the generator and runs $N$ using each output in turn as the sequence of random choices, accepting if and only if a majority of runs accepts.

Setting $d$ to be a large enough constant depending on $j, k, \varepsilon$ and the constant $c$ in the statement of Theorem \ref{t:hardness_PRG}, it can be shown that this simulation can be done in quasi-polynomial time, and that it is correct for each input $y$ of large enough length whenever $N'$ does not output `?'. The key here is that by the failure of $M_{\varepsilon}$ for at least one input of any large enough length, the string $R_n$ is guaranteed to be hard enough that the generator is correct. This is because $R_n$ contains a subfunction of sufficiently large worst-case circuit complexity. Hence cycling over all seeds of the generator and taking the majority value gives the correct answer for $N$ on input $y$. Finally, under our initial assumption that $M$ has exponentially small failure probability, by a union bound the probability that $N'$ outputs `?' on any large enough input is small. This concludes the proof that $\mathsf{BPP} \subseteq \mathsf{ZPQP}$.
\end{proof}

\begin{theorem} [High-end fully uniform Karp-Lipton theorem for zero-error exponential time]
\label{t:ZPEXP_Karp-Lipton}~\\
If $\mathsf{ZPEXP} \subseteq \mathtt{i.o.}\mathsf{Circuit}[2^{n/3}]$, then $\mathsf{ZPEXP} \subseteq \mathtt{i.o.}\mathsf{ESUBEXP}$.
\end{theorem}

\begin{proof}
Observe that the proof of Theorem \ref{t:REXP_KarpLipton_LowEnd} establishes that if $\mathsf{REXP} \not \subseteq \mathtt{i.o.}\mathsf{ESUBEXP}$, then $\mathsf{PSPACE} \subseteq \mathsf{BPQP}$. By Lemma \ref{l:EasyWitness}, if $\mathsf{ZPEXP} \not \subseteq \mathtt{i.o.}\mathsf{ESUBEXP}$, then $\mathsf{BPP} \subseteq \mathsf{ZPQP}$, and hence by upward translation, $\mathsf{BPQP} = \mathsf{ZPQP}$. Putting these together, we have that if $\mathsf{ZPEXP} \not \subseteq \mathtt{i.o.}\mathsf{ESUBEXP}$, then $\mathsf{PSPACE} \subseteq \mathsf{ZPQP}$. Now by upward translation, we have that $\mathsf{EXPSPACE} = \mathsf{ZPEXP}$, and hence by Corollary \ref{c:iockthard}, we get $\mathsf{ZPEXP} \not \subseteq \mathtt{i.o.}\mathsf{Circuit}[2^{n/3}]$.
\end{proof}

We have learned from Valentine Kabanets (private communication) that he has independently established Theorem \ref{t:ZPEXP_Karp-Lipton} in an unpublished manuscript.

In fact, we can get a non-trivial consequence from the {\it weakest} possible non-trivial assumption about the circuit size of Boolean functions computable in zero-error exponential time. This extension of Theorem \ref{t:ZPEXP_Karp-Lipton} relies on the following simple lemma.

\begin{lemma} [Maximally hard functions in exponential space]
\label{l:maxhard}
Let $s_{\mathsf{max}} \colon \mathbb{N} \rightarrow \mathbb{N}$ be such that for each $n \in \mathbb{N}$, $s_{\mathsf{max}}(n)$ is the maximum circuit complexity among Boolean functions on $n$ bits. Then $\mathsf{EXPSPACE} \not \subseteq \mathtt{i.o.}\mathsf{Circuit}[s_{\mathsf{max}}-1]$.
\end{lemma}

\begin{proof} (Sketch) The proof is by simple diagonalization. In exponential space, we can systematically list the truth tables of Boolean functions on $n$ bits, and maintain the one with the highest circuit complexity. To compute the circuit complexity of a listed truth table can be done by cycling over all circuits, starting from the smallest one, and checking for each circuit whether it computes the given truth table. Once the truth table of a function with maximum circuit complexity has been computed, we simply look up the corresponding entry in the truth table for any particular input.
\end{proof}

Now by using the same proof as for Theorem \ref{t:ZPEXP_Karp-Lipton} but applying Lemma \ref{l:maxhard} instead of Corollary \ref{c:iockthard}, we have the following stronger version of Theorem \ref{t:ZPEXP_Karp-Lipton}. (We note that Theorems \ref{t:BPEXP_KarpLipton_LowEnd} and \ref{t:REXP_KarpLipton_LowEnd} admit similar extensions.) 

\begin{theorem} [Strong Karp-Lipton Theorem for zero-error probabilistic exponential time]
\label{t:ZPEXP_max-hard}~\\
Let $s_{\mathsf{max}} \colon \mathbb{N} \rightarrow \mathbb{N}$ be such that for each $n \in \mathbb{N}$, $s_{\mathsf{max}}(n)$ is the maximum circuit complexity among Boolean functions on $n$ bits. If $\mathsf{ZPEXP} \subseteq \mathtt{i.o.}\mathsf{Circuit}[s_{\mathsf{max}}-1]$, then $\mathsf{ZPEXP} \subseteq \mathtt{i.o.}\mathsf{ESUBEXP}$.
\end{theorem}

\section{Hardness of the Minimum Circuit Size Problem}\label{s:mcsp}

We will be dealing with various notions of non-uniform reduction to versions of the Minimum Circuit Size Problem (MCSP). Reductions computable in a non-uniform class $\mathfrak{C}$ are formalized using oracle $\mathfrak{C}$-circuits, which are $\mathfrak{C}$-circuits with oracle gates. 
We only use oracle circuits where oracle gates appear all at the same level. In this setting, we can define size and depth of oracle circuits to be the size and depth respectively of the oracle circuits with oracle gates replaced by AND/OR gates. 

\begin{definition} [Non-uniform Reductions]
\label{d:nonunifref}
Let $\mathfrak{C}$ be a typical class of circuits, and $L$ and $L'$ be languages.
\begin{itemize}

\item \emph{($m$-reduction)} We say $L$ $\mathfrak{C}$-reduces to $L'$ via $m$-reductions if there is a sequence of poly-size oracle $\mathfrak{C}$-circuits computing the slices $L_n$ of $L$ when the circuits are given oracle $L'$, and such that each oracle circuit has a single oracle gate, which is also the top gate of the circuit.

\item \emph{($tt$-reduction)} We say $L$ $\mathfrak{C}$-reduces to $L'$ via $tt$-reductions if there is a sequence of poly-size oracle $\mathfrak{C}$-circuits computing the slices $L_n$ of $L$ when the circuits are given oracle $L'$, and such that no oracle circuit has a directed path from one oracle gate to another.

\item \emph{($\varepsilon$-approximate reductions)} We extend these notions to hold between approximations of a language. Given a function $\varepsilon \colon \mathbb{N} \rightarrow [0,1]$ and languages $L$ and $L'$, we say that $L$ reduces to $\varepsilon$-approximating $L'$ under a certain notion of reduction if for each $\widetilde{L'}$ which agrees with $L'$ on at least a $1-\varepsilon(n)$ fraction of inputs of length $n$ for large enough $n$, $L$ reduces to $\widetilde{L'}$ under that notion of reduction. We say that $\varepsilon$-approximating $L$ reduces to $L'$ if there is a language $\widetilde{L}$ which agrees with $L$ on at least a $1-\varepsilon(n)$ fraction of inputs of length $n$ for large enough $n$, such that $\widetilde{L}$ reduces to $L'$. More generally, we say that $\varepsilon$-approximating $L$ reduces to $\varepsilon'$-approximating $L'$ under a certain notion of reduction if for any language $\widetilde{L'}$ that $\varepsilon'(n)$-approximates $L'$ on inputs of length $n$ for large enough $n$, there is a language $\widetilde{L}$ that $\varepsilon(n)$-approximates $L$ on inputs of length $n$ for large enough $n$, and a corresponding reduction from $\widetilde{L}$ to $\widetilde{L'}$.

\item \emph{(Parameterized reduction)} If a reduction is not computed by polynomial size circuits, we extend these definitions in the natural way, and say that $L$ $\mathfrak{C}[s]$-reduces to $L'$, where $s \colon \mathbb{N} \to \mathbb{N}$ is the appropriate circuit size bound. 

\end{itemize}

\end{definition}

The following proposition is immediate from the definitions and the fact that typical circuit classes are closed under composition.

\begin{proposition} [Transitivity of reductions]
\label{p:transitive}
Let $\mathfrak{C}$ be a typical circuit class, $L$, $L'$, $L''$ be languages, and $\varepsilon, \varepsilon', \varepsilon'' \colon \mathbb{N} \rightarrow [0,1]$ be functions.  
\begin{itemize}
\item [\emph{(}i\emph{)}] If $L$ $\mathfrak{C}$-reduces to $L'$ via $m$-reductions \emph{(}resp. $tt$-reductions\emph{)} and $L'$ $\mathfrak{C}$-reduces to $L''$ via $m$-reductions \emph{(}resp. $tt$-reductions\emph{)}, then $L$ $\mathfrak{C}$-reduces to $L''$ via $m$-reductions \emph{(}resp. $tt$-reductions\emph{)}.
\item [\emph{(}ii\emph{)}] If $\varepsilon(\mathsf{poly}(n))$-approximating $L$ $\mathfrak{C}$-reduces to $\varepsilon'(\mathsf{poly}(n))$-approximating $L'$ via $m$-reductions \emph{(}$tt$-reductions\emph{)} and $\varepsilon'(\mathsf{poly}(n))$-approximating $L'$ $\mathfrak{C}$-reduces to $\varepsilon''(\mathsf{poly}(n))$-approximating $L''$ via $m$-reductions \emph{(}$tt$-reductions\emph{)}, it follows that $\varepsilon(\mathsf{poly}(n))$-approximating $L$ $\mathfrak{C}$-reduces to $\varepsilon''(\mathsf{poly}(n))$-approximating $L''$ via $m$-reductions \emph{(}$tt$-reductions\emph{)}.
\end{itemize}
\end{proposition}

Using the notation introduced above, the following fact is trivial to establish.

\begin{proposition} [Relation between parameterized and unparameterized versions of MCSP]
\label{p:MCSPFact}
~\\For any typical circuit class $\mathfrak{C}$, $\mathsf{MCSP}$-$\mathfrak{C}[2^{n/2}]$ $\mathsf{AC}^0$-reduces to $\mathsf{MCSP}$-$\mathfrak{C}$ via m-reductions. 
\end{proposition}

\begin{theorem} [Hardness of MCSP for weakly approximating functions in typical circuit classes]\label{t:NWNonunifHardness}
Let $\mathfrak{C}$ be a typical circuit class that contains $\mathsf{AC}^0[p]$, for some fixed prime $p$. For every Boolean function $f \in \mathfrak{C}[n^k]$ there exists $c = c(k,\delta) \in \mathbb{N}$ such that $(1/2- \Omega(1/n^c))$-approximating $f$ $\mathsf{AC}^0$-reduces to $\mathsf{MCSP}$-$\mathfrak{C}[2^{n/2}]$ via $tt$-reductions, as well as to any property with density at least $1/4$ that is useful against $\mathfrak{C}[2^{\delta n}]$ for some fixed $\delta \in (0,1)$.
\end{theorem}

\begin{proof} (Sketch) Let $f = \{f_n\}_{n \in \mathbb{N}}$ be a function in $\mathfrak{C}[n^k]$, where $f_n \colon \{0,1\}^n \to \{0,1\}$ and $\mathsf{AC}^0[p] \subseteq \mathfrak{C}[\mathsf{poly}]$. Further, let $0 < \delta < 1$ be a constant. We let
$$
\mathsf{NW}_c(f_n) \eqdef \{g_z \colon \{0,1\}^{c \log n}  \to \{0,1\} \mid z \in \{0,1\}^{\Theta(n^2)}\}
$$
be the family (multiset) of functions obtained by instantiating the Nisan-Wigderson \citep{DBLP:journals/jcss/NisanW94} construction with the $\mathsf{AC}^0[p]$-computable designs from \citep{CIKK16} and $f_n$. A bit more precisely, each $g_z$ is a function specified by a seed $z$ of length $\Theta(n^2)$, the family of sets $\mathcal{S}_n = \{S_w \subseteq [\Theta(n^2)] \mid w \in \{0,1\}^{c \log n}\}$, and $f_n$, and we have $g_z(w) \eqdef f_n(z_{S_w})$. Here each $S_w \subseteq [\Theta(n^2)]$ contains exactly $n$ elements ($S_w$ is the $w$-th set in the design), and $z_{S_w} \in \{0,1\}^n$ is the projection of $z$ to coordinates $S_w$. By taking $c =  c(\delta, k)$ sufficiently large and using that $\delta > 0$, $f_n \in \mathfrak{C}[n^k]$, and that the design can be implemented in $\mathfrak{C}[\mathsf{poly}]$, it follows from \citep{DBLP:journals/jcss/NisanW94, CIKK16} that for large enough $n$:\\ 

\noindent (A) Each $g_z$ is a function on $m \eqdef c \log n$ input bits of $\mathfrak{C}$-circuit complexity $\leq 2^{\delta m}$.\\

On the other hand, if $h_m \sim \mathcal{F}_m$ is a uniformly random Boolean function on $m$ input bits, using that $\delta < 1$ it easily follows from a counting argument (e.g. Lemma \ref{l:number_circuits}) that for large enough $n$ (recall that $m = c \log n$):\\

\noindent (B) $h_m$ has $\mathfrak{C}$-circuit complexity $> 2^{\delta m}$ with probability $1 - o(1)$.\\

Consequently, from (A) and (B) we get that an oracle to $\mathsf{MCSP}$-$\mathfrak{C}[2^{n/2}]$ (corresponding to $\delta = 1/2$) can be used to distinguish the multiset $\mathsf{NW}_c(f_n)$ (sampled according to $z \sim \{0,1\}^{\Theta(n^2)})$ from a random function on $m$ input bits. 

We argue next that it follows from the description of the Nisan-Wigderson reconstruction procedure \citep{DBLP:journals/jcss/NisanW94} that there is a $tt$-reduction from $(1/2 - \Omega(1/n^c))$-approximating $f_n$ to $\mathsf{MCSP}$-$\mathfrak{C}[2^{n/2}]$ that is computable by $\mathsf{AC}^0$-circuits. That some non-uniform approximate reduction with oracle access to $f_n$ exists immediately follows from the proof of their main result. That it can be computed in $\mathsf{AC}^0[\mathsf{poly}(n)]$ with oracle access to the distinguisher $\mathsf{MCSP}$-$\mathfrak{C}[2^{n/2}]$ (and without oracle access to $f_n$) follows by our choice of parameters (in particular, $|S_{w_1} \cap S_{w_2}| = O(\log n)$ for every pair $w_1 \neq w_2$), non-uniformity of the reduction, and the fact that the output of $f_n$ on any particular $n$-bit input can be hardwired into the (non-uniform) $\mathsf{AC}^0$ circuit computing the reduction. Finally, we remark that the $\Omega(1/n^c)$ advantage in the approximation comes from the truth-table size of each $g_z$ and the hybrid argument in \citep{DBLP:journals/jcss/NisanW94}, and that we get a $tt$-reduction because the reconstruction procedure is non-adaptive.

In fact, the same argument shows that $(1/2 - \Omega(1/n^c))$-approximating $f$ $\mathsf{AC}^0$-reduces via $tt$-reductions to any property useful against $\mathfrak{C}[2^{\delta n}]$ for some $\delta \in (0,1)$ and with density at least $1/4$, since this suffices to implement the Nisan-Wigderson reconstruction routine. This completes the proof of Theorem \ref{t:NWNonunifHardness}.
\end{proof}

\begin{corollary} [Hardness of the standard circuit version of MCSP]
\label{c:MCSPCkt}
For any Boolean function $f \in \mathsf{Circuit}[\mathsf{poly}(n)]$, there exists $c \geq 1$ such that  $(1/2 - 1/n^c)$-approximating $f$ $\mathsf{AC}^0$-reduces via $tt$-reductions to $\mathsf{MCSP}$-$\mathsf{Circuit}$ and to any property with density at least $1/4$ that is useful against $\mathsf{Circuit}[2^{n/2}]$. 
\end{corollary}

\begin{proof}
The second item follows immediately from Theorem \ref{t:NWNonunifHardness}, since $\mathsf{Circuit}$ is typical. The first item follows from Theorem \ref{t:NWNonunifHardness}, Proposition \ref{p:MCSPFact} and Proposition \ref{p:transitive}.
\end{proof}

\begin{proposition} [Hardness amplification for $\mathsf{Formula}$]
\label{p:NC1ampl}
There exists a Boolean function $f \in \mathsf{Formula}$ that is $\mathsf{Formula}$-hard under $\mathsf{AC}^0$-reductions such that for every integer $d \geq 1$, $f$ $\mathsf{TC}^0$-reduces via $tt$-reductions to $(1/2-1/n^d)$-approximating $f$. 
\end{proposition}

\begin{proof} (Sketch) This is achieved using a standard hardness amplification argument using the existence of a random self-reducible complete problem in $\mathsf{NC}^1$, as well as the XOR lemma. It is known that the circuits used in the hardness amplification reconstruction procedure and for random-self-reducibility can be implemented in non-uniform $\mathsf{TC}^0$. For more details, we refer to  \citep{DBLP:journals/siamcomp/ShaltielV10, DBLP:conf/approx/AllenderAW10, DBLP:books/sp/goldreich2011/GoldreichNW11}.
\end{proof}

\begin{corollary} [Hardness of MCSP for $\mathsf{NC}^1$]
\label{c:NC1MCSP}
For every Boolean function $f \in \mathsf{Formula}$, $f$ $\mathsf{TC}^0$-reduces to the following problems via $tt$-reductions\emph{:}
\begin{enumerate}

\item[\emph{1.}] $\mathsf{MCSP}$-$\mathsf{Formula}[2^{n/2}]$. 

\item[\emph{2.}] Any property useful against $\mathsf{Formula}[2^{\delta n}]$ for $\delta \in (0,1)$ and with density at least $1/4$.

\item[\emph{3.}] $\mathsf{MCSP}$-$\mathsf{Formula}$.

\item[\emph{4.}] $\mathsf{MCSP}$-$\mathfrak{C}$ for any typical circuit class $\mathfrak{C} \supseteq \mathsf{Formula}$.

\end{enumerate}

\end{corollary}

\begin{proof}
Items 1 and 2 follow from Theorem \ref{t:NWNonunifHardness} applied to the typical class $\mathsf{Formula}$, together with Propositions \ref{p:transitive} and \ref{p:NC1ampl}. Item 3 follows from Item 1 and Propositions \ref{p:transitive} and \ref{p:MCSPFact}. Finally, in order to prove Item 4, note that $\mathsf{MCSP}$-$\mathfrak{C}[2^{n/2}]$ is useful against $\mathsf{Formula}[2^{n/2}]$, using the assumption that formulas are subclasses of $\mathfrak{C}$ circuits. Moreover, $\mathsf{MCSP}$-$\mathfrak{C}[2^{n/2}]$ as a property has density $1-o(1)$, since a random function has circuit complexity higher than $2^{n/2}$ with probability exponentially close to $1$ by the usual counting argument. Thus, it follows using the same argument as for Item 2 that $\mathsf{MCSP}$-$\mathfrak{C}[2^{n/2}]$ is $\mathsf{TC}^0$-hard under $tt$-reductions for $\mathsf{Formula}$. Item 4 now follows from this via Propositions \ref{p:transitive} and \ref{p:MCSPFact}.
\end{proof}

Hardness results as in Corollary \ref{c:NC1MCSP} also follow for other classes such as non-uniform  logarithmic space and the class of problems reducible to the determinant using non-uniform $\mathsf{TC}^0$ reductions, since these classes also have random self-reducible complete problems and admit worst-case to average-case reducibility in low complexity classes. We will not further elaborate on this here. 

A closely related problem is whether a string has high KT complexity (cf. \citep{DBLP:journals/siamcomp/AllenderBKMR06}). KT complexity is a version of Kolmogorov complexity, where a string has low complexity if it has a short description from which its bits are efficiently computable. We will not explore consequences for this notion in this work, but we expect that some of our results can be transferred to the problem of whether a string has high KT-complexity using standard observations about the relationship between this problem and MCSP.

\section{Open Problems and Further Research Directions}\label{s:open_problems}

We describe here a few directions and problems that we find particularly interesting, and that deserve further investigation.\\

\noindent \textbf{$\bullet$ Speedups in Computational Learning Theory.} One of our main conceptual contributions is the discovery of a surprising speedup phenomenon in learning under the uniform distribution using membership queries (Lemma \ref{l:learningspeedup}). Naturally, it would be relevant to understand which learning models admit similar speedups. In particular, is there an analogous result for learning under the uniform distribution using random examples? An orthogonal question is to weaken the assumptions on concept classes for which learning speedups hold.\\

\noindent \textbf{$\bullet$ Applications in Machine Learning.} Is it possible to use part of the machinery behind the proof of the Speedup Lemma (Lemma \ref{l:learningspeedup}) to obtain faster algorithms in practice? Notice that speedups are available for classes containing a constant number of layers of threshold gates, as $\mathsf{TC}^0$ is a typical circuit class according to our definition. Since these circuits can be seen as discrete analogues of neural networks, which have proven quite successful in several contexts of practical relevance, we believe that it is worth exploring these implications.\\

\noindent \textbf{$\bullet$ Non-Uniform Circuit Lower Bounds from Learning Algorithms.} As discussed in \citep{DBLP:conf/stoc/Williams14}, strong lower bounds are open even for seemingly weak classes such as $\mathsf{MOD}_2 \circ \mathsf{AND} \circ \mathsf{THR}$ and $\mathsf{AND} \circ \mathsf{OR} \circ \mathsf{MAJ}$ circuits. We would like to know if the learning approach to non-uniform lower bounds (Theorem \ref{t:learning_lbs_general}) can lead to new lower bounds against such heavily constrained circuits. More ambitiously, it would be extremely interesting to understand the learnability of $\mathsf{ACC}^0$, given that the existence of a nontrivial algorithm for large enough circuits implies $\mathsf{REXP} \nsubseteq \mathsf{ACC}^0$ (Theorem \ref{t:learning_lbs_ACC}).\\

\noindent \textbf{$\bullet$ The Frontier of Natural Proofs.} Is there a natural property against $\mathsf{ACC}^0$? Williams \citep{DBLP:journals/jacm/Williams14} designed a non-trivial satisfiability algorithm for sub-exponential size $\mathsf{ACC}^0$ circuits, which implies in particular that $\mathsf{NEXP} \nsubseteq \mathsf{ACC}^0$. On the other hand, Corollary \ref{c:williams_natural_lb} shows that the existence of a natural property against such circuits implies the stronger lower bound $\mathsf{ZPEXP} \nsubseteq \mathsf{ACC}^0$.\\

\noindent \textbf{$\bullet$ Connections between Learning, Proofs, Satisfiability, and Derandomization.} Together with previous work (e.g. \citep{DBLP:journals/siamcomp/Williams13, DBLP:journals/jacm/Williams14, DBLP:journals/cc/SanthanamW14, DBLP:conf/icalp/JahanjouMV15}), it follows that non-trivial learning, non-trivial proofs of tautologies (in particular, nontrivial satisfiability algorithms), and non-trivial derandomization algorithms all imply (randomized or nondeterministic) exponential time circuit lower bounds. These are distinct algorithmic frameworks, and the argument in each case is based on a different set of techniques. Is there a more general theory that is able to explain and to strengthen these connections? We view Corollary \ref{c:combination} as a very preliminary result indicating that a more general theory along these lines might be possible.\\

\noindent \textbf{$\bullet$ Unconditional Nontrivial Zero-Error Simulation of $\mathsf{REXP}$.} Establish unconditionally that $\mathsf{REXP} \subseteq \mathtt{i.o.}\mathsf{ZPESUBEXP}$. We view this result as an important step towards the ambitious goal of unconditionally derandomizing probabilistic computations, and suspect that it might be within the reach of current techniques. In particular, this would follow if one can improve Lemma \ref{l:EasyWitness}, which unconditionally establishes that either $\mathsf{BPP} \subseteq \mathsf{ZPQP}$ or $\mathsf{ZPEXP} \subseteq \mathtt{i.o.}\mathsf{ESUBEXP}$, to a result of the same form but with $\mathsf{REXP}$ in place of $\mathsf{ZPEXP}$.\\

\noindent \textbf{$\bullet$ Learning Algorithms vs. Pseudorandom Functions.} The results from Section \ref{s:learning_prfs} establish an equivalence between learning algorithms and the lack of pseudorandom functions in a typical circuit class, in the non-uniform exponential time regime. It would be interesting to further investigate this dichotomy, and to understand whether a more uniform equivalence can be established.\\

\noindent \textbf{$\bullet$ Hardness of the Minimum Circuit Size Problem.} Show that $\mathsf{MCSP} \notin \mathsf{AC}^0[p]$. We have established that if $\mathsf{MCSP} \in \mathsf{TC}^0$ then $\mathsf{NC}^1 \subseteq \mathsf{TC}^0$. Prove that if $\mathsf{MCSP} \in \mathsf{TC}^0$ then $\mathsf{Circuit}[\mathsf{poly}] \subseteq \mathsf{TC}^0$.\\

\noindent \rule{\textwidth}{0.5pt}\vspace{0.1cm}
\noindent \textbf{Acknowledgements.} We thank Eric Allender, Marco Carmosino, Russell Impagliazzo, Valentine Kabanets, Antonina Kolokolova, Jan Kraj\'{\i}\v{c}ek, Tal Malkin, J{\'a}n Pich, Rocco Servedio and Ryan Williams for discussions. We also thank Ruiwen Chen for several conversations during early stages of this work. 

The first author received support from CNPq grant 200252/2015-1. The second author was supported by the European Research Council under the European Union's Seventh Framework Programme (FP7/2007-2013)/ERC Grant No. 615075. Part of this work was done during a visit of the first author to Oxford supported by the second author's ERC grant.

\bibliographystyle{alpha}	
\bibliography{refs}	

\end{document}